\documentclass[a4paper,onecolumn,11pt]{quantumarticle}
\pdfoutput=1
\usepackage[utf8]{inputenc}
\usepackage[english]{babel}
\usepackage[T1]{fontenc}
\usepackage{hyperref}
\usepackage{mdframed}
\usepackage{tikz}
\usepackage{lipsum}
\usepackage{amsmath,amsfonts,amssymb} % Ensure mathematical symbols are supported
\usepackage{tocloft}
\usepackage{algorithmic}
\usepackage{etoolbox}
\usepackage{algorithm}
\usepackage{array}
\usepackage{textcomp}
\usepackage{stfloats}
\usepackage{url}
\usepackage{verbatim}
\usepackage{graphicx}
\usepackage{cite}
\usepackage{amsthm}%
\usepackage{mathrsfs}%
\usepackage{xcolor}%
\usepackage{textcomp}%
\usepackage{manyfoot}
\usepackage{booktabs}%
\usepackage{syntax}
\usepackage{fancybox}
\usepackage{adjustbox}
\usepackage{physics}
\usepackage{soul}
\usepackage{listings}
\usepackage{multirow}
\usepackage{multicol}
\usepackage{comment}
\usepackage{nameref}
\usepackage{mathtools}
\usepackage{tikz-cd}
\usepackage{float}
\usepackage{titlesec}
\newtheorem{theorem}{Theorem}% 
\newtheorem{proposition}[theorem]{Proposition}% 
\newtheorem{remark}{Remark}%
\newtheorem{fact}{Fact}

\newtheorem{definition}{Definition}%
\newtheorem{corollary}{Corollary}[theorem]
\newtheorem{lemma}[theorem]{Lemma}

\newcommand{\QOTPEnc}{\mathsf{QOTPEnc}}

\newcommand{\Uperm}{U_{\sigma_\ell}}
\newcommand{\Deph}{\Gamma_{\lambda}} % dephasing/pinching channel

\newcommand{\secpar}{\eta}

\makeatletter
\let\if@restonecol\iftrue
\let\if@twocolumn\iffalse
\makeatother

\begin{document}

\title{Computational Quantum Anamorphic Encryption and Quantum Anamorphic Secret-Sharing}

\author{Sayantan Ganguly}
\email{sayantan.ganguly.90@tcgcrest.org}
\affiliation{Institute for Advancing Intelligence (IAI), TCG CREST, EP \& GP Block, Sector - V, Salt Lake, Kolkata-700091, West Bengal, India.}
\orcid{0000-0002-2445-2701}

\affiliation{Ramakrishna Mission Vivekananda Educational and Research Institute (RKMVERI), Belur, West Bengal, India.}
\orcid{0000-0003-0290-4698}
%%\thanks{}

\author{Shion Samadder Chaudhury}
\email{chaudhury.shion@gmail.com}
\affiliation{Institute for Advancing Intelligence (IAI), TCG CREST, EP \& GP Block, Sector - V, Salt Lake, Kolkata-700091, West Bengal, India.}
\orcid{0000-0002-2445-2701}
\affiliation{Faculty of Mathematical and Information Sciences, Academy of Scientific and Innovative Research(AcSIR), Gazhiabad, India}

\maketitle

\begin{abstract}
  The concept of anamorphic encryption, first formally introduced by Persiano et al. in their influential 2022 paper titled ``Anamorphic Encryption: Private Communication Against a Dictator,'' enables embedding covert messages within ciphertexts. One of the key distinctions between a ciphertext embedding a covert message and an original ciphertext, compared to an anamorphic ciphertext, lies in the indistinguishability between the original ciphertext and the anamorphic ciphertext. This encryption procedure has been defined based on a public-key cryptosystem. Initially, we present a quantum analogue of the classical anamorphic encryption definition that is based on public-key encryption. Additionally, we introduce a definition of quantum anamorphic encryption that relies on symmetric key encryption. Furthermore, we provide a detailed generalized construction of quantum anamorphic symmetric key encryption within a general framework, which involves taking any two quantum density matrices of any different dimensions and constructing a single quantum density matrix, which is the quantum anamorphic ciphertext containing ciphertexts of both of them. Subsequently, we introduce a definition of computational anamorphic secret-sharing and extend the work of \c{C}akan et al. on computational quantum secret-sharing to computational quantum anamorphic secret-sharing, specifically addressing scenarios with multiple messages, multiple keys, and a single share function. This proposed secret-sharing scheme demonstrates impeccable security measures against quantum adversaries.
\end{abstract}

%%\section*{Contents}
\tableofcontents

\section{Introduction}\label{sec1}
With the advent of quantum computing, the field of cryptography is experiencing an unprecedented paradigm shift \cite{Shor1997Quantum, Grover1996Fast}. Quantum communication systems have several crucial advantages compared to classical cryptographic and communicational methods \cite{Bennett1984Quantum}. These advantages stem from the unique properties of quantum mechanics, enabling new forms of security and computational capabilities that classical systems cannot provide \cite{Nielsen2010Quantum}. Quantum communication offers several advantages over classical communication systems, mainly in terms of its security and efficiency \cite{Gisin2002Quantum}. A key quantum communication advantage is unconditional security, which is especially suitable for QKD applications that cannot be guaranteed by any classical communication system \cite{Lo1999Unconditional}. The no-cloning theorem states that under quantum communication, it is not possible to copy an unknown quantum state exactly \cite{Wootters1982Single}. Quantum communication is a very important thread in the framework for quantum computing systems \cite{Kimble2008Quantum}. With quantum computation, quantum communication allows for exponential gains in processing large datasets or complex algorithms \cite{Ladd2010Quantum}. Quantum communication is expected to be the backbone of the quantum networks spreading the quantum states among distributed quantum computers and sensors to perform complex tasks like distributed quantum computing, quantum-enhanced sensing, and secure global QKD \cite{Wehner2018Quantum}.

\medskip

The word ``anamorphic" characterizes a distorted or deformed projection or drawing; however, from a given point of view or technique, it seems or appears normal. Anamorphic encryption is a cryptographic encryption technique, a notion invented by Persiano et. al. \cite{persiano2022anamorphic}. According to Persiano et al. \cite{persiano2022anamorphic}, its success depends on two often-taken-for-granted assumptions: sender freedom and receiver privacy. The first assumes senders can choose the message, and the second assumes the receiver's secret key is uncompromised. While these assumptions are natural in most cases, they may be at risk in nations where law enforcement can force users to hand over their decryption keys. In dictatorships, citizens may only communicate regime-approved content, diminishing the sender's freedom \cite{catalano2024anamorphic}. Persiano et al. \cite{persiano2022anamorphic} added anamorphic encryption to these challenging cases. In \cite{persiano2022anamorphic}, two primitives are proposed depending on which assumption is unreliable: Sender anamorphic encryption handles circumstances where the sender's freedom assumption fails, whereas receiver anamorphic encryption addresses compromised private keys \cite{catalano2024anamorphic}. In an anamorphic encryption scheme, Alice can send a message to Bob, an original message and a covert message under dictatorial supervision in such a way that the original ciphertexts and the anamorphic ciphertexts are indistinguishable from the dictator \cite{persiano2022anamorphic}. One thing that makes anamorphic encryption stand out is that, unlike other steganographic methods, it can hide communication in a message that looks like any other encrypted message. It keeps the very existence of the hidden anamorphic message undetectable unless performed with the usage of the correct decryption key or method. The anamorphic message uses a special anamorphic key or protocol to decrypt it, while the original message can be decrypted through another key distinct from the anamorphic key. The existing works on classical encryption are mainly based on public-key encryption \cite{persiano2022anamorphic,catalano2024anamorphic,banfi2024anamorphic}. 

\medskip

Here, we shall make a distinction between steganography and anamorphic encryption \cite{Anderson1998Limits}. Steganography is the technique of masking information within some other unsuspecting data in such a way that even the existence of the hidden data may not be detectable \cite{Johnson2001Survey}. A steganographic technique typically hides the original message in a \textit{non-encrypted} cover message through images, audio, or text using subtle modifications so that it would not raise suspicions in an observer that something is being hidden \cite{Fridrich2009Steganography}. The challenge is to ensure that the modifications made to the cover object do not raise suspicion \cite{Petitcolas1999Information}. This requires that the alterations are small enough to be undetectable \cite{Katzenbeisser2000Principles}. If the cover object is examined closely, statistical analysis (e.g., using \textit{steganalysis} techniques) might reveal the existence of hidden data \cite{Chandramouli2002Analysis}.

\medskip

The main goal of anamorphic encryption is to construct such an encryption scheme in such a way that the original ciphertext and the anamorphic ciphertext are indistinguishable from the observation of the adversary. The crucial point here is that the ciphertext looks like normal encryption, and an observer is unaware that a second covert message exists. The encryption process ensures that the ciphertext can be decrypted in two ways, depending on the key used. The challenge is to ensure that the ciphertext is \textit{indistinguishable} from a normal ciphertext, even to an adversary who suspects that hidden messages might exist. The observer cannot tell that the ciphertext contains more than one message without the anamorphic key.

\medskip
In quantum secret-sharing(\(\mathsf{QSS}\)), a dealer distributes a secret, which is a quantum state, among the set of players. In the paper \cite{ccakan2023computational}, \c{C}akan et al. initiated the \textit{computationally secure} \(\mathsf{QSS}\) and showed that similar to the classical secret-sharing scheme, computational assumptions can significantly help in building \(\mathsf{QSS}\) schemes. In that paper, they have constructed \textit{polynomial-time} computationally secure \(\mathsf{QSS}\) schemes under standard hardness assumptions for a wide class of access structures, which also includes many access structures that necessarily require exponential share size. They have also studied the class of access structures that can be implemented \textit{efficiently} when the \(\mathsf{QSS}\) scheme has access to a given number of copies of the secret, including all functions in \(\mathsf{P}\) and \(\mathsf{NP}\) \cite{ccakan2023computational}. Here \textit{efficient} means both the share and the reconstruction function can be computed in polynomial time \cite{ccakan2023computational}. 
\medskip

The no-cloning theorem prohibits the exact duplication of unknown quantum states \cite{wootters1982nocloning, girling2024nocloning}. This implies that fundamental approaches, such as distributing identical components to multiple players, cannot be applied in quantum contexts, as the precise replication of shared states is infeasible \cite{gottesman1999theory, yu2022qss}. Furthermore, no quantum secret-sharing techniques currently exist to implement the OR function, which underscores the difficulty of directly adapting classical methodologies for quantum applications \cite{cakan2023unclonable}. Consequently, essential classical outcomes, including the use of logical functions like OR, face intrinsic incompatibilities with quantum mechanics due to these constraints \cite{li2024threshold}.

\medskip
In this paper, we have addressed the following specific questions:

\begin{itemize}
    \item \textbf{Question 1.} Is it possible to define an analogous quantum anamorphic encryption scheme where both the original and covert messages are represented as general quantum states or quantum density matrices of any finite dimension, while accounting for the presence of quantum adversaries?

    \medskip
    \item  \textbf{Question 2.} If such a quantum anamorphic encryption scheme is feasible, how can it be constructed to securely encode general quantum density matrices?
    
    \medskip
    \item \textbf{Question 3.} While the classical anamorphic encryption framework is built upon public-key encryption (\(\mathsf{PKE}\)), is it possible to develop a quantum anamorphic symmetric key encryption scheme that is secure against quantum adversaries? What are the inherent challenges in such a construction?

    \medskip
    \item \textbf{Question 4.} How can an anamorphic secret-sharing scheme be formally defined? What are the associated challenges, and what are the potential attacks that need to be considered?

    \medskip
    \item \textbf{Question 5.} How can anamorphic encryption, and consequently anamorphic secret-sharing, be designed to ensure that even if an adversary suspects the presence of a covert message, they remain incapable of decrypting it from the ciphertext?

\end{itemize}

\medskip

\noindent \(\bullet\) \textbf{Quantum advantage:}
Quantum anamorphic encryption can exploit quantum superposition and entanglement to hide the existence of the inner message more effectively. For instance, quantum states can encode multiple messages simultaneously, and any attempt by an adversary to measure or intercept the message would disturb the quantum state, revealing the presence of an eavesdropper. Also, because of the no-cloning theorem, it is impossible to create an exact copy of an unknown quantum state. This property ensures that an adversary cannot clone the quantum-encrypted message to analyze it without detection. In high-stakes environments (e.g., political dissidents, and military operations), quantum anamorphic encryption provides a higher level of assurance that the hidden message cannot be detected or decrypted by a coercer, even if they have access to quantum computing resources. Quantum entanglement can be used to distribute shares of a secret in a way that any unauthorized attempt to access the shares would disrupt the entangled state, alerting the participants to the breach.  Quantum systems can encode information in higher-dimensional Hilbert spaces, allowing for more efficient and secure sharing of secrets compared to classical systems. Quantum symmetric key anamorphic encryption is practical for scenarios where high-speed encryption and decryption are required, such as in real-time communication systems. The use of symmetric keys reduces computational overhead compared to public-key systems, while the quantum components ensure security against quantum attacks.

\medskip
\noindent \textbf{Applications:} There are many real-life applications of anamorphic encryption, for example, in diplomatic or military communications, for whistleblowing, or in activism. Whether it is an autocratic regime or an environment that suppresses free speech, a journalist, activist, or whistleblower will need to air out sensitive information without the hawk-eyed views of governmental censors or repressive regimes. A whistleblower operating within one of many corrupt governmental agencies decides to leak classified documents to a journalist. In international diplomacy and military operations, sensitive communications must be kept from adversaries or other foreign intelligence agencies. Anamorphic encryption allows the diplomat or military person to send secret messages without giving away the fact that they are transmitting sensitive information. In this domain, anamorphic encryption can give added security by embedding the covert information and making it accessible only through the proper anamorphic key. Sensitive data is often stored on a third-party cloud server in cloud storage. Even if the data are encrypted, it may be in a form that the service provider can detect the existence of sensitive data, thus raising concerns about the privacy of stored information. Anamorphic encryption is one way by which data can be stored hidden without being detected by a cloud provider. A company might store business reports on a cloud server encrypted with the original key but allow the service provider to audit or perform checks. Yet those very same files can also have classified financial data or intellectual property embedded in them with the anamorphic key, whose access is limited to only the authorized personnel of the company. 

\subsection{Related Works}
In classical cryptography, particularly in public-key encryption, significant progress has been made in the study of anamorphic encryption in recent years. Notable contributions include works by Banfi et al., Catalano et al., and Kutylowski \cite{banfi2024anamorphic,catalano2024anamorphic,catalano2024generic,catalano2024limits,kutylowski2023self}, building upon its original introduction by Persiano et al. in \cite{persiano2022anamorphic}. More recently, Jaeger and Stracovsky \cite{jaeger2025dictators} proposed additional conditions to refine the definition of anamorphic encryption. Furthermore, Wang et al. \cite{wang2023sender} presented a robust and generic construction, reformulating the concept of sender-anamorphic encryption.  

In the quantum setting, there has been growing interest in quantum public-key encryption, with foundational work by Okamoto et al. \cite{okamoto2000quantum}. More recently, Barooti et al. \cite{barooti2023public} introduced a quantum public-key encryption scheme utilizing quantum public keys, expanding the scope of secure quantum communication.  

Secret-sharing plays a fundamental role in classical cryptography. For a comprehensive survey on classical secret-sharing schemes, we refer the reader to Beimel’s article \cite{beimel2011secret}. Efficiency is a crucial aspect of secret-sharing, ensuring that both the sharing and reconstruction processes are computationally feasible, i.e., executable in polynomial time. The seminal works of Shamir \cite{shamir1979share} and Blakley \cite{blakley1979safeguarding} introduced efficient threshold secret-sharing schemes for \( t \)-out-of-\( n \) access structures. For all functions in \(\mathsf{monotone}\) \(\mathsf{P}\), Yao \cite{Yao89} and Vinod et al. \cite{vinod2003power} developed efficient computational secret-sharing schemes. Additionally, Komargodski, Naor, and Yogev \cite{KNY14} constructed efficient computational secret-sharing schemes capable of realizing all functions in \(\mathsf{mNP}\).  

In quantum cryptography, secret-sharing naturally extends to the sharing of quantum states. However, not all monotone functions permit quantum secret-sharing. The problem of quantum secret-sharing for specific classes of monotone functions has been explored in \cite{HBB99,KKI99,CGL99}. Gottesman \cite{Got00} and Smith \cite{Smi00} provided constructions for quantum secret-sharing schemes realizing all allowable monotone functions. Imai et al. \cite{imai2003quantum} proposed a general model for quantum secret-sharing, while Smith \cite{smith2000quantum} constructed quantum secret-sharing schemes for general access structures. Furthermore, in \cite{Smi00}, Smith developed quantum secret-sharing schemes for monotone functions \( f \), ensuring that the total share size corresponds to the size of the smallest monotone span program (MSP) computing \( f \), thus extending a classical result by Karchmer and Wigderson to the quantum domain. More recently, Çakan et al. \cite{ccakan2023computational} constructed and described an efficient computational model for quantum secret-sharing, further advancing this field.

\subsection{Our Contribution}

In this paper, we introduce and rigorously analyze quantum analogues of anamorphic encryption and secret-sharing in quantum communications.

\noindent \textbf{Quantum Anamorphic Encryption:}
 We propose a quantum analogous definition of the classical anamorphic encryption scheme definition in the quantum public-key encryption setting [Subsection \ref{sec:qae}, Definition \ref{def:qape}] as well as in the quantum symmetric-key encryption setup [Subsection \ref{sec:qae}, Definition \ref{def:qase}]. We have constructed a general quantum anamorphic symmetric-key encryption scheme [Subsection \ref{subsec:main}] and rigorously proved the computational indistinguishability of the original ciphertext and the anamorphic ciphertext in the Theorem \ref{thm:ind}. Our construction ensures that the anamorphic ciphertext $M_f^{(1)}$, which contains both the original ciphertext and a covert ciphertext, remains indistinguishable from the original ciphertext $M_f^{(0)}$ to an adversary or to the dictator in our model. We formally established that both $M_f^{(1)}$ and $M_f^{(0)}$ are valid quantum density matrices, ensuring the integrity of our construction.

\medskip
\noindent \textbf{Indistinguishability and Fidelity Analysis:}
 We demonstrate in Theorem \ref{thm:main} the crucial role of the multiplicative factor $\eta$ in maintaining the positive semi-definiteness of the quantum density matrix $M_f^{(1)}$, showing that it is $\frac{1}{\eta}$-indistinguishable from $M_f^{(0)}$. We derived a lower bound on $\eta$ in Corollary \ref{cor:main}, further deriving the sufficient condition to ensure that the anamorphic ciphertext \(M_{f}^{(1)}\) is a valid quantum density matrix. We analyzed the expected states and their computational indistinguishability (Theorem \ref{thm:exp}). Utilizing fidelity as a measure of closeness between quantum states, we establish in Theorem \ref{thm:fid} that the fidelity between the original quantum ciphertext and the anamorphic quantum ciphertext is at least $\Big(1-\frac{1}{\eta}\Big)$, indicating a high level of similarity. Also, we have analyzed the von Neumann entropy, mutual information and relative entropy for our construction and for some particular cases in [Section \ref{discussion}].

\medskip
\noindent \textbf{Quantum Anamorphic Secret-Sharing:}
We propose a new definition of anamorphic secret-sharing (Definition \ref{def:ass}) and construct a quantum anamorphic secret-sharing scheme (Theorem \ref{thm:ass}) and rigorously prove the correctness of our scheme (Theorem \ref{thm:corr}) and establish perfect privacy (Theorem \ref{thm:priv}). Our scheme generalizes the work of \c{C}akan et al. \cite{ccakan2023computational} to support multiple key distributions.

\medskip
\noindent \textbf{Security Analysis:}
We analyze two potential attacks in Section \ref{sec:attack} and propose countermeasures to prevent adversarial reconstruction of the covert secret. We extend Ogata et al.'s definition of cheating probability \cite{ogata2004new} to \textit{partial cheating probability} (Definition \ref{def:parchit}) within the anamorphic secret-sharing context. We demonstrate that the adversary or dictator can be effectively prevented with high partial cheating probability (Equation \ref{eq:parchit}).

\subsection{Paper Organization} The paper is distributed among the following sections: In section \ref{prelim}, we outline the preliminary concepts and notations required for this work. In the section \ref{sec:cae} we have described the classical anamorphic encryption. We proposed a definition of quantum anamorphic public-key encryption and quantum anamorphic symmetric key encryption in \ref{sec:qae}. We present our main construction of quantum anamorphic encryption in the symmetric-key encryption setup and computational analysis in section \ref{technical}. The study on quantum anamorphic secret-sharing schemes is done in section \ref{sec:ass} and the compiler is presented in section \ref{compiler}. In section \ref{discussion}, we analyze the qubit requirements, analyzed the von Neumann entropy, mutual information, relative entropy for our construction and also we have discussed some possible attacks and how to prevent it. The paper is concluded in section \ref{conclude}.

\section{Preliminaries}\label{prelim}
\subsection{Notations}
In this paper, we have denoted non-empty sets by uppercase letters. We denote \( [n]= \{1, \ldots, n\} \). Let \(v\) be a vector and \(S\) be a non-empty set. We denote the symmetric group of \(n\) elements by \(\operatorname{Sym}(n)\). We denote \(v^{P}\) to indicate the vector \((v_{i})_{i \in P}\). Let \(\{\mathcal{S}_{i}\}_{i \in [n]}\) be a family of sets, and for \(P \subseteq [n]\), we denote \(\mathcal{S}_{P}:=\prod_{i\in P \subseteq [n]}\mathcal{S}_{i}\). For convenience, we have denoted the set of \(n\) players by \([n]\) and \(\{P_{1},\ldots,P_{n}\}\) interchangeably at appropriate places. We denote \(R {\gets} \mathcal{R}\) to denote that \(R\) is uniformly distributed on \(\mathcal{R}\). We used \(\rho\) to denote a density matrix acting on \(\mathcal{H}\). It will be clear from the context whether we mean a vector \(\ket{\psi}\) in a Hilbert space \(\mathcal{H}\) representing a pure state or a density matrix \(\rho\) acting on \(\mathcal{H}\), representing a mixed state. The \textit{trace norm} \( \|\cdot\|_1 \) for any operator \( X \) is defined by
\(\|X\|_1 = \text{Tr}\left( \sqrt{X^\dagger X} \right),\)
with \( X^\dagger \) denoting the Hermitian conjugate (or adjoint) of \( X \).  The \textit{operator norm} of \( X \), denoted \( \|X\| \) or \( \|X\|_{\infty} \), is given by \(\|X\|= \sup_{|\psi\rangle \in \mathcal{H}, ||\psi||=1} \| X|\psi\rangle \|_2\), where $\| X|\psi\rangle \|_2$ is the Euclidean norm of the resulting vector. For the important special case of a \textit{Hermitian positive semi-definite (PSD)} operator $X \succeq 0$, that is $X=X^\dagger$ and $\langle \psi | X | \psi \rangle \ge 0$ for all $|\psi\rangle$, all eigenvalues are non-negative. In this case, the operator norm, the numerical radius, and the largest eigenvalue coincide:
\(\|X\| = \sup_{|\psi\rangle \in \mathcal{H}, \|\psi\| = 1} \langle \psi | X | \psi \rangle,\)
where the supremum is taken over all unit vectors \( |\psi\rangle \in \mathcal{H} \). We have used both the notations \(\|X\|\) and \(\|X\|_{\infty}\) interchangebly. For $X\in\mathcal{L}(\mathbb{C}^d)$, the Frobenius norm is defined by \(\|X\|_2 \ :=\ \big(\mathrm{Tr}(X^\dagger X)\big)^{1/2}\). If \( x \) is a vector in \( \mathbb{R}^n \), its Euclidean norm is denoted \( \|x\|_2 \) and is defined as
\(\|x\|_2 = \sqrt{x_1^2 + x_2^2 + \cdots + x_n^2}.\) We denote \(\log x\) as the base-2 logarithm of \(x\), unless explicitly specified otherwise. We have used the notation \(P\) to denote different quantities in different scenarios. \(P\) may denote the set of players, the Pauli operators or the distributions. We have taken care to mention the context where the notation is used. Given an adversary or distinguisher $\mathcal{D}$ and an oracle channel~$\mathcal{O}$, the notation $\mathcal{D}^{\mathcal{O}}$ indicates that $\mathcal{D}$ can make adaptive, polynomially many queries to~$\mathcal{O}$, with fully quantum inputs and outputs. The identity operator defined on a finite-dimensional Hilbert space \(\mathcal{H}_{M}\) is denoted as \(I_{M}\) and when \(\dim_{\mathbb{C}}(\mathcal{H}_{M})=2^{d_{1}}\), we also denote \(I_{M}\) as \(I_{2^{d_{1}}}\).

\subsection{Quantum information theory}
In this section, we review a few basic definitions and a few theoretical concepts of quantum information that will be used in our paper. We refer the following references \cite{nielsen2001quantum,wilde2017quantum,choi1975completely} to the reader.

\medskip

Let \( \mathcal{H} \) and \( \mathcal{K} \) be finite-dimensional complex Hilbert spaces associated with the input and output quantum systems, respectively. The space of linear operators on \( \mathcal{H} \) is denoted \( \mathcal{L}(\mathcal{H}) \). The state of the quantum system is described by a density operator \( \rho \in \mathcal{L}(\mathcal{H}) \), which satisfies:
\begin{enumerate}
    \item \textbf{Hermitian}: \(\rho^{\dagger}=\rho\)
    \item \textit{\textbf{Positive semi-definiteness:}} \( \rho \succeq 0 \) 
    \item \textit{\textbf{Unit trace:}} \( \operatorname{Tr}(\rho) = 1 \).
\end{enumerate}

We denote the set of all density matrices defined on \(\mathcal{H}\) by 
\[
\mathcal{D}(\mathcal{H}) = \{ \rho \in \mathcal{L}(\mathcal{H}) \mid \rho^{\dagger}=\rho,\rho \succeq 0, \operatorname{Tr}(\rho) = 1 \}.
\]
\medskip

A quantum channel maps quantum states (represented as density operators on a Hilbert space) to other quantum states, accounting for potential noise and decoherence effects. It is a \textit{completely positive, trace-preserving (CPTP) linear map} on the space of density operators.

\begin{definition}(Quantum channel \cite{wilde2017quantum,choi1975completely,nielsen2001quantum})
    A \textit{quantum channel} is a linear map \( \Phi: \mathcal{L}(\mathcal{H}) \longrightarrow \mathcal{L}(\mathcal{K}) \) satisfying the following properties:
\begin{enumerate}
    \item \textbf{Complete Positivity}: For any \( n \in \mathbb{N} \), the map \(( \Phi \otimes \operatorname{Id}_n) : \mathcal{L}(\mathcal{H} \otimes \mathbb{C}^n) \longrightarrow \mathcal{L}(\mathcal{K} \otimes \mathbb{C}^n) \) is positive, where \( \operatorname{Id}_n \) is the identity map on \( \mathbb{C}^n \). That is, for all \( X \in \mathcal{L}(\mathcal{H} \otimes \mathbb{C}^n) \) with \( X \succeq 0 \), we have \( (\Phi \otimes \operatorname{Id}_n)(X) \succeq 0 \).
    
    \item \textbf{Trace Preservation}: For all \( \rho \in \mathcal{L}(\mathcal{H}) \), \( \operatorname{Tr}(\Phi(\rho)) = \operatorname{Tr}(\rho) \).
\end{enumerate}

Therefore, \( \Phi \) maps density operators on \( \mathcal{H} \) to density operators on \( \mathcal{K} \), preserving the physical validity of quantum states.
\end{definition}

\medskip

We will consider finite-dimensional quantum systems with $m$ degrees of freedom, represented by the algebra of $m \times m$ matrices over \(\mathbb{C}\), referred to as $\mathbb{M}_m$. The state of a quantum system $X$ is characterized by its density matrix $\rho_X \in \mathbb{M}_m$.

\medskip

\begin{definition}(von Neumann entropy \cite{schumacher1995quantum,nielsen2001quantum,wilde2017quantum,wehrl1978general})
The von Neumann entropy of a quantum system $X$ with a density matrix $\rho_X \in \mathbb{M}_m$ is defined as
\[
\mathsf{S}(X) = -\Tr(\rho_X \log \rho_X) = -\sum_{1 \leq j \leq m} \lambda_j \log \lambda_j,
\]
where $\lambda_1, \ldots, \lambda_m$ are the eigenvalues of $\rho_X$. One way to look at the quantum entropy is as the average amount of qubits needed to describe a system $X$ \cite{schumacher1995quantum}.
\end{definition}

The maximum value for \( \mathsf{S}(A) \) is given by:
\[
\mathsf{S}(X) \leq \log \dim (\mathcal{H}_{X})= \log m,
\]
where \(\mathcal{H}_{X}\) is the Hilbert space associated with the quantum system \(X\).

\medskip

Let $XY$ be a bipartite quantum system. Let $\rho_{XY}$ be the density matrix associated to \(XY\) resides on the Hilbert space $\mathcal{H}_{XY} = \mathcal{H}_X \otimes \mathcal{H}_Y$. The subsystems $X$ and $Y$ will be represented by the partial traces $\rho_X = \Tr_{Y} (\rho_{XY})$ and $\rho_Y = \Tr_{X} (\rho_{XY})$. The von Neumann entropy of a quantum system $X$ conditional on another quantum system $Y$ is defined as (see \cite{nielsen2001quantum,imai2003quantum}):
\[
\mathsf{S}(X|Y) = \mathsf{S}(XY) - \mathsf{S}(Y),
\]
where $\mathsf{S}(XY) = -\Tr(\rho_{XY} \log \rho_{XY})$ is the \textit{joint entropy of} \(XY\) and $\mathsf{S}(Y) = -\Tr(\rho_Y \log \rho_Y)$.

\medskip
The joint entropy of two quantum systems satisfies two important properties:
\begin{itemize}
\item \textbf{Subadditivity:} \[
\mathsf{S}(XY) \leq \mathsf{S}(X) + \mathsf{S}(Y) ,
\]

\item \textbf{Araki-Lieb inequality:} \[
\mid \mathsf{S}(X) - \mathsf{S}(Y) \mid \leq \mathsf{S}(XY)
.\]

\end{itemize}

The mutual information between two quantum systems \( X\) and \( Y \) is defined as:
\[
I(X : Y) = \mathsf{S}(X) + \mathsf{S}(Y) - \mathsf{S}(XY).
\]

\medskip

Despite many analogies between quantum and classical entropies, they are fundamentally distinct. The conditional von Neumann entropy can be negative, but the conditional classical Shannon entropy is always non-negative.

\medskip

\begin{definition} (Statistical distance/Total variation distance \cite{levin2009markov, cover2006elements,billingsley1995probability,gibbs2002choosing}). The total variation distance between two random variables \( X\) and \(Y \) defined on the same sample space \( \mathcal{R} \) is defined as
\[
\Delta(X, Y) = \max_{\mathcal{G} \subseteq \mathcal{R}} \Big| \Pr[X \in \mathcal{G}] - \Pr[Y \in \mathcal{G}] \Big| = \frac{1}{2} \left( \sum_{a \in \mathcal{R}} | \Pr[X = a] - \Pr[Y = a] | \right).
\]
\end{definition}

The analogue of the total variation distance in the quantum setting is the trace distance.

\medskip

\begin{definition}(Adversarial Pseudometric \cite{ccakan2023computational})\label{def:qadv}
For a family \( \mathcal{F} \) of quantum circuits that produce a single-bit classical output, the \textit{distinguishing advantage} of \( \mathcal{F} \) between two quantum density matrices \( \rho \) and \( \sigma \) of appropriate dimensions is defined as:
\[
\mathsf{Adv}_\mathcal{F}(\rho, \sigma) = \max_{C \in \mathcal{F}} \Big| \Pr[C(\rho) = 1] - \Pr[C(\sigma) = 1] \Big|.
\]

The adversarial pseudometric quantifies the maximum probability difference with which any circuit in \( \mathcal{F} \) can distinguish between \( \rho \) and \( \sigma \).
\end{definition}

\medskip

\begin{definition}(Trace distance \cite{nielsen2001quantum,wilde2017quantum,fuchs1999cryptographic}) The trace distance between two density matrices \( \rho \) and \( \sigma \) with the same dimensions is defined as
\[
D(\rho, \sigma) = \frac{1}{2} \|\rho - \sigma\|_1.
\]
\end{definition}

\medskip

\begin{definition}(Fidelity \cite{nielsen2001quantum,wilde2017quantum,uhlmann1976transition,jozsa1994fidelity})
Let \(\rho\) and \(\sigma\) be two density matrices (quantum states) acting on the same Hilbert space \(\mathcal{H}\). Fidelity \(F(\rho, \sigma)\) is defined as:
\[
F(\rho, \sigma) = \operatorname{Tr}\left(\sqrt{\sqrt{\rho} \sigma \sqrt{\rho}}\right).
\]
\end{definition}

\medskip

\begin{lemma}(\cite{nielsen2001quantum}) \label{lemma:advbd}
For any family of quantum circuits $\mathcal{F}$ and two density matrices \(\rho\) and \(\sigma\) of the same dimension, the adversarial advantage is bounded by the trace distance
\[
\operatorname{Adv}_{\mathcal{F}}(\rho, \sigma) \leq D(\rho, \sigma).
\]
\end{lemma}

\begin{proof}
Consider any quantum circuit $C \in \mathcal{F}$ that performs a measurement on a quantum state and outputs a bit $b \in \{0, 1\}$. Any such measurement can be described by a positive operator-valued measure (POVM).

\medskip

We denote the measurement operators corresponding to output $1$ and $0$ as $E$ and $(I - E)$, respectively, where $0 \preceq E \preceq I$ and $I$ is the identity operator on $\mathcal{H}$. The probability that $C$ outputs $1$ when measuring $\rho$ is
\(
\Pr\left[ C(\rho) = 1 \right] = \operatorname{Tr}( E \rho).
\)

Similarly, the probability that $C$ outputs $1$ when measuring $\sigma$ is
\(
\Pr\left[ C(\sigma) = 1 \right] = \operatorname{Tr}( E \sigma).
\)
Then
\(
\left| \Pr\left[ C(\rho) = 1 \right] - \Pr\left[ C(\sigma) = 1 \right] \right| = \left| \operatorname{Tr}( E (\rho - \sigma) ) \right|.
\)
To find the maximum advantage over all circuits in $\mathcal{F}$, we consider the maximal value over all valid measurement operators $E$. Since $0 \preceq E \preceq I$, we have
\[
\operatorname{Adv}_{\mathcal{F}}(\rho, \sigma) \leq \max_{0 \preceq E \preceq I} \left| \operatorname{Tr}( E (\rho - \sigma) ) \right|.
\]
However, in quantum hypothesis testing, it is known that the maximum of $\left| \operatorname{Tr}( E (\rho - \sigma) ) \right|$ over all $0 \preceq E \preceq I$ is equal to the trace distance between $\rho$ and $\sigma$. Specifically, from the definition of the trace norm:
\[
\| \rho - \sigma \|_1 = \operatorname{Tr} \left[ | \rho - \sigma | \right ] = 2 \max_{0 \preceq E \preceq I} \left| \operatorname{Tr}( E (\rho - \sigma) ) \right|.
\]

Therefore,
\[
\max_{0 \preceq E \preceq I} \left| \operatorname{Tr}(E (\rho - \sigma)) \right| = \frac{1}{2} \| \rho - \sigma \|_1 = D(\rho, \sigma).
\]

and consequently,
\[
\operatorname{Adv}_{\mathcal{F}}(\rho, \sigma) \leq D(\rho, \sigma).
\]
\end{proof}

\medskip

\begin{lemma}(The Fuchs–van de Graaf inequalities \cite{nielsen2001quantum})
The Fuchs–van de Graaf inequalities are given by:
\[
1 - F(\rho, \sigma) \leq D(\rho, \sigma) \leq \sqrt{1 - F(\rho, \sigma)^2}.
\]
\end{lemma}

\medskip
Let $X\in \mathcal{L}(\mathcal{H})$ be a positive semidefinite operator on a finite-dimensional Hilbert space $\mathcal{H}$. The \emph{support} of $X$, denoted $\textit{supp}(X)$, is the subspace spanned by all eigenvectors of $X$ with nonzero eigenvalues. Equivalently, if $X$ has spectral decomposition,
\[
  X \;=\; \sum_{i}\,\lambda_i\,\ket{u_i}\bra{u_i}
  \quad(\lambda_i\ge0),
\]
then $\textit{supp}(X) := span\{\ket{u_i} : \lambda_i>0\}$.

\medskip
When we write $(X)^{-1}$ for $X \succeq 0$, we always mean the inverse restricted to $\textit{supp}(X)$; that is, we invert only the strictly positive eigenvalues, and vectors in the kernel of $X$ are understood to be outside the domain of $(X)^{-1}$. Such an inverse is sometimes called the \emph{generalized inverse on the support}.

\medskip

\noindent For a linear operator $Y \in \mathcal{L}(\mathcal{H})$ on a finite-dimensional complex Hilbert space $\mathcal{H}$, the \emph{spectral norm} $\|Y\|$ is defined as

\begin{align*}
\|Y\|
  \;=\;
  \sup_{\ket{\psi}\neq 0, 
        \ket{\psi} \in \mathbb{C}^{n}}
  \frac{\|Y\ket{\psi}\|_{2}}{\|\ket{\psi}\|_{2}},
\end{align*}
where $\|\ket{\psi}\|_{2}$ is the usual Euclidean norm of the vector $\ket{\psi}$. The spectral norm is equivalent to the largest singular value of \( A \). Formally, if the singular values of \( A \), denoted \( \sigma_1, \sigma_2, \ldots, \sigma_n \), are the square roots of the eigenvalues of the positive-semidefinite matrix \( A^*A \) (where \( A^* \) is the conjugate transpose of \( A \)), then:
\[
\|A\| = \sigma_{\text{max}} = \sqrt{\lambda_{\text{max}}(A^*A)},
\]
where \( \lambda_{\text{max}}(A^*A) \) is the largest eigenvalue of \( A^*A \).

\medskip

\begin{definition}(Quantum Relative Entropy \cite{nielsen2001quantum})
Given two quantum states \(\rho\) and \(\sigma\), where \(\rho\) and \(\sigma\) are density matrices, the \textit{quantum relative entropy} is defined as
\[
\mathsf{S}(\rho || \sigma) =
\begin{cases}
\text{Tr} \big(\rho (\log \rho - \log \sigma) \big), & \text{if } \text{supp}(\rho) \subseteq \text{supp}(\sigma), \\
+\infty, & \text{otherwise}.
\end{cases}
\]
\end{definition}

\medskip

Now we will go through some definitions and results, which will be useful to prove the Theorem \ref{thm:main}.

\medskip

\begin{definition}(Moore-Penrose Inverse \cite{horn2012matrix,BenIsraelGreville}) \label{def:MPinverse}
Let $A$ be a (real or complex) $m\times n$ matrix. The \emph{Moore-Penrose inverse} of $A$, denoted $A^+$, is the unique $n\times m$ matrix satisfying the following four equations called the \textit{Penrose equations:}

\medskip
\noindent i) \(A\,A^+\,A \;=\; A,\) \\
\noindent ii) \(A^+\,A\,A^+ \;=\; A^+,\) \\
\noindent iii) \((A\,A^+)^\dagger \;=\; A\,A^+,\) \\
\noindent iv) \((A^+\,A)^\dagger \;=\; A^+\,A,\) \\

\noindent where $X^\dagger$ denotes the conjugate transpose of $X$.
\end{definition}

\medskip
It is well known that for a strictly positive-definite, square, Hermitian matrix $M$, its Moore-Penrose inverse $M^{+}$ is precisely $M^{-1}$.

\begin{theorem}\label{thm:MPisInverse}(\cite{horn2012matrix,BenIsraelGreville,Penrose1955})
Let $M$ be an $n\times n$ strictly positive-definite matrix over $\mathbb{C}$ or $\mathbb{R}$. Then its Moore-Penrose inverse $M^{+}$ equals the usual inverse $M^{-1}$.
\end{theorem}

\begin{proof}
For the proof, see \cite{horn2012matrix}.
\end{proof}

\medskip

\noindent \(\textbf{Schur complement \cite{horn2012matrix}. }\) Consider a partitioned block matrix \(M \in \mathbb{C}^{n \times n}\), structured as:
\[
M = \begin{pmatrix}
B & C \\
C^\dagger & D
\end{pmatrix},
\]
where \(B \in \mathbb{C}^{k \times k}\), \(C \in \mathbb{C}^{k \times (n-k)}\), \(D \in \mathbb{C}^{(n-k) \times (n-k)}\).

If \(D\) is invertible, the \textit{Schur complement of \(D\) in \(M\)} is defined as:
\[
S_D = B - C D^{-1} C^\dagger.
\]

If \(B\) is invertible, the \textit{Schur complement of \(B\) in \(M\)} is defined as:
\[
S_B = D - C^\dagger B^{-1} C.
\]

\begin{theorem}(Schur complement condition \cite{horn2012matrix})\label{thm:schur}
A Hermitian matrix \(M \in \mathbb{C}^{n \times n}\) is positive semi-definite (\(M \succeq 0\)) if and only if
\begin{enumerate}
    \item The leading principal submatrix \(B\) is positive semi-definite: \(B \succeq 0\).
    \item The Schur complement of \(B\) in \(M\), defined as \(D - C^\dagger B^{+} C\), is positive semi-definite:
    \[
    D - C^\dagger B^{+} C \succeq 0,
    \]
    where \(B^{+}\) is the Moore-Penrose pseudo-inverse of \(B\).
\end{enumerate}
\end{theorem}

\begin{corollary}(Schur complement condition \cite{horn2012matrix})\label{cor:schur}
    If \(B\) is invertible, the condition simplifies:
\begin{enumerate}
    \item \textbf{Positive definiteness of \(B\):} \(B \succ 0\).
    \item \textbf{Positive semi-definiteness of the Schur complement:} \(D - C^\dagger B^{-1} C \succeq 0\).
\end{enumerate}
\end{corollary}

\medskip

\begin{definition}(Twirling map \cite{keyl2002fundamentals,chiribella2010on})
Let \( \mathcal{H} \) be a finite-dimensional Hilbert space, and let \( \rho \) be a density matrix on \( \mathcal{H} \), i.e.,
\begin{equation}
    \rho \in \mathcal{D}(\mathcal{H}) = \{ \rho \in \mathcal{L}(\mathcal{H}) \mid \rho^{\dagger}=\rho,\rho \succeq 0, \operatorname{Tr}(\rho) = 1 \}.
\end{equation}
Given a unitary representation \( U: G \longrightarrow \mathcal{U}(\mathcal{H}) \) of a compact group \( G \), the \emph{twirling map} is a quantum channel defined by:
\begin{equation}
    \mathcal{T}_G(\rho) = \int_G U(g) \rho U^\dagger(g) \, d\mu(g),
\end{equation}
where \( d\mu(g) \) is the \emph{Haar measure} on \( G \), \( U(g) \) is the unitary representation of \( g \in G \), and the integration is taken over all elements of \( G \). The twirling operation averages the state \( \rho \) over the unitary transformations of the group \( G \), producing a state that is invariant under the group \( G \) \cite{fulton2013representation, serre1977linear}.
\end{definition}

\medskip
 When \(G\) is the symmetric group \(\operatorname{Sym(n)}\) and the unitary operators are the permutation matrices, then we can define the twirling map over \(\operatorname{Sym}(n)\). We use this definition to compute the expectations later.

\begin{definition}(\cite{nielsen2001quantum})
    Let \( \mathcal{H} = (\mathbb{C}^d)^{\otimes n} \) be the Hilbert space of \( n \) quantum systems, each of dimension \( d \). The symmetric group \( \operatorname{Sym}(n) \) acts on \( \mathcal{H} \) via \emph{permutation operators} \( U_{\sigma_l} \), which permute the tensor factors according to \( \sigma_l \in \operatorname{Sym}(n) \).

\smallskip
\noindent The \emph{twirling operation over \( \sigma_l \in \operatorname{Sym}(n) \)} is defined as:
\begin{equation}
    \mathcal{T}_{ \operatorname{Sym}(n)}(\Phi) = \frac{1}{| \operatorname{Sym}(n) |} \sum_{\sigma_l \in \operatorname{Sym}(n)} U_{\sigma_{l}} \Phi U_{\sigma_{l}}^\dagger,
\end{equation}
where \( U_{\sigma_l} \) is the unitary \emph{permutation operator} corresponding to \( \sigma_l \in \operatorname{Sym}(n) \). This operation symmetrizes the density matrix \( \Phi \) with respect to all possible permutations of the subsystems.

\end{definition}

\medskip

The \(\operatorname{Sym}(n)\)-twirling map is a completely positive and trace-preserving (CPTP) map that has the following properties:
\begin{itemize}
    \item \textbf{Linear}: For \(a,b \in \mathbb{C}\), \( \mathcal{T}_{\operatorname{Sym}(n)} (a\rho_1 + b\rho_2) = a \mathcal{T}_{\operatorname{Sym}(n)}(\rho_1) + b \mathcal{T}_{\operatorname{Sym}(n)}(\rho_2) \),
    
    \item \textbf{Trace-Preserving}: \( \operatorname{Tr}(\mathcal{T}_{\operatorname{Sym}(n)}(\rho)) = \operatorname{Tr}(\rho) = 1 \).
\end{itemize}

\medskip

\if 0

\noindent The representation of \(\operatorname{Sym}(n)\) is defined by \( \pi: \operatorname{Sym}(n) \longrightarrow \mathcal{U}(\mathbb{C}^n) \) by
\[
\pi(\sigma) \, e_i = e_{\sigma(i)}, \quad \text{for} \quad i=1,\ldots,n,
\]
where \( \{e_i\}_{i=1}^n \) is the standard orthonormal basis of \( \mathbb{C}^n \) and \(\mathcal{U}(\mathbb{C}^n)\) denotes the group of unitary operators on \( \mathbb{C}^n \), [See \cite{fulton2013representation, serre1977linear}]. It is well known that this representation is \textit{reducible}.

\fi

\medskip

\noindent \textbf{Quantum one-time pad encryption(\(\mathsf{QOTP})\) \cite{ambainis2000private,cakan2023unclonable}.} The quantum one-time pad encryption (\(\mathsf{QOTP}\)) can perfectly hide any quantum message using a random classical key.

\noindent The quantum one-time pad encryption scheme is defined by a pair of quantum encryption and decryption circuits \( (\mathsf{QOTPEnc}, \mathsf{QOTPDec}) \) with
\[
\mathsf{QOTPEnc} : \mathcal{D}((\mathbb{C}^2)^{\otimes n}) \times \{0, 1\}^{2n} \longrightarrow \mathcal{D}((\mathbb{C}^2)^{\otimes n}) \text{ and } \mathsf{QOTPDec} : \mathcal{D}((\mathbb{C}^2)^{\otimes n}) \times \{0, 1\}^{2n} \longrightarrow \mathcal{D}((\mathbb{C}^2)^{\otimes n})
\]
defined as
The encryption and decryption are defined as:
\[
\mathsf{QOTPEnc}(\rho, k) = U_k \rho U_k^\dagger,
\]
\[
\mathsf{QOTPDec}(\sigma, k) = U_k^\dagger \sigma U_k.
\]
for any message \( \rho \in \mathcal{D}((\mathbb{C}^2)^{\otimes n}) \) and key \( k \in \{0, 1\}^{2n} \), where \(U_k = \bigotimes_{j=1}^{n} X^{k_{2j-1}} Z^{k_{2j}}\) represent the quantum operation applying the standard Pauli gates \( X, Z \).

\medskip

\begin{lemma}\label{prop:otp}(\cite{ambainis2000private,ccakan2023computational}) The quantum one-time pad encryption scheme is correct and perfectly secure for a randomly chosen key. That is,
\[
\mathsf{QOTPDec}(\mathsf{QOTPEnc}(\rho, k), k) = \rho
\]
for any key \( k \in \{0, 1\}^{2n} \), and
\[
\sum_{k \in \{0,1\}^{2n}} \frac{1}{2^{2n}} \mathsf{QOTPEnc}(\rho, k) = \sum_{k \in \{0,1\}^{2n}} \frac{1}{2^{2n}} \mathsf{QOTPEnc}(\sigma, k)
\]
for any two quantum states \( \rho, \sigma \in \mathcal{D}((\mathbb{C}^2)^{\otimes n}) \).
\end{lemma}

\medskip
\begin{definition}[Contraction]
Let $\mathcal{H}$ be a Hilbert space. A bounded linear operator $T \in \mathcal{B}(\mathcal{H})$ is called a \textbf{contraction} if its operator norm is less than or equal to one:
\[
\|T\| \le 1.
\]
\end{definition}

\begin{theorem}[Halmos's Dilation Theorem (1950) \cite{Halmos1967}]
Let $T \in \mathcal{B}(\mathcal{H})$ be a contraction on a Hilbert space $\mathcal{H}$. Then there exists a larger Hilbert space $\mathcal{K}$ and a unitary operator $U \in \mathcal{L}(\mathcal{K})$ such that $\mathcal{H}$ is a subspace of $\mathcal{K}$ and
\[
T^n = P_\mathcal{H} U^n \big|_{\mathcal{H}}
\quad \text{for all integers } n \ge 0.
\]
In this case, $U$ is called a \textbf{unitary dilation} of $T$.

\noindent Here:
\begin{itemize}
    \item $P_\mathcal{H}: \mathcal{K} \longrightarrow \mathcal{H}$ is the orthogonal projection onto the subspace $\mathcal{H}$.
    \item $\big|_{\mathcal{H}}$ denotes the restriction of the operator to the subspace $\mathcal{H}$.
\end{itemize}
\end{theorem}

\begin{comment}
\begin{theorem}[von Neumann inequality]
If $T\in\mathcal B(\mathcal H)$ is a contraction, then for every complex polynomial $p$ one has
\begin{equation}\label{eq:VN}
  \|p(T)\| \;\le\; \sup_{|z|\le 1} |p(z)| .
\end{equation}
\end{theorem}
\end{comment}

\begin{lemma}[Hoeffding's Inequality, \cite{Hoeffding1963}]\label{lem:Hoeffding}
Let $X_1, \dots, X_n$ be independent random variables such that $X_i \in [a_i, b_i]$ almost surely. Let $S_n = \sum_{i=1}^n X_i$. Then, for any $t > 0$,
\begin{equation}
    \Pr\left( |S_n - \mathbb{E}[S_n]| \ge t \right) \le 2\exp\left( - \frac{2t^2}{\sum_{i=1}^n (b_i - a_i)^2} \right).
\end{equation}
In the specific case of estimating the mean $\mu$ of a variable $X \in [a,b]$ using the empirical mean $\bar{X}$ of $N$ samples, this implies
\begin{equation}
    \Pr\left( |\bar{X} - \mu| \ge \epsilon \right) \le 2\exp\left( - \frac{2N\epsilon^2}{(b-a)^2} \right).
\end{equation}
\end{lemma}

\begin{fact}[Normalized Pauli basis {\cite[Ch.~10]{nielsen2001quantum}}]\label{fact:PauliONB}
Let $d=2^{d_1}$ and $\mathcal{P}_{d_1}$ be the $d_1$--qubit Pauli group modulo phases. The set $\{E_P:=P/\sqrt{d}: P\in\mathcal{P}_{d_1}\}$ is an orthonormal basis of $\mathcal{L}(\mathbb{C}^d)$ for the Frobenius inner product:
$\mathrm{Tr}(E_P^\dagger E_Q)=\delta_{P,Q}$ and $X=\sum_{P}\langle E_P,X\rangle_F\,E_P$ for all $X$.
\end{fact}

\begin{fact}[Commuting Pauli frames {\cite{wootters1989optimal,durt2010mutually}}]\label{fact:frames}
The nonidentity Paulis partition into $d+1$ disjoint maximal commuting sets (frames) $\mathcal{C}_0,\dots,\mathcal{C}_d$, each of size $d-1$. A projective measurement in the joint eigenbasis of a frame yields $\{\pm1\}$ outcomes for all $P\in\mathcal{C}_s$ in a single shot.
\end{fact}

\begin{lemma}[Schatten--norm comparison {\cite[Thm.~IV.2.5]{bhatia2013matrix}}]\label{lem:Schatten}
For $X\in\mathcal{L}(\mathbb{C}^d)$, $\|X\|_1\le \sqrt{d}\,\|X\|_2$.
\end{lemma}

\subsection{Quantum adversarial model}
Unless otherwise specified, we consider a quantum computational security setting where the adversaries are quantum polynomial-time (QPT) algorithms. A \textit{QPT adversary} or a \textit{circuit} \(C\) means a non-uniform family of circuits \(\{C_{\lambda}\}_{\lambda \in \mathbb{Z}^{+}}\) with 1-bit classical output, where each circuit has size bounded by \(\operatorname{poly}(\lambda)\) and is allowed to use a fixed basis set of gates, for example, \(\{\mathsf{H}, \mathsf{CNOT},\mathsf{S},\mathsf{T}\}\), etc. Each ancilla qubit is initialized to \(\ket{0}\). This model has been studied in quantum secret-sharing and has been studied in \cite{ccakan2023computational} for computational quantum secret-sharing.

\section{Classical Anamorphic Encryption}\label{sec:cae}
First, we review the original description and definition of anamorphic encryption. Anamorphic encryption is a form of public-key encryption (PKE) that enables a hidden communication mode alongside a regular encryption mode. Specifically, this construction allows a receiver to decrypt a ciphertext to reveal a standard message or, alternatively, a covert message, depending on the use of specific secret keys. Such a scheme can be deployed securely even under coercive environments where a user may be forced to reveal their private key. We refer \cite{persiano2022anamorphic, catalano2024anamorphic} to the reader for a detailed description.

\medskip

An anamorphic encryption scheme is defined as a public key encryption scheme $\mathcal{E} = (\mathsf{Gen}, \mathsf{Encrypt}, \mathsf{Decrypt})$, with additional algorithms $\mathcal{A} = (\mathsf{Gen}_a, \mathsf{Encrypt}_a, \mathsf{Decrypt}_a)$ that enable the encryption and decryption of covert messages.

\begin{enumerate}
    \item \textbf{Standard Encryption Scheme:} The PKE scheme $\mathcal{E} = (\mathsf{Gen}, \mathsf{Encrypt}, \mathsf{Decrypt})$ consists of the following algorithms:
    \begin{itemize}
        \item \(\mathsf{Gen}(1^{\lambda})\): A key generation algorithm that, given a security parameter \(\lambda\), outputs a public-private key pair \((\mathsf{pk}, \mathsf{sk})\).
        \item \(\mathsf{Encrypt}(\mathsf{pk}, m)\): An encryption algorithm that takes a public key \(\mathsf{pk}\) and a plaintext message \(m\), and outputs a ciphertext \(c\).
        \item \(\mathsf{Decrypt}(\mathsf{sk}, c)\): A decryption algorithm that takes a private key \(\mathsf{sk}\) and a ciphertext \(c\), and outputs the original message \(m\) or a special symbol $\bot$ if decryption fails.
    \end{itemize}

    \item \textbf{Anamorphic Triplet:} The anamorphic triplet $\mathcal{A} = (\mathsf{Gen}_a, \mathsf{Encrypt}_a, \mathsf{Decrypt}_a)$ introduces an additional encryption and decryption layer that enables hidden communication:
    \begin{itemize}
        \item \(\mathsf{Gen}_a(1^{\lambda})\): Given a security parameter \(\lambda\), outputs an \textit{anamorphic public key} \(\mathsf{apk}\), an \textit{anamorphic secret key} \(\mathsf{ask}\), a \textit{double key} \(\mathsf{dk}\), and an optional \textit{trapdoor key} \(\mathsf{tk}\).
        \item \(\mathsf{Encrypt}_a(\mathsf{apk}, \text{dk}, m, \hat{m})\): Given the anamorphic public key \(\mathsf{apk}\), double key \(\mathsf{dk}\), a visible message \(m\), and a covert message \(\hat{m}\), it produces an \textit{anamorphic ciphertext} \(\mathsf{act}\).
        \item \(\mathsf{Decrypt}_a(\mathsf{dk}, \mathsf{tk}, \mathsf{ask}, \mathsf{act})\): A decryption algorithm that takes the keys \(\mathsf{dk}\), \(\mathsf{tk}\), \(\mathsf{ask}\), and an anamorphic ciphertext \(\mathsf{act}\), outputting the covert message \(\hat{m}\) or a special symbol \(\perp\) if decryption fails.
    \end{itemize}
\end{enumerate}

\medskip

To define the security of anamorphic encryption schemes, two games \( \mathsf{RealG}_{\mathcal{E}}(\lambda, \mathcal{D}) \) and \( \mathsf{AnamorphicG}_{\mathcal{A}}(\lambda, \mathcal{D}) \) are defined to represent interactions with the real and anamorphic encryption modes, respectively. The goal is to evaluate whether an adversary \(\mathcal{D}\) can distinguish between these two games.

\medskip

In the real game, the encryption scheme operates as a standard PKE scheme without covert capabilities:
\begin{enumerate}
    \item \textbf{Key Generation:} A key pair \((\mathsf{pk}, \mathsf{sk}) \gets \mathsf{Gen}(1^{\lambda})\) is generated using the standard key generation function.
    \item \textbf{Oracle Access:} The adversary \(\mathcal{D}\) is provided access to an oracle \(\mathcal{O}_\mathsf{E}\) defined by
    \[
    \mathcal{O}_\mathsf{E}(\mathsf{pk}, m, \hat{m}) = \mathsf{Encrypt}(\mathsf{pk}, m),
    \]
    where \(\mathsf{Encrypt}\) produces a ciphertext containing only \(m\), ignoring any covert message \(\hat{m}\).
\end{enumerate}
The adversary \(\mathcal{D}\) uses \(\mathcal{O}_\mathsf{E}\) to query pairs of messages \((m, \hat{m})\) and receives corresponding ciphertexts generated in a standard encryption mode.

\smallskip

\begin{center}
\fbox{%
\begin{minipage}{0.5\textwidth}
\[
\begin{array}{l}
\mathsf{RealG}_\mathcal{E}(\lambda, \mathcal{D}) \\
\hline
(\mathsf{pk}, \mathsf{sk}) \leftarrow \mathsf{Gen}(1^{\lambda}) \\
\textbf{return } \mathcal{D}^{\mathcal{O}_{\mathsf{E}}(\mathsf{pk}, \cdot, \cdot)}(\mathsf{pk}, \mathsf{sk}); \\
\text{where } \mathcal{O}_{\mathsf{E}}(\mathsf{pk}, m, \hat{m}) = \mathsf{Encrypt}(\mathsf{pk}, m)
\end{array}
\]
\end{minipage}%
}
\end{center}

\medskip

In the anamorphic game, the encryption scheme operates in a mode that embeds a covert channel:
\begin{enumerate}
    \item \textbf{Anamorphic Key Generation:} Anamorphic key generation produces keys \((\mathsf{apk}, \mathsf{ask})\), \(\mathsf{tk}\), and \(\mathsf{dk}\) such that \((\mathsf{apk}, \mathsf{ask}), \mathsf{tk}, \mathsf{dk} \gets \mathsf{Gen}_a(1^{\lambda})\).
    \item \textbf{Oracle Access:} The adversary \(\mathcal{D}\) is provided access to an oracle \(\mathcal{O}_\mathsf{A}\) defined by
    \[
    \mathcal{O}_\mathsf{A}(\mathsf{apk}, \mathsf{dk}, m, \hat{m}) = \mathsf{Encrypt}_a(\mathsf{apk}, \mathsf{dk}, m, \hat{m}),
    \]
    where \(\mathsf{Encrypt}_a\) produces an anamorphic ciphertext that encodes both the visible message \(m\) and the covert message \(\hat{m}\).
\end{enumerate}

\smallskip
\begin{center}
\fbox{%
\begin{minipage}{0.6\textwidth}
\[
\begin{array}{l}
\mathsf{AnamorphicG}_\mathcal{A}(\lambda, \mathcal{D}) \\
\hline
((\mathsf{apk}, \mathsf{ask}), \mathsf{tk}, \mathsf{dk}) \leftarrow \mathsf{Gen}_{a}(1^{\lambda}) \\
\textbf{return } \mathcal{D}^{\mathcal{O}_{\mathsf{A}}(\mathsf{apk}, \mathsf{dk}, \cdot, \cdot)}(\mathsf{apk}, \mathsf{ask}); \\
\text{where } \mathcal{O}_{\mathsf{A}}(\mathsf{apk}, \mathsf{dk}, m, \hat{m}) = \mathsf{Encrypt}_{a}(\mathsf{apk}, \mathsf{dk}, m, \hat{m})
\end{array}
\]
\end{minipage}%
}
\end{center}

\medskip

\noindent The advantage of an adversary \(\mathcal{D}\) in distinguishing between the two games is given by:
\[
\mathsf{Adv}^{\mathsf{AME}}_{\mathcal{D}, \mathcal{E}, \mathcal{A}}(\lambda) = \left| \Pr\left[ \mathsf{RealG}_{\mathcal{E}}(\lambda, \mathcal{D}) = 1 \right] - \Pr\left[ \mathsf{AnamorphicG}_{\mathcal{A}}(\lambda, \mathcal{D}) = 1 \right] \right|,
\]
where \(\Pr[\mathsf{RealG}_{\mathcal{E}}(\lambda, \mathcal{D}) = 1]\) is the probability that \(\mathcal{D}\) identifies the game as the real game, and \(\Pr[\mathsf{AnamorphicG}_{\mathcal{A}}(\lambda, \mathcal{D}) = 1]\) is the probability that \(\mathcal{D}\) identifies the game as the anamorphic game.
 
\medskip

\begin{definition} (Anamorphic Encryption \cite{persiano2022anamorphic,catalano2024anamorphic, catalano2024generic}). 

A PKE scheme \(\mathcal{E} = (\mathsf{Gen}, \mathsf{Encrypt}, \mathsf{Decrypt})\) is called an \textit{anamorphic encryption scheme} if:
\begin{enumerate}
    \item It satisfies \textit{IND-CPA} security (indistinguishability under chosen-plaintext attack).
    \item There exists an anamorphic triplet \(\mathcal{A} = (\mathsf{Gen}_a, \mathsf{Encrypt}_a, \mathsf{Decrypt}_a)\) such that for any probabilistic polynomial-time (PPT) adversary \(\mathcal{D}\), the distinguishing advantage \(\mathsf{Adv}^{\mathsf{AME}}_{\mathcal{D}, \mathcal{E}, \mathcal{A}}(\lambda)\) is negligible in \(\lambda\).

    \item The existence of a covert message \(m'\) remains deniable as the ciphertext \(c'\) is indistinguishable from a regular ciphertext.
\end{enumerate}

\end{definition}

We refer the reader to [Page 13, \cite{catalano2024anamorphic}] for a detailed explanation.

\smallskip 
In the following section, we introduce an analogue of classical anamorphic encryption within the quantum encryption framework, encompassing both public-key and symmetric-key encryption schemes.

\section{Quantum Anamorphic Encryption}\label{sec:qae}
In this section, we propose an analogous definition of classical anamorphic encryption in the quantum encryption model, where the secrets are quantum density matrices from finite-dimensional Hilbert spaces with a quantum polynomial time (QPT) adversary and quantum adversarial pseudometric. In a quantum environment, an anamorphic encryption scheme would need to incorporate quantum-safe encryption methods.

\medskip

We define a quantum anamorphic encryption scheme as a quantum public key encryption scheme $\mathcal{Q} = (\mathsf{Gen}, \mathsf{QEnc}, \mathsf{QDec})$ with additional algorithms $\mathcal{Q}_a = (\mathsf{Gen}_a, \mathsf{QEnc}_a, \mathsf{QDec}_a)$ to support hidden communication in the presence of quantum adversaries.

\medskip

A quantum anamorphic encryption scheme is a quantum public-key encryption (QPKE) scheme with an additional anamorphic triplet of quantum algorithms, enabling hidden messages within ciphertexts. We denote this quantum anamorphic encryption scheme by the tuple $\mathcal{Q} = (\mathsf{Gen}, \mathsf{QEnc}, \mathsf{QDec})$ and $\mathcal{Q}_a = (\mathsf{Gen}_a, \mathsf{QEnc}_a, \mathsf{QDec}_a)$, where:

\begin{enumerate}
    \item \textbf{Quantum Public Key Encryption Scheme}: The QPKE scheme $\mathcal{Q} = (\mathsf{Gen}, \mathsf{QEnc}, \mathsf{QDec})$ consists of:
    \begin{itemize}
        \item $\mathsf{Gen}(1^{\lambda})$: A key generation algorithm that takes a security parameter $\lambda$ and outputs a public-private key pair $(\mathsf{pk}, \mathsf{sk})$.
        \item $\mathsf{QEnc}(\mathsf{pk}, \rho)$: A quantum encryption algorithm that takes a public key $\mathsf{pk}$ and a quantum state $\rho$ representing a message, producing a ciphertext in a quantum state $\mathsf{qc}$.
        \item $\mathsf{QDec}(\mathsf{sk}, \mathsf{qc})$: A quantum decryption algorithm that takes a private key $\mathsf{sk}$ and a ciphertext $\mathsf{qc}$, outputting the original message $\rho$ or a special failure symbol $\bot$ if decryption fails.
    \end{itemize}
    
    \item \textbf{Quantum Anamorphic Triplet}: The anamorphic triplet $\mathcal{Q}_a = (\mathsf{Gen}_a, \mathsf{QEnc}_a, \mathsf{QDec}_a)$ introduces additional quantum algorithms to enable covert communication:
    \begin{itemize}
        \item $\mathsf{Gen}_a(1^{\lambda})$: An anamorphic key generation algorithm that, given a security parameter $\lambda$, outputs an anamorphic public key $\mathsf{apk}$, an anamorphic secret key $\mathsf{ask}$, a double key $\mathsf{dk}$, and a potentially empty trapdoor key $\mathsf{tk}$.
        \item $\mathsf{QEnc}_a(\mathsf{apk}, \mathsf{dk}, \rho, \hat{\rho})$: An anamorphic encryption algorithm that, given $\mathsf{apk}$, $\mathsf{dk}$, a visible message state $\rho$, and a covert message state $\hat{\rho}$, outputs an anamorphic ciphertext $\mathsf{act}$.
        \item $\mathsf{QDec}_a(\mathsf{dk}, \mathsf{tk}, \mathsf{ask}, \mathsf{act})$: An anamorphic decryption algorithm that takes the double key $\mathsf{dk}$, trapdoor key $\mathsf{tk}$, anamorphic secret key $\mathsf{ask}$, and anamorphic ciphertext $\mathsf{act}$, and outputs the covert message $\hat{\rho}$ or a failure symbol $\perp$.
    \end{itemize}
\end{enumerate}

\medskip

To evaluate the security of quantum anamorphic encryption schemes, we introduce two games that distinguish the quantum real and anamorphic encryption modes. These games, \( \mathsf{RealG}_{\mathcal{Q}}(\lambda, \mathcal{D}) \) and \( \mathsf{AnamorphicG}_{\mathcal{Q}_a}(\lambda, \mathcal{D}) \), consider adversaries in the quantum polynomial-time (QPT) model.

\medskip

In the real game, the encryption scheme operates in a traditional QPKE setting with no covert channel.
\begin{enumerate}
    \item \textbf{Key Generation:} The key pair \((\mathsf{pk}, \mathsf{sk}) \gets \mathsf{Gen}(1^{\lambda})\) is generated using the standard quantum key generation function.
    \item \textbf{Oracle Access:} The adversary $\mathcal{D}$ is provided access to a quantum oracle $\mathcal{O}_{\mathsf{E}}$ defined by:
    \[
    \mathcal{O}_{\mathsf{E}}(\mathsf{pk}, \rho, \hat{\rho}) = \mathsf{QEnc}(\mathsf{pk}, \rho),
    \]
    where $\mathsf{QEnc}$ produces a ciphertext containing only the visible message $\rho$, ignoring any covert message $\hat{\rho}$.
\end{enumerate}
The adversary $\mathcal{D}$ interacts with $\mathcal{O}_{\mathsf{E}}$, querying pairs of quantum states $(\rho, \hat{\rho})$ and receiving quantum ciphertexts generated in the standard encryption mode.

\smallskip
\begin{center}
\fbox{%
\begin{minipage}{0.5\textwidth}
\[
\begin{array}{l}
\mathsf{RealG}_\mathcal{E}(\lambda, \mathcal{D}) \\
\hline
(\mathsf{pk}, \mathsf{sk}) \xleftarrow{\$} \mathsf{Gen}(1^{\lambda}) \\
\textbf{return } \mathcal{D}^{\mathcal{O}_{\mathsf{E}}(\mathsf{pk}, \cdot, \cdot)}(\mathsf{pk}, \mathsf{sk}); \\
\text{where } \mathcal{O}_{\mathsf{E}}(\mathsf{pk}, \rho, \hat{\rho}) = \mathsf{Encrypt}(\mathsf{pk}, \rho)
\end{array}
\]
\end{minipage}%
}
\end{center}

\medskip

In the anamorphic game, the encryption scheme operates in an anamorphic mode that enables covert communication.
\begin{enumerate}
    \item \textbf{Anamorphic Key Generation:} The anamorphic key generation algorithm produces $(\mathsf{apk}, \mathsf{ask})$, $\mathsf{tk}$, and $\mathsf{dk}$, such that $(\mathsf{apk}, \mathsf{ask}), \mathsf{tk}, \mathsf{dk} \gets \mathsf{Gen}_a(1^{\lambda})$.
    \item \textbf{Oracle Access:} The adversary $\mathcal{D}$ is provided access to a quantum oracle $\mathcal{O}_{\mathsf{A}}$ defined by:
    \[
    \mathcal{O}_{\mathsf{A}}(\mathsf{apk}, \mathsf{dk}, \rho, \hat{\rho}) = \mathsf{QEnc}_a(\mathsf{apk}, \mathsf{dk}, \rho, \hat{\rho}),
    \]
    where $\mathsf{QEnc}_a$ produces an anamorphic ciphertext encoding both the visible message $\rho$ and the covert message $\hat{\rho}$.
\end{enumerate}
The adversary $\mathcal{D}$ uses $\mathcal{O}_{\mathsf{A}}$ to query quantum states $(\rho, \hat{\rho})$, receiving anamorphic ciphertexts that contain both visible and covert components.

\smallskip
\begin{center}
\fbox{%
\begin{minipage}{0.6\textwidth}
\[
\begin{array}{l}
\mathsf{AnamorphicG}_\mathcal{A}(\lambda, \mathcal{D}) \\
\hline
((\mathsf{apk}, \mathsf{ask}), \mathsf{tk}, \mathsf{dk}) \xleftarrow{\$} \mathsf{Gen}_{a}(1^{\lambda}) \\
\textbf{return } \mathcal{D}^{\mathcal{O}_{\mathsf{A}}(\mathsf{apk}, \mathsf{dk}, \cdot, \cdot)}(\mathsf{apk}, \mathsf{ask}); \\
\text{where } \mathcal{O}_{\mathsf{A}}(\mathsf{apk}, \mathsf{dk}, \rho, \hat{\rho}) = \mathsf{Encrypt}_{a}(\mathsf{apk}, \mathsf{dk}, \rho, \hat{\rho})
\end{array}
\]
\end{minipage}%
}
\end{center}

\medskip

To quantify the advantage of a quantum adversary in computationally distinguishing between the two games, we define the distinguishing advantage with respect to a family of quantum circuits $\mathcal{F}$:
\[
\mathsf{Adv}_{\mathcal{F}}^{\mathsf{AME}}(\rho, \sigma) = \max_{C \in \mathcal{F}} \left| \Pr\left[C(\rho) = 1\right] - \Pr\left[C(\sigma) = 1\right] \right|,
\]
where $C$ is a circuit from $\mathcal{F}$ that outputs 1 if it identifies the quantum state as belonging to the anamorphic game \cite{ccakan2023computational}.

\medskip

Next, we define the quantum analogue of anamorphic encryption.

\begin{definition}\label{def:qape}
    A QPKE scheme $\mathcal{Q} = (\mathsf{Gen}, \mathsf{QEnc}, \mathsf{QDec})$ is defined as a \textit{quantum anamorphic encryption scheme} if:
\begin{enumerate}
    \item It satisfies quantum \textsf{qIND-qCPA} security (indistinguishability under chosen-plaintext attack with quantum adversaries).
    \item There exists a quantum anamorphic triplet $\mathcal{Q}_a = (\mathsf{Gen}_a, \mathsf{QEnc}_a, \mathsf{QDec}_a)$ such that, for any quantum polynomial-time (QPT) adversary $\mathcal{D}$, the computationally distinguishing advantage
\[\mathsf{Adv}_{\mathcal{D}}^{\mathsf{AME}}(\lambda) = \left| \Pr\left[ \mathsf{RealG}_{\mathcal{Q}}(\lambda, \mathcal{D}) = 1 \right] - \Pr\left[ \mathsf{AnamorphicG}_{\mathcal{Q}_a}(\lambda, \mathcal{D}) = 1 \right] \right| < \mathsf{negl}(\lambda).
\]
\end{enumerate}
\end{definition}

\medskip

The above definition of quantum anamorphic encryption, which we have discussed, is based on public-key encryption. Now we propose an analogous definition of quantum anamorphic encryption based on symmetric key encryption.

\medskip

A general quantum symmetric key encryption is defined as follows:

\begin{definition}
A \textit{quantum symmetric-key encryption scheme} is a triplet of quantum algorithms\((\mathsf{Gen}, \mathsf{QEnc}, \mathsf{QDec}),\) where:
\begin{itemize}
    \item $\mathsf{Gen}(1^\lambda)$: Takes as input a security parameter $1^{\lambda}$ and outputs a secret key $k$. The key $k$ can be classical or quantum.
    \item $\mathsf{QEnc}(k, \rho)$: Takes the secret key $k$ and a quantum message state $\rho$ in some Hilbert space $\mathcal{H}$, and outputs a ciphertext $\mathsf{qc}$ a quantum state in some, possibly different, Hilbert space $\mathcal{H}_{C}$. 
    \item $\mathsf{QDec}(k, \mathsf{qc})$: Takes the secret key $k$ and a ciphertext state $\mathsf{qc}$, and attempts to recover the original message $\rho$. If the ciphertext is invalid, it outputs $\bot$.
\end{itemize}
\end{definition}

\medskip

Now, we propose a definition of quantum anamorphic encryption. A \textit{quantum anamorphic encryption scheme} adds a second \textit{anamorphic} mode of operation, consisting of another triplet of quantum algorithms \((\mathsf{Gen}_a, \mathsf{QEnc}_a, \mathsf{QDec}_a)\)  which allows \textit{embedding} a covert message $\hat{\rho}$ inside the same ciphertext structure, but in such a way that an adversary cannot distinguish between normal encryption and anamorphic encryption with high probability, where:

\begin{itemize}
    \item \textbf{Anamorphic Key Generation \(\mathsf{(Gen_{a}})\):}  The algorithm \(\mathsf{Gen_{a}}\) takes as input the security parameter \(1^{\lambda}\) and returns  \(k_{a}=(k, \mathsf{dk}, \mathsf{tk})\), where $k_{a}$ is the \textit{anamorphic secret key} when the system is operating in \textit{anamorphic mode}. The key \(k\) is the normal key used to encrypt the original message, and the $\mathsf{dk}$ double key and $\mathsf{tk}$ trapdoor key are additional secret keys that can be classical keys or quantum states that may be necessary to embed and extract covert messages and which will never be given to the dictator. Either might be an empty string if not needed.

    \item \textbf{Anamorphic Encryption (\(\mathsf{QEnc}_a\)):} The algorithm \(\mathsf{QEnc}_a\) takes as input the key \(k_{a}\) and two quantum messages $\rho$, the original and $\hat{\rho}$, the covert quantum message, and returns a quantum anamorphic ciphertext \(\mathsf{qact}\), that is, \(\mathsf{QEnc}_a\bigl(k_a, \rho, \hat{\rho}\bigr)=\mathsf{qact}\), which must be indistinguishable from a \text{original} ciphertext produced by $(k, \rho) \mapsto \mathsf{QEnc}(k, \rho)$.

    \item \textbf{Anamorphic Decryption (\(\mathsf{\mathsf{QDec}_a}\)):} The algorithm takes as input the anamorphic key \(k_a\) and anamorphic ciphertext \(\mathsf{act}\) and outputs the covert message \(\hat{\rho}\), that is, \(\mathsf{QDec}_a\bigl(k_a, \mathsf{qact}\bigr)=\hat{\rho}.\) If recovery fails, the algorithm outputs $\bot$.
\end{itemize}

\medskip

\noindent \textbf{Security definition via two indistinguishability games:} We define two worlds or games as previous, a \textit{real} (no covert channel) world and an \textit{anamorphic} (with covert channel) world. An adversary attempts to distinguish these two scenarios.

\medskip

\begin{description}
    \item \textbf{The Real Game} ($\mathsf{RealG}_{\mathsf{Q}}(\lambda, \mathcal{D})$):
    \begin{enumerate}
        \item A secret key $k \leftarrow \mathsf{Gen}(1^\lambda)$ is generated, which is the normal encryption key.
        \item The adversary $\mathcal{D}$ is given oracle access to a real encryption map
        \[
          \mathcal{O}_{\mathsf{E}}(k, \rho, \hat{\rho}) \;=\; \mathsf{QEnc}\bigl(k, \rho\bigr),
        \]
        which \textit{ignores} the covert message $\hat{\rho}$, that is, the oracle only encrypts the original state $\rho$ in the normal mode.
        \item $\mathcal{D}$ can make polynomially many queries and eventually outputs a guess bit $b \in \{0,1\}$, meaning it guesses whether it is in the real or anamorphic game.
    \end{enumerate}

    \item \textbf{The Anamorphic Game} ($\mathsf{AnamorphicG}_{\mathsf{Q}_a}(\lambda, \mathcal{D})$):
    \begin{enumerate}
        \item The anamorphic keys are generated: $k_a \leftarrow \mathsf{Gen}_a(1^\lambda)$.
        \item The adversary $\mathcal{D}$ is given oracle access to
        \[
          \mathcal{O}_{\mathsf{A}}\bigl(k_a, \rho, \hat{\rho}\bigr)
          \;=\;
          \mathsf{QEnc}_a\bigl(k_a, \rho, \hat{\rho}\bigr),
        \]
        which produces an \emph{anamorphic ciphertext} that contains both $\rho$ and $\hat{\rho}$. 
        \item As before, $\mathcal{D}$ makes a number of queries and finally outputs a guess bit $b \in \{0,1\}$.
    \end{enumerate}
\end{description}

\medskip

\noindent \textbf{Adversarial Advantage.} We define the advantage of adversary $\mathcal{D}$ distinguishing the above two games by
\[
  \mathsf{Adv}^{\mathsf{AME}}_{\mathcal{D}, \mathsf{Q}, \mathsf{Q}_a}(\lambda)
  \;=\;
  \Bigl|\,
    \Pr\bigl[\mathsf{RealG}_{\mathsf{Q}}(\lambda,\mathcal{D}) = 1\bigr]
    \;-\;
    \Pr\bigl[\mathsf{AnamorphicG}_{\mathsf{Q}_a}(\lambda,\mathcal{D}) = 1\bigr]
  \Bigr|.
\]

\medskip

\begin{definition}(Quantum Anamorphic Symmetric-Key Encryption)\label{def:qase}
\label{def:QAE-Symm}
A triple $(\mathsf{Gen}, \mathsf{QEnc}, \mathsf{QDec})$ is called a \emph{quantum symmetric-key encryption scheme}, and a triple $(\mathsf{Gen}_a, \mathsf{QEnc}_a, \mathsf{QDec}_a)$ is called its \emph{anamorphic extension}, if:

\begin{enumerate}
    \item \textbf{Correctness.} For all original and covert quantum messages, $\rho$ and $\hat{\rho}$, respectively,
    \[
       \mathsf{QDec}\bigl(k,\, \mathsf{QEnc}(k,\rho)\bigr) \;=\; \rho
    \]
    and similarly,
    \[
       \mathsf{QDec}_a\bigl(k_a,\, \mathsf{QEnc}_a(k_a,\rho,\hat{\rho})\bigr)
       \;=\;
       \hat{\rho}.
    \]

    \item \textbf{Security Against Chosen-Plaintext (Quantum) Attacks.} 
    The scheme $(\mathsf{Gen}, \mathsf{QEnc}, \mathsf{QDec})$ is quantum \textsf{qIND-qCPA} secure (in the symmetric-key sense), meaning that no QPT adversary $\mathcal{D}$ can distinguish encryptions of two chosen quantum states (or classical messages) with more than negligible advantage in $\lambda$.

    \item \textbf{Anamorphic Indistinguishability.} 
    There is an anamorphic extension $(\mathsf{Gen}_a, \mathsf{QEnc}_a, \mathsf{QDec}_a)$ such that for every QPT adversary $\mathcal{D}$, the distinguishing advantage 
    \[
      \mathsf{Adv}^{\mathsf{AME}}_{\mathcal{D}, \mathsf{Q}, \mathsf{Q}_a}(\lambda)
      \;=\;
      \bigl|\,
        \Pr\bigl[\mathsf{RealG}_{\mathsf{Q}}(\lambda,\mathcal{D}) = 1\bigr]
        \;-\;
        \Pr\bigl[\mathsf{AnamorphicG}_{\mathsf{Q}_a}(\lambda,\mathcal{D}) = 1\bigr]
      \bigr|
      \;<\;
      \mathsf{negl}(\lambda).
    \]
    In other words, the ciphertexts generated by $\mathsf{QEnc}$ versus those generated by $\mathsf{QEnc}_a$ (even on \emph{pairs} of inputs $(\rho,\hat{\rho})$) are computationally indistinguishable to any quantum adversary.

    \if 0
    
    \item \textbf{Deniability.} 
    The scheme also provides a deniability property: even if an adversary inspects the key $k_a$ or partially obtains some side information about $\text{dk},\text{tk}$, there is no efficient method to prove that a ciphertext $\text{qact}$ contains a covert message (beyond the normal encryption of $\rho$). Equivalently, $\text{qact}$ is \emph{indistinguishable} from a ciphertext produced by the standard encryption algorithm on some state. 

    \fi
\end{enumerate}
\end{definition}

\medskip

\subsection{Quantum Indistinguishability under Quantum Chosen Plaintext Attack (qIND-qCPA)}
In the paper \cite{gagliardoni2016semantic}, the definition of \(\mathsf{qIND-qCPA}\) has been defined. We now state the indistinguishability under the quantum chosen-plaintext attack experiment in a fully operator-theoretic model for our scheme. Our definition follows the game-based approach of \cite{gagliardoni2016semantic,boneh2013secure,alagic2017quantum}, with \emph{quantum} oracle access. 

\smallskip
Fix $(d_1,d_2,\eta)$ and the scheme parameters. The \emph{encryption channel} $\mathcal{E}$ is the CPTP map that, on input a pair of message registers $(M_o,M_c)$ (possibly entangled with an arbitrary environment $E$ held by the adversary), samples uniform $k\in\{0,1\}^{2d_1}$, $k'\in\{0,1\}^{2d_2}$ and $\sigma_l \in \mathrm{Sym}(2^{d_{1}+1})$, computes $M_o',M_c',M_c''$ as above, forms $M_a$ by~\eqref{eq:Ma}, obfuscates as in~\eqref{eq:Mf}, and returns the ciphertext register $RM$. Formally, letting $\mathbb{E}$ denote the classical expectation over $(k,k',\sigma_l)$,
\[
\mathcal{E}(\rho_{M_o M_c}) \;=\; \mathbb{E}_{k,k',\sigma_l}\big[\,U_{\sigma_l}\,M_a(M_o,M_c;k,k')\,U_{\sigma_l}^\dagger\,\big].
\]

\begin{definition}[\textsf{qIND--qCPA} experiment]\label{def:qind-qcpa}
A \emph{quantum} adversary $\mathcal{A}$ is any interactive QPT that can make polynomially many (adaptive) queries to the channel $\mathcal{E}$ on inputs of its choice, both \emph{before and after} the challenge phase.

\smallskip
\noindent\textbf{Setup and pre-challenge:} $\mathcal{A}$ starts in an arbitrary joint state $\rho_{E}$ (on any finite-dimensional Hilbert space $\mathcal{H}_E$) and may interact with the encryption oracle $\mathcal{E}$ on arbitrarily chosen inputs $(M_o,M_c)$ that may be entangled with $E$. Let $(I_E\otimes\mathcal{E})$ denote one such query.

\smallskip
\noindent\textbf{Challenge:} $\mathcal{A}$ outputs two \emph{pairs} of density operators of matching dimensions,
\[
(M_o^{(0)},M_c^{(0)}),\qquad (M_o^{(1)},M_c^{(1)}),
\]
possibly entangled with (disjoint parts of) $E$. The challenger samples a uniform bit $b\in\{0,1\}$ and returns the challenge ciphertext
\[
C_b \;=\; \mathcal{E}\!\left(M_o^{(b)}\otimes M_c^{(b)}\right).
\]

\smallskip
\noindent\textbf{Post-challenge:} $\mathcal{A}$ may continue to make polynomially many (adaptive) queries to $\mathcal{E}$ on arbitrary inputs independent of the bit $b$.

\smallskip
\noindent\textbf{Guess:} $\mathcal{A}$ outputs a bit $b'\in\{0,1\}$. The \emph{advantage} is
\[
\mathrm{Adv}^{\mathrm{qIND\text{-}qCPA}}_{\mathsf{QAE}}(\mathcal{A})\;:=\;\big|\,\Pr[b'=1\mid b=1]-\Pr[b'=1\mid b=0]\,\big|.
\]
We say $\mathsf{QAE}$ is \emph{\textsf{qIND--qCPA}-secure} if $\mathrm{Adv}^{\mathrm{qIND\text{-}qCPA}}_{\mathsf{QAE}}(\mathcal{A}) < \mathsf{negl}(\lambda)$ for all QPT adversaries $\mathcal{A}$.
\end{definition}

\paragraph{Remarks on the model.}
(i) Fresh $(k,k',\sigma_l)$ are sampled independently for \emph{every} oracle use (including the challenge). (ii) The oracle is strictly CPTP and consumes its input registers; hence, no trivial \emph{identity} channels are present. (iii) The definition allows arbitrary entanglement between the adversary's private memory $E$ and its query registers; the proof below handles this by proving that the oracle is a \emph{constant} channel after averaging the coins.

\medskip

\section{Technical Details}\label{technical}
\noindent In this section, we have proposed a construction of quantum anamorphic symmetric key encryption. Let \(\mathcal{H}_{M}=(\mathbb{C}^{2})^{\otimes d_{1}}\) and \(\mathcal{H}_{M_{c}}=(\mathbb{C}^{2})^{\otimes d_{2}}\). Let \(M_{o} \in \mathcal{D}(\mathcal{H}_{M})\) be the mixed density matrix representing the original message, and let \(M_{c} \in \mathcal{D}(\mathcal{H}_{M_{c}})\) be the mixed density matrix representing the covert message. Hence both the matrices \(M_{o}\) and \(M_{c}\) are Hermitian, positive semi-definite with \(\Tr(M_{o})=1\) and \(\Tr(M_{c})=1\). But as per our construction, we restrict the density matrix \(M_{o}\) to be strictly positive definite. The anamorphic message contains both the original and covert messages, which Alice needs to send to Bob. On the dictator's demand, Bob will hand over only the anamorphic ciphertext and the original keys to the dictator so that the dictator gets only the original message, and also he will be unable to distinguish between the original and the anamorphic ciphertexts.

\smallskip
In the following subsection, first, we have described our construction mathematically, and then we have presented our algorithms based on quantum operations. 

\medskip

\subsection{Main Construction}\label{subsec:main}
We independently encrypt \( M_o \) and \( M_c \) using the \(\mathsf{QOTP}\) scheme, with separate keys \( k \) and \( k' \), respectively.

\begin{enumerate}
    \item \textbf{Encryption of \( M_o \):} Let \( M_o' = \mathsf{QOTPEnc}(M_o, k) \), where \( \mathsf{QOTPEnc} \) denotes the \(\mathsf{QOTP}\) encryption operator with key \( k \in \{0, 1\}^{2d_1} \). This operation is defined as follows:
    
    \begin{equation}
    M_o' = \left( X^{k_1} Z^{k_2} \otimes X^{k_3} Z^{k_4} \otimes \cdots \otimes X^{k_{2d_1-1}} Z^{k_{2d_1}} \right) M_o \left( X^{k_1} Z^{k_2} \otimes X^{k_3} Z^{k_4} \otimes \cdots \otimes X^{k_{2d_1-1}} Z^{k_{2d_1}} \right)^{\dagger}.
    \end{equation}

    \item \textbf{Encryption of \( M_c \):} Let \( M_c' = \mathsf{QOTPEnc}(M_c, k') \), where \( k' \in \{0, 1\}^{2d_2} \) is the \(\mathsf{QOTP}\) key used for encrypting \( M_c \). This operation is defined as:
    
    \begin{equation}
    M_c' = \left( X^{k'_1} Z^{k'_2} \otimes X^{k'_3} Z^{k'_4} \otimes \cdots \otimes X^{k'_{2d_2-1}} Z^{k'_{2d_2}} \right) M_c \left( X^{k'_1} Z^{k'_2} \otimes X^{k'_3} Z^{k'_4} \otimes \cdots \otimes X^{k'_{2d_2-1}} Z^{k'_{2d_2}} \right)^{\dagger}.
    \end{equation}
    
\end{enumerate}

\noindent The Hilbert spaces associated with \(M_o'\) and \(M_c'\) are \(\mathcal{H}_{M} = (\mathbb{C}^2)^{\otimes d_1}\) which is of dimension \(2^{d_1}\) and \(\mathcal{H}_{M_{c}} = (\mathbb{C}^2)^{\otimes d_2}\) which is of the dimension \(2^{d_2}\). 

\medskip 

\noindent Without loss of generality, let \(d_{2} \leq d_{1}\) and if \(d_{2} < d_{1}\), then pad the density matrix \(M_{c}'\) with \((2^{d_{1}}-2^{d_{2}})\) zero rows and columns to make it a \((2^{d_{1}} \times 2^{d_{1}})\) matrix, and we denote it by \(M_{c}''\).

\medskip

\if 0

\noindent Choose a basis
\[
\Bigl\{\ket {0},\; \ket{1},\;\ldots,\;\ket{2^{d_2}-1} \Bigr\} \quad \text{for }\mathcal{H}_c, \quad \text{and} \quad \Bigl\{\ket{0},\;  \ket{1},\; \ldots,\;\ket{2^{d_1}-1} \Bigr\} \quad \text{for }\mathcal{H}_{o},
\]
then \(M_{c}''\) acts on that larger basis by
\[
  M_{c}'' \,\ket{x}
  \;=\;
  \begin{cases}
    M_{c}'\,\ket{x}, & 0 \;\le x < 2^{d_2},\\[6pt]
    0, & 2^{d_2} \,\le x < 2^{d_1}. \label{eq:pad1}
  \end{cases}
\]

\fi

\noindent We construct \(M_{c}''\) by introducing \((d_1 - d_2)\) 
\emph{ancillary qubits} in a fixed state \(\lvert 0\rangle^{\otimes (d_1 - d_2)}\).

\medskip

\noindent Define the extended Hilbert space
\[
  \mathcal{H}_{M_{c}} \,\otimes\, (\mathbb{C}^2)^{\otimes (d_1 - d_2)}
  \;\cong\;
  \bigl(\mathbb{C}^2\bigr)^{\otimes d_1}
   \;=\;
  \mathcal{H}_{M}.
\]
Consider the isometric embedding
\[
  V : 
  \mathcal{H}_{M_{c}}
  \;\longrightarrow\;
  \mathcal{H}_{M}
\]
defined by
\[
  V \,\lvert \psi\rangle
  \;=\;
  \lvert \psi\rangle
  \;\otimes\;
  \lvert 0\rangle^{\otimes (d_1 - d_2)},
  \quad
  \forall\,\lvert \psi\rangle \in \mathcal{H}_{M_{c}}.
\]
Here, \(\lvert 0\rangle\) denotes the computational-basis state of a single qubit.

\noindent For a density matrix \(M_{c}' \in \mathcal{D}(\mathcal{H}_{M_{c}})\), define
\[
  M_{c}''
  \;=\;
  V\,M_{c}'\,V^\dagger
  \;\in\;
  \mathcal{D}(\mathcal{H}_{M}).
\]
We typically describe transformations on density operators by completely positive, trace-preserving (CPTP) maps. The above padding can be described as a linear map
\[
  \mathcal{E}^{\mathrm{pad}}
    : 
    \mathcal{D}\bigl(\mathcal{H}_{M_{c}}\bigr)
    \;\longrightarrow\;
    \mathcal{D}\bigl(\mathcal{H}_{M}\bigr).
\]
We define

\begin{equation} 
  \mathcal{E}^{\mathrm{pad}}(M_{c}')
  \;=M_{c}''=\;
  \begin{cases}
    M_{c}', & \text{if } d_1=d_2,\\[6pt]
    \displaystyle 
    V \,M_{c}' \, V^\dagger,
      & \text{if } d_2 < d_1,
  \end{cases}
  \label{eq:pad2}
\end{equation}

where \(V\) is the isometric embedding from above. Since \(V^\dagger V = I_{\mathcal{H}_{M_{c}}}\) and \(V\,V^\dagger\) 
is the projector \(\Pi_V\) onto \(\mathcal{H}_{M_{c}} \otimes \lvert 0\rangle^{\otimes (d_1 - d_2)}\), 
\(\mathcal{E}^{\mathrm{pad}}\) is completely positive and trace-preserving. Since \(V\) is an isometry,

\medskip

\noindent \(\bullet\) \textbf{Positivity is preserved:}
For any \(\phi \in \mathcal{H}_{M}\),
\[
\bra{\phi}VM_{c}' V^{\dagger}\ket{\phi} \geq 0.
\]
\noindent \(\bullet\) \textbf{Trace is preserved:}
\[
\Tr(VM_{c}'V^{\dagger})=\Tr(M_{c}'), \quad (\text{since} \quad V^{\dagger}V=I_{\mathcal{H}_{M_{c}}}).
\]
Hence, it is a valid quantum channel.

\medskip

\noindent Given the security parameter \(\mathsf{negl}(\lambda) >0\), choose $\eta \in \mathbb{Z}^+$ such that $\dfrac{1}{\eta} < \mathsf{negl}(\lambda)$ and \begin{equation}
  \frac{1}{\eta^2}
  \left\|(M_c'')^\dagger (M_o')^{-1} M_c'' \right\|
  \;\;\le\;\;
  \frac{1}{4} \lambda_{\min}(M_o'),
\end{equation}

where \(\lambda_{\min}(M_o')\) is the minimum eigenvalue of \(M_o'\) on its support and \(\|\cdot\|\) is the operator norm. As \(\eta\) is a non-zero real number, it is clear that the matrix \(M_{o}'\) should be not only positive semi-definite but also strictly positive definite, and hence the matrix \(M_{o}\) should be strictly positive definite for our construction.
In our construction, we construct a state \(M_{a}\) on a larger Hilbert space \(\mathcal{H}_{C}=(\mathbb{C}^{2})^{\otimes ({d_{1}+1})}\), that is we construct \(M_{a} \in \mathcal{D}(\mathcal{H}_{C})\).

\medskip
\noindent We define 
\begin{align}
M_{a} & :=
    \ket{0}\bra{0}
    \otimes
    \tfrac{1}{2}\,M_o'
  \;+\;
  \ket{0}\bra{1}
    \otimes
    \tfrac{1}{\eta}\,M_c''
  \;+\;
  \ket{1}\bra{0}
    \otimes
    \tfrac{1}{\eta}\,\bigl(M_c''\bigr)^\dagger
  \;+\;
  \ket{1}\bra{1}
    \otimes
    \tfrac{1}{2}\,M_o'\\
    &=\begin{pmatrix} \frac{1}{2}M_o' & \frac{1}{\eta}M_{c}'' \\ \frac{1}{\eta}(M_{c}'')^{\dagger} & \frac{1}{2}M_o' \end{pmatrix}. %\label{eq:M_a}
\end{align}

\medskip

For \(b \in \{0,1\}\), we define
\begin{align}\label{eq:Ma}
    M_{a}^{(b)}&:= (1-b)\Big(\ket{0}\bra{0}
    \otimes
    \tfrac{1}{2}\,M_o'
  \;+\;
  \ket{1}\bra{1}
    \otimes
    \tfrac{1}{2}\,M_o'\Big) \; + \; b M_{a} \\
    &=(1-b) \begin{pmatrix} \frac{1}{2}M_o' & \mathbf{0}_{2^{d_{1}} \times 2^{d_{1}}} \\ \mathbf{0}_{2^{d_{1}} \times 2^{d_{1}}} & \frac{1}{2}M_o' \end{pmatrix} + b M_{a}.
\end{align}

\medskip

To construct the final state for encoding, we choose a permutation matrix \(U_{\sigma_{l}}\) of order $(2^{d_1 + 1} \times 2^{d_1 + 1})$ uniformly randomly and create the final state 
\begin{equation} \label{eq:Mf}
M_{f}^{(b)} := U_{\sigma_{l}} M_{a}^{(b)} U_{\sigma_{l}}^{\dagger} \in \mathcal{D}(\mathcal{H}_{C}). 
\end{equation}

\medskip
Note that the state \(M_{f}^{(0)}\) is the \textit{original ciphertext}; as before, applying the permutation matrix, it is only encrypted using the original key and can be decrypted using only the original key.

\medskip

\noindent We now describe the quantum anamorphic decryption to extract the original message.

\noindent Define
\begin{equation}
  M_{d}^{(1)} 
  \;:=\;
  U_{\sigma_{l}}^{\dagger}\, M_{f}^{(1)}\, U_{\sigma_{l}}.
\end{equation}
Since \(U_{\sigma_{l}}\) is a permutation matrix, it is unitary, and
\(U_{\sigma_{l}}^{\dagger} = U_{\sigma_{l}}^{-1}\). Hence \(M_{d}^{(1)}\) recovers the block structure used by the encryption.

\noindent We define the following projectors on the first qubit:
\begin{equation}
  \Pi_{0} \;=\;
  \ket{0}\bra{0} \,\otimes\, I_{2^{d_1}},
  \qquad
  \Pi_{1} \;=\;
  \ket{1}\bra{1} \,\otimes\, I_{2^{d_1}},
\end{equation}
where \(I_{2^{d_1}}\) is the identity operator on \(\mathcal{H}_{M_{o}}\). Then 
\begin{equation}
\Pi_0 + \Pi_1
  \;=\;
  I_{2^{d_1+1}},
  \quad
  \Pi_0\,\Pi_1 \;=\; 0,
  \quad
  \Pi_0^2=\Pi_0,
  \quad
  \Pi_1^2=\Pi_1.
\end{equation}
Define
\begin{equation}
  M_{d}^{(1)}(0,0)
  \;=\;
  \Pi_{0}\,M_{d}^{(1)}\,\Pi_{0} \quad \text{in} \quad \mathcal{D}((\mathbb{C}^{2})^{\otimes{(d_{1}+1})}), 
\end{equation}
and 
\begin{equation}
  M_{d}^{(1)}(1,1)
  \;=\;
  \Pi_{1}\,M_{d}^{(1)}\,\Pi_{1} \quad \text{in} \quad \mathcal{D(}(\mathbb{C}^{2})^{\otimes{(d_{1}+1})}).
\end{equation}
By \emph{extracting} these two blocks (first and fourth) and \emph{adding} them, we obtain an operator on 
\((\mathbb{C}^{2})^{\otimes(d_{1}+1)}\) (since \(\Pi_{0}+\Pi_{1} = I_{2^{d_1+1}}\) on the first qubit, restricted 
to the appropriate blocks). 

\noindent We get
\begin{equation}
  M_{a}^{(0)} 
  \;=\;
  M_{d}^{(1)}(0,0) \;+\; M_{d}^{(1)}(1,1)
  \;=\;
  \bigl(\ket{0}\bra{0}\otimes I_{2^{d_1}}\bigr)\,M_{d}^{(1)}\,\bigl(\ket{0}\bra{0}\otimes I_{2^{d_1}}\bigr)
  \;+\;
  \bigl(\ket{1}\bra{1}\otimes I_{2^{d_1}}\bigr)\,M_{d}^{(1)}\,\bigl(\ket{1}\bra{1}\otimes I_{2^{d_1}}\bigr).
\end{equation}

\medskip

\noindent Now to extract \((0,0)\)-block and \((1,1)\)-block, that is the first block and the fourth block only, in the reduced space \(\mathcal{H}_{M}\), we define

\begin{align}
  \widetilde{M}_{o}'(0,0)\Big\rvert_{\mathcal{H}_{M}} 
  &:= \left( \bra{0} \otimes I_{2^{d_1}} \right) M_{a}^{(0)} 
       \left( \ket{0} \otimes I_{2^{d_1}} \right), \\
  \widetilde{M}_{o}'(1,1) \Big\rvert_{\mathcal{H}_{M}} 
  &:= \left( \bra{1} \otimes I_{2^{d_1}} \right) M_{a}^{(0)} 
       \left( \ket{1} \otimes I_{2^{d_1}} \right).
\end{align}

\noindent Then, \(M_{o}' := \widetilde{M}_{o}'(0,0) \Big\rvert_{\mathcal{H}_{M}} + \widetilde{M}_{o}'(1,1)\Big\rvert_{\mathcal{H}_{M}}\), since each block represents half of the total contribution, 
\begin{equation}
\frac{1}{2}M_{o}' + \frac{1}{2}M_{o}' = M_{o}'.
\end{equation}

\medskip

\noindent Recall the \(\mathsf{QOTP}\) decryption operation 
\(\mathsf{QOTPDec}(\cdot,k) : \mathcal{D}(\mathcal{H}_{M})\longrightarrow \mathcal{D}(\mathcal{H}_{M})\). 
For a key \(k = (k_1,k_2,\dots,k_{2d_1-1},k_{2d_1})\), the corresponding \(\mathsf{QOTP}\) \emph{encryption} 
is given by
\begin{equation}
  \mathsf{QOTPEnc}(M_o, k)
  \;=\;
  \Bigl(\!\bigotimes_{j=1}^{d_1} X^{k_{2j-1}}\,Z^{k_{2j}}\Bigr)\,M_o\,
  \Bigl(\!\bigotimes_{j=1}^{d_1} X^{k_{2j-1}}\,Z^{k_{2j}}\Bigr)^\dagger.
\end{equation}
Hence, \(\mathsf{QOTPDec}(\cdot,k)\) applies the adjoint of that unitary factor
\begin{equation}
  \mathsf{QOTPDec}(M', k)
  \;=\;
  \Bigl(\!\bigotimes_{j=1}^{d_1} X^{k_{2j-1}}\,Z^{k_{2j}}\Bigr)^\dagger
  \,M_{o}'\,
  \Bigl(\!\bigotimes_{j=1}^{d_1} X^{k_{2j-1}}\,Z^{k_{2j}}\Bigr).
\end{equation}
Accordingly, our final \text{normal} decryption for the original state is
\begin{equation}
  M_{o}
  \;=\;
  \mathsf{QOTPDec}\bigl(M_{o}',\,k\bigr)
  \;=\;
  \Bigl(\!\bigotimes_{j=1}^{d_1} X^{k_{2j-1}} Z^{k_{2j}}\Bigr)^\dagger
  \,M_{o}'\,
  \Bigl(\!\bigotimes_{j=1}^{d_1} X^{k_{2j-1}} Z^{k_{2j}}\Bigr).
\end{equation}
The output is precisely the original density matrix \(M_o\), completing the decryption for the original message.

\medskip

Next, we describe the extraction and decryption of the covert message. After computing \(U_{\sigma_{l}}^{\dagger}M_{f}^{(1)} U_{\sigma_{l}}\), the \emph{second} block or the \emph{third} block of \(M_{d}^{(1)}\) corresponds to \(\tfrac{1}{\eta}\,M_c''\) or its adjoint.

\medskip

Define the partial extraction operators
\[
\begin{aligned}
    \Pi_{0,1}(X) &= \left( \ket{0}\bra{0} \otimes I_{2^{d_1}} \right)(X)
    \left( \ket{1}\bra{1} \otimes I_{2^{d_1}} \right),\\
    \Pi_{1,0}(X) &= \left( \ket{1}\bra{1} \otimes I_{2^{d_1}} \right)(X)
    \left( \ket{0}\bra{0} \otimes I_{2^{d_1}} \right).
\end{aligned}
\]

\noindent Applying the partial operators to \(M_{d}^{(1)}\), we get
\begin{equation}
  M_{d}^{(1)}(0,1)
  \;=\Pi_{0,1} (M_d^{(1)})=\;
  \bigl(\ket{0}\bra{0}\otimes I_{2^{d_1}}\bigr)\,
    M_{d}^{(1)}\,
  \bigl(\ket{1}\bra{1}\otimes I_{2^{d_1}}\bigr) \quad \text{in} \quad \mathcal{D}((\mathbb{C}^{2})^{\otimes (d_{1}+1)})
\end{equation}
and
\begin{equation}
  M_{d}^{(1)}(1,0)
  \;=\Pi_{1,0} (M_d^{(1)})=\;
  \bigl(\ket{1}\bra{1}\otimes I_{2^{d_1}}\bigr)\,
    M_{d}^{(1)}\,
  \bigl(\ket{0}\bra{0}\otimes I_{2^{d_1}}\bigr)  \quad \text{in} \quad \mathcal{D}((\mathbb{C}^{2})^{\otimes (d_{1}+1)}).
\end{equation}

\noindent Now we reduce to the Hilbert space \(\mathcal{H}_{M}\) to extract the covert block

\begin{align}
  M_{\mathrm{covert}}(0,1) \big|_{\mathcal{H}_{M}}
  &= \left( \bra{0} \otimes I_{2^{d_1}} \right) M_d^{(1)}
     \left( \ket{1} \otimes I_{2^{d_1}} \right), \\
  M_{\mathrm{covert}}(1,0) \big|_{\mathcal{H}_{M}}
  &= \left( \bra{1} \otimes I_{2^{d_1}} \right) M_d^{(1)}
     \left( \ket{0} \otimes I_{2^{d_1}} \right).
\end{align}

\noindent We may choose either of these blocks and denote it by \(M_{\mathrm{covert}} \in \mathcal{L}((\mathbb{C}^2)^{\otimes d_1})\).

\noindent Recall that, at encryption, the covert block had a factor of \(\tfrac{1}{\eta}\).
Hence, to recover the padded covert operator \(M_c''\), we define
\begin{equation}
  \widetilde{M}_{c}''
  \;:=\; \eta\,M_{\mathrm{covert}}.
\end{equation}
If \(M_{\mathrm{covert}} = \tfrac{1}{\eta}\,M_c''\), then 
\(\widetilde{M}_{c}'' = M_c''\); or if 
\(M_{\mathrm{covert}} = \tfrac{1}{\eta}\,(M_c'')^\dagger\), then 
\(\widetilde{M}_{c}'' = (M_c'')^\dagger\).

\noindent Recall that 
\(\widetilde{M}_c'' \in \mathcal{L}\!\bigl(\mathcal{H}_{M}\bigr)\), 
but the covert message belongs in \(\mathcal{D}(\mathcal{H}_{M_{c}})\).  To recover the covert message in 
\(\mathcal{D}(\mathcal{H}_{M_{c}})\), we \emph{unembed} via the adjoint of \(V\). 

\medskip

\noindent Define
\begin{equation}
  \widetilde{M}_{c}'
  \;:=\;
  V^\dagger\, \widetilde{M}_{c}''\, V
  \;\in\;
  \mathcal{D}\bigl(\mathcal{H}_{M_{c}}\bigr).
\end{equation}
If \(\widetilde{M}_{c}'' = M_c''\), we get 
\(\widetilde{M}_c' = M_c'\).  In particular, since \(V^\dagger\) 
removes the zero-padding (the last \((2^{d_1} - 2^{d_2})\) rows and columns), 
\(\widetilde{M}_{c}'\) is the recovered padded covert operator 
in the original dimension \(2^{d_2}\times 2^{d_2}\). 

\medskip

\noindent The \(\mathsf{QOTP}\) \emph{decryption} for key 
\(k'=(k'_1, k'_2,\dots,k'_{2d_2-1},k'_{2d_2})\) is:
\begin{equation}
  \mathsf{QOTPDec}\bigl(M',\,k'\bigr)
  \;=\;
  \Bigl(\!\bigotimes_{j=1}^{d_2} X^{k'_{2j-1}}\,Z^{k'_{2j}}\Bigr)^\dagger
  \;M'\;
  \Bigl(\!\bigotimes_{j=1}^{d_2} X^{k'_{2j-1}}\,Z^{k'_{2j}}\Bigr).
\end{equation}
Hence the covert user obtains
\begin{equation}
  M_c
  \;=\;
  \mathsf{QOTPDec}\bigl(\widetilde{M}_c',\,k'\bigr)
  \;=\;
  \Bigl(\!\bigotimes_{j=1}^{d_2} X^{k'_{2j-1}}\,Z^{k'_{2j}}\Bigr)^\dagger 
    \;\widetilde{M}_c'\;
  \Bigl(\!\bigotimes_{j=1}^{d_2} X^{k'_{2j-1}}\,Z^{k'_{2j}}\Bigr).
\end{equation}
\noindent This final output is the \emph{covert} density operator \(M_c\) originally encrypted.

\medskip

Now, we describe the quantum algorithms based on quantum operations for the block constructions of $M_a$ and $M^{(1)}_f$. The following algorithms to extract the original quantum message(\(\mathsf{DOM}\)) and the covert message(\(\mathsf{DCM}\)) are described with standard quantum operations.

\smallskip
The control is a single-qubit register $\mathcal{H}_{R} \cong \mathbb{C}^2$. The Hilbert spaces $\mathcal{H}_{B}$ and $\mathcal{H}_{F}$ are finite-dimensional ancillary registers determined below by Stinespring dilations. Throughout, $M_o\in\mathcal{D}(\mathcal{H}_{M})$, $M_c\in\mathcal{D}(\mathcal{H}_{M_c})$, and the security parameter is $\secpar\in\mathbb{Z}^{+}$.
\noindent The quantum anamorphic encryption algorithm is described below:

\begin{algorithm}[H]
\caption{Quantum Anamorphic Encryption(\(\mathsf{QAE}\)) Part-1}
\label{alg:quantum-anamorphic-encryption}
\begin{algorithmic}[1]
\STATE \textbf{Input:} Original density matrix \(M_o \in \mathcal{D}(\mathcal{H}_{M})\) with \(M_{o} \succ 0\); where \(\mathcal{H}_{M}=(\mathbb{C}^{2})^{\otimes d_{1}}\), covert density matrix \(M_c \in \mathcal{D}(\mathcal{H}_{M_{c}})\); \(\mathcal{H}_{M_{c}}=(\mathbb{C}^{2})^{\otimes d_{2}}\), dimensions \(d_1, d_2\) (\(d_2 \leq d_1\)), security parameter \(\eta \in \mathbb{Z}^+\) satisfying the Theorem \ref{thm:main} condition, and permutation matrix \(U_{\sigma_l}\) on \(\mathcal{H}_{R} \otimes \mathcal{H}_{M}\), where \(\dim_{\mathbb{C}} \mathcal{H}_R=2\).

\STATE \textbf{Output:} Anamorphic quantum state \(M_f^{(1)}\in \mathcal{D}(\mathcal{H}_R\otimes \mathcal{H}_{M})\); and \((\mathcal{H}_R\otimes \mathcal{H}_{M})=\mathcal{H}_{C}\).

\STATE \textbf{Steps:}
\STATE \textbf{1. Sample keys:}
\STATE \hspace{1cm} Draw key \(k {\gets} \{0, 1\}^{2d_1}\) uniformly at random.
\STATE \hspace{1cm} Draw key \(k' {\gets} \{0, 1\}^{2d_2}\) uniformly at random.

\STATE \textbf{2. Encrypt \(M_o\) using \(\mathsf{QOTP}\):}
\STATE \hspace{1cm} Compute \(M_o' = \mathsf{QOTPEnc}(M_o, k)\) as:
\[
M_o' = \left( X^{k_1} Z^{k_2} \otimes \cdots \otimes X^{k_{2d_1-1}} Z^{k_{2d_1}} \right) M_o \left( X^{k_1} Z^{k_2} \otimes \cdots \otimes X^{k_{2d_1-1}} Z^{k_{2d_1}} \right)^\dagger.
\]

\STATE \textbf{3. Encrypt \(M_c\) using \(\mathsf{QOTP}\):}
\STATE \hspace{1cm} Compute \(M_c' = \mathsf{QOTPEnc}(M_c, k')\) as:
\[
M_c' = \left( X^{k'_1} Z^{k'_2} \otimes \cdots \otimes X^{k'_{2d_2-1}} Z^{k'_{2d_2}} \right) M_c \left( X^{k'_1} Z^{k'_2} \otimes \cdots \otimes X^{k'_{2d_2-1}} Z^{k'_{2d_2}} \right)^\dagger.
\]

\STATE \textbf{4. Pad \(M_c'\):}
\STATE \hspace{1cm} If \(d_2 < d_1\), extend \(M_c'\) to \(M_c'' \in \mathcal{D}((\mathbb{C}^2)^{\otimes d_1})\) using the isometric embedding:
\[
V : \mathcal{H}_{M_c} \longrightarrow \mathcal{H}_{M_c} \otimes (\mathbb{C}^2)^{\otimes (d_1 - d_2)},
\]
defined as:
\[
V |\psi\rangle = |\psi\rangle \otimes |0\rangle^{\otimes (d_1 - d_2)}, \quad \forall |\psi\rangle \in \mathcal{H}_{M_c}.
\]
The padded matrix is:
\[
M_c'' = V M_c' V^\dagger.
\]

%%\STATE \textbf{5. Construct \(M_a\):}
\begin{comment}
\STATE \textbf{Support smoothing for full rank.}
\STATE If $\textit{supp}(M_c'')\nsubseteq \textit{supp}(M_o')$, set
\[
M_o' \leftarrow (1-\delta)M_o' + \delta\,2^{-d_1}I,\qquad 0<\delta\ll 1,
\]
ensuring \(M_o'\) is full-rank. \emph{Henceforth, all objects (\(r,\Pi,\sqrt{M_o'}\), etc.) are computed for this updated \(M_o'\).}
\end{comment}

\STATE \textbf{Step 4: Spectral decomposition}
\STATE Compute eigendecomposition 
\[M_o' = \sum_{i=1}^r \lambda_i |\psi_i\rangle\langle\psi_i|; \; \forall i, \lambda_{i} > 0\;;\; r=\mathrm{rank}(M'_{o})=2^{d_1}\;.
\]
Set the support projector \(\Pi := \sum_{i=1}^{r}\ket{\psi_i}\!\bra{\psi_i}\).

\STATE \textbf{Step 5: Support projector}
\STATE Define 
\[\Pi = \sum_{i=1}^r |\psi_i\rangle\langle\psi_i|\]

\end{algorithmic}
\end{algorithm}

\begin{algorithm}[H]
\caption{Quantum Anamorphic Encryption(\(\mathsf{QAE}\)) Part-2}
\label{alg:qa-ske-enc-a-rigorous-part2}
\begin{algorithmic}[1]
\begin{comment}
\STATE If $\textit{supp}(M_c'')\nsubseteq \textit{supp}(M_o')$,
 
\STATE set
\[
M''_{o}=(1-\delta)M'_{o}+\delta\,2^{-d_1}I,\quad 0<\delta\ll 1,\quad \Pi:=I,\quad M'_{o}\leftarrow M''_{o}.
\]
\end{comment}

\STATE \textbf{Step 6: Canonical Purification of \(M'_{o}\).}
\STATE Define 
\[|\Phi\rangle = 2^{-d_1/2} \sum_{j=0}^{2^{d_1}-1} |j\rangle_M \otimes |j\rangle_B.\]

\STATE Compute \[|\varphi_{M_o'}\rangle = \sqrt{2^{d_1}}(I_M \otimes (\sqrt{M_o'})^{\mathsf{T}}) |\Phi\rangle\]
Then \(\|\ket{\varphi_{M_o'}}\|=1\).

\STATE \textbf{7. Construct a contraction on \(\textit{supp}(M'_{o})\).}

\STATE Compute $$(M_o')^{-1/2} = \sum_{i=1}^r \lambda_i^{-1/2} |\psi_i\rangle\langle\psi_i|\; \text{on}\;  \textit{supp}(M_o').$$  

\STATE Define the support-normalized operator:
\[
V_0 \;:=\; \Pi\,(M_o')^{-1/2}\; M_c''\; (M_o')^{-1/2}\,\Pi .
\]
\STATE Enforce contraction without trace renormalization:
\[
\kappa \;:=\; \|V_0\|_\infty,\quad \kappa_{\max}=\max\{1,\kappa\}, \quad
W_0 \;:=V_0/\kappa_{\max}; \quad (\text{so that}\quad \|W_0\|_\infty\le 1).\; 
\]

\STATE \textbf{8. Halmos unitary dilation of \(W_0^{\mathsf T}\).}

Let \(\mathcal{H}_F\cong\mathbb{C}^2\) be a \emph{single-qubit} ancilla with basis \(\{\ket{0}_F,\ket{1}_F\}\).
Define the unitary \(U_{BF}\) on \(\mathcal{H}_B\otimes \mathcal{H}_F\) by its \(2\times 2\) block form in the \(F\)-basis:
\[
U_{BF} \;=\;
\begin{pmatrix}
W_0^{\mathsf T} & \sqrt{\,I - W_0^{\mathsf T}(W_0^{\mathsf T})^\dagger\,} \\[.6ex]
\sqrt{\,I - (W_0^{\mathsf T})^\dagger W_0^{\mathsf T}\,} & - (W_0^{\mathsf T})^\dagger
\end{pmatrix},
\qquad
U_{BF}U_{BF}^\dagger=U_{BF}^\dagger U_{BF}=I.
\]
Consequently,
\(
\bra{0}_F\,U_{BF}\,\ket{0}_F \;=\; W_0^{\mathsf T}.
\)

\smallskip
\STATE \textbf{9. Prepare canonical purification.} 

\STATE Initialize control and dilation ancilla:
\[
|+\rangle_R = \tfrac{1}{\sqrt{2}}(|0\rangle_R+|1\rangle_R),\qquad \mathcal{H}_F \cong \mathbb{C}^2,\ \ |0\rangle_F:=|0\rangle.
\]
\STATE Create coherent superposition:
\[
|\Psi_{\mathrm{r}}\rangle 
\;:=\;
\tfrac{1}{\sqrt{2}}\Big(\,
|0\rangle_R \otimes |\varphi_{M_o'}\rangle \otimes |0\rangle_F
\;+\;
|1\rangle_R \otimes (I_M\otimes U_{BF})\big(|\varphi_{M_o'}\rangle\otimes |0\rangle_F\big)
\,\Big).
\]

\end{algorithmic}
\end{algorithm}

\begin{algorithm}[H]
\caption{Quantum Anamorphic Encryption(\(\mathsf{QAE}\)) Part-3}
\label{alg:qa-ske-enc-a-rigorous-part3}
\begin{algorithmic}[1]
\STATE \textbf{10. Reduce to \(RM\) and identify the off-diagonal.}

\STATE Trace out \(BF\):
\[
\Omega_{\mathrm r} \;:=\; \Tr_{BF}\!\big[\ket{\Psi_{\mathrm r}}\!\bra{\Psi_{\mathrm r}}\big]
\;=\;
\begin{pmatrix}
\tfrac{1}{2}M_o' & \tfrac{1}{2}X \\
\tfrac{1}{2}X^\dagger & \tfrac{1}{2}M_o'
\end{pmatrix}_{\!R},
\]
with
\[
X\;=\;\Tr_B\!\Big[\,(I\otimes \bra{0}_F U_{BF}\ket{0}_F)\,\ket{\varphi_{M_o'}}\!\bra{\varphi_{M_o'}}\,\Big]
\;=\;\Tr_B\!\big[(I\otimes W_0^{\mathsf T})\,\ket{\varphi_{M_o'}}\!\bra{\varphi_{M_o'}}\big].
\]
Using the identity in the Conventions and \(\Pi^2=\Pi\),
\[
X \;=\; \sqrt{M_o'}\,W_0\,\sqrt{M_o'}
\;=\; \frac{1}{\kappa_{\max}}\,\Pi\,M_c''\,\Pi.
\]

\STATE \textbf{11. Dephasing (pinching) on \(R\) to set the covert amplitude.}
\STATE Define the single-qubit dephasing channel on \(R\),
\[
\Deph(Y)\;:=\;\left(\frac{1+\lambda}{2}\right)\,Y \;+\; \left(\frac{1-\lambda}{2}\right)\, Z_R\,Y\,Z_R,
\qquad \lambda\in[-1,1].
\]
\STATE Apply \( (\Deph \otimes I_M)\) to \(\Omega_{\mathrm r}\):
\[
M_a \;:=\; (\Deph \otimes I_M)(\Omega_{\mathrm r})
\;=\;
\begin{pmatrix}
\tfrac{1}{2}M_o' & \tfrac{\lambda}{2}X \\
\tfrac{\lambda}{2}X^\dagger & \tfrac{1}{2}M_o'
\end{pmatrix}_{\!R}.
\]

\STATE Compute 
\[
\frac1{\eta^2}\;
  \bigl\|\,
    (M_c'')^\dagger\,(M_o')^{-1}\,M_c''
  \bigr\|
  \;\;\le\;\;
  \frac14\,\lambda_{\min}(M_o').
\]

and choose \(\eta\) such \(\eta \geq 2\,\kappa_{\max}\).
\STATE Choose
\[
\lambda \;=\; \frac{2\,\kappa_{\max}}{\eta}
\quad\text{(so \(|\lambda|\le 1\) whenever \(\eta\ge 2\,\kappa_{\max}\))},
\]
to obtain
\(
\tfrac{\lambda}{2}X \,=\, \dfrac{1}{\eta}\,\Pi M_c''\Pi.
\)
Thus, in the \(\{\ket{0},\ket{1}\}\) block form on \(R\),
\[
M_a \;=\;
\left[\ket{0}\bra{0}\!\otimes \tfrac{1}{2}M_o'
\;+\;
\ket{0}\bra{1}\!\otimes \frac{1}{\eta}\,\Pi M_c''\Pi
\;+\;
\ket{1}\bra{0}\!\otimes \frac{1}{\eta}\,(\Pi M_c''\Pi)^\dagger
\;+\;
\ket{1}\bra{1}\!\otimes \tfrac{1}{2}M_o'.
\right]\]
\noindent \emph{Note:} Since \(M_c''\) is Hermitian, \((\Pi M_c''\Pi)^\dagger=\Pi M_c''\Pi\).

\smallskip
\STATE \textbf{12. Apply a permutation unitary.}

\STATE Draw \(\sigma_l\) uniformly; let \(U_{\sigma_l}\) be the induced permutation unitary on \(\mathcal{H}_R\otimes \mathcal{H}_{M}\), and set
\[
M_f^{(1)} \;:=\; U_{\sigma_l}\, M_a\, U_{\sigma_l}^\dagger.
\]

\STATE \textbf{13. Output.} Return \(M_f^{(1)}\).
\end{algorithmic}
\end{algorithm}

\noindent\textbf{Remark (feasibility condition on $\eta$).}
The choice $\lambda=\tfrac{2\kappa_{\max}}{\eta}$ requires $\eta\ge 2\kappa_{\max}$ to keep $\lambda\le 1$.
Since $\kappa=\big\|\Pi (M_o')^{-1/2} M_c'' (M_o')^{-1/2}\Pi\big\|_\infty$ is known to the encryptor, one can always pick $\eta$ large enough (or equivalently use a smaller effective $\lambda$ after increasing the control pinching depth) to satisfy this bound and obtain the exact off-diagonal $\tfrac{1}{\eta}M_c''$.

\medskip

\begin{lemma}[Existence and basic properties of defect operators, \cite{SzNagyFoias}]
\label{lem:defect-existence}
Let \(C\) be a contraction on a finite-dimensional Hilbert space \(\mathcal{H}\), i.e.,
\(\|C\|_\infty\le 1\). Then \(I - C^\dagger C \succeq 0\) and \(I - C C^\dagger \succeq 0\).
Consequently, the unique positive square roots
\[
D_C := \sqrt{\,I - C^\dagger C\,}, \qquad
D_{C^\dagger} := \sqrt{\,I - C C^\dagger\,}
\]
exist. Moreover,
\[
D_C^2 = I - C^\dagger C, \qquad D_{C^\dagger}^2 = I - C C^\dagger.
\]
\end{lemma}

\begin{proof}
See Sz.-Nagy and Foias \cite[Chapter I, Sec. 3]{SzNagyFoias}, specifically the definition of defect operators. Since \(\|C\|_\infty\le 1\), we have \(C^\dagger C \preceq I\) and \(C C^\dagger \preceq I\).
Therefore \(I - C^\dagger C \succeq 0\) and \(I - C C^\dagger \succeq 0\). In finite dimension,
every positive operator has a unique positive square root, yielding \(D_C, D_{C^\dagger}\) with
the stated properties.
\end{proof}

\medskip
\begin{lemma}[Intertwining for defect operators (Sz.-Nagy identity) Page 6, Eq. (3.4), \cite{SzNagyFoias}]
\label{lem:intertwining}
Let \(C\) be a contraction and let \(D_C, D_{C^\dagger}\) be as in Lemma~\ref{lem:defect-existence}.
Then
\[
C\,D_C \;=\; D_{C^\dagger}\,C
\qquad\text{and}\qquad
D_C\,C^\dagger \;=\; C^\dagger\,D_{C^\dagger}.
\]
\end{lemma}

\begin{proof}
This is a standard identity in the theory of contractions; see Sz.-Nagy and Foias \cite[Chapter I, Sec. 3, Eq. (3.4)]{SzNagyFoias}. Consider polynomials \(p(t)=\sum_{n\ge0} a_n t^n\). For such \(p\),
\[
C\,(C^\dagger C)^n \;=\; (C C^\dagger)^n\,C
\quad\text{for all } n\in\mathbb{N},
\]
by a simple induction on \(n\). Hence for every polynomial \(p\),
\[
C\,p(C^\dagger C) \;=\; p(C C^\dagger)\,C.
\]
By Weierstrass approximation and spectral calculus, this identity extends
to continuous functions on \([0,1]\), and in particular to \(f(t):=\sqrt{1-t}\).
Thus
\[
C\,\sqrt{I - C^\dagger C}
\;=\;
\sqrt{I - C C^\dagger}\, C,
\]
which is exactly \(C\,D_C = D_{C^\dagger}\,C\).
Taking adjoints yields the second identity \(D_C\,C^\dagger = C^\dagger\,D_{C^\dagger}\).
\end{proof}

\medskip
\begin{lemma}[Transpose preserves operator norm, \cite{HornJohnson}]
\label{lem:norm-transpose}
For any matrix \(A\), \(\|A^{\mathsf T}\|_\infty = \|A\|_\infty\).
\end{lemma}

\begin{proof}
See Horn and Johnson \cite[Sec. 5.6]{HornJohnson}. Let \(A = U \Sigma V^\dagger\) be a singular value decomposition (SVD), with \(\Sigma \succeq 0\).
Then \(A^{\mathsf T} = \overline{U}\,\Sigma\,\overline{V}^\dagger\) is also an SVD with the same
singular values. Therefore \(\|A^{\mathsf T}\|_\infty = \|A\|_\infty\).
\end{proof}

\medskip
\begin{theorem}[Halmos unitary dilation is unitary~\cite{Halmos1950,paulsen2002}]
\label{thm:halmos-unitary}
Let \(C\) be a contraction on the Hilbert space \(\mathcal{H}_B\) (i.e., \(\|C\| \le 1\)). Define the defect operators
\[
D_C = \sqrt{I - C^\dagger C} \quad \text{and} \quad D_{C^\dagger} = \sqrt{I - C C^\dagger}.
\]
The block operator on \(\mathcal{H}_B \oplus \mathcal{H}_B\) defined by
\[
U \;:=\;
\begin{pmatrix}
C & D_{C^\dagger} \\[.6ex]
D_C & -\,C^\dagger
\end{pmatrix}
\]
is unitary.
\end{theorem}

\begin{proof}
Since \(C\) is a contraction, the operators \(I - C^\dagger C\) and \(I - C C^\dagger\) are positive semi-definite, ensuring that the defect operators \(D_C\) and \(D_{C^\dagger}\) are well-defined, Hermitian, and unique. To prove that \(U\) is unitary, we must verify that \(U^\dagger U = I \oplus I\). Note that the adjoint of \(U\) is given by
\[
U^\dagger = \begin{pmatrix}
C^\dagger & D_C \\
D_{C^\dagger} & -C
\end{pmatrix},
\]
where we have used the Hermiticity of the defect operators.

We rely on the fundamental intertwining identity for defect operators
\begin{equation}
    \label{eq:intertwining}
    C D_C = D_{C^\dagger} C.
\end{equation}
This follows from the observation that \(C(I - C^\dagger C) = (I - C C^\dagger)C\). By the functional calculus, this commutativity relation extends to the square root function, yielding Eq.~\eqref{eq:intertwining}. Taking the adjoint of Eq.~\eqref{eq:intertwining} yields the complementary relation \(D_C C^\dagger = C^\dagger D_{C^\dagger}\).

Computing the block matrix product \(U^\dagger U\), we obtain
\[
U^\dagger U \;=\;
\begin{pmatrix}
C^\dagger & D_C \\
D_{C^\dagger} & -C
\end{pmatrix}
\begin{pmatrix}
C & D_{C^\dagger} \\
D_C & -C^\dagger
\end{pmatrix}
\;=\;
\begin{pmatrix}
C^\dagger C + D_C^2 & C^\dagger D_{C^\dagger} - D_C C^\dagger \\
D_{C^\dagger} C - C D_C & D_{C^\dagger}^2 + C C^\dagger
\end{pmatrix}.
\]
The diagonal blocks simplify to the identity using the definitions of the defect operators: \(C^\dagger C + D_C^2 = C^\dagger C + (I - C^\dagger C) = I\), and similarly for the bottom-right block. The off-diagonal blocks vanish as a direct consequence of the intertwining identities derived above. Specifically, \(C^\dagger D_{C^\dagger} - D_C C^\dagger = 0\) and \(D_{C^\dagger} C - C D_C = 0\). Thus \(U^\dagger U = I\). A symmetric argument shows that \(U U^\dagger = I\), concluding the proof.
\end{proof}

\smallskip
\noindent The \emph{vectorization operator} $\mathrm{vec}(X)$ for a matrix $X \in \mathbb{C}^{m \times n}$ stacks the columns of $X$ into a single column vector in $\mathbb{C}^{mn}$. For a matrix $X = [x_1 \mid x_2 \mid \cdots \mid x_n] \in \mathbb{C}^{m \times n}$ with columns $x_j \in \mathbb{C}^m$, the vectorization operator is defined as
\[
\mathrm{vec}(X) = \begin{bmatrix} x_1 \\ x_2 \\ \vdots \\ x_n \end{bmatrix} \in \mathbb{C}^{mn}.
\]
\begin{lemma}[\cite{Watrous2018}]\label{lem:vec-trace}
For matrices \(A,B,X\) of compatible sizes:
\begin{align*}
\mathrm{(i)}\quad & \mathrm{vec}(A X B^{\mathsf T}) \;=\; (B \otimes A)\,\mathrm{vec}(X). \\
\mathrm{(ii)}\quad & \Tr_2\!\big[\,|\mathrm{vec}(A)\rangle\langle \mathrm{vec}(B)|\,\big] \;=\; A\,B^\dagger.
\end{align*}
\end{lemma}

\begin{proof}
See Watrous \cite{Watrous2018} (Sec. 2.1.2) for the vectorization identity (i), and \cite{Watrous2018} (Sec. 2.2, Eq. 2.63) for the partial trace identity (ii). Note that Watrous defines the partial trace over the second system $\mathcal{Y}$ as $\Tr_{\mathcal{Y}}(\text{vec}(A)\text{vec}(B)^*) = A B^*$, which corresponds to our notation, taking into account the conjugate transpose. 
\end{proof}

\smallskip
\begin{lemma}[Purification trace identity used in \(\mathsf{QAE}\)]
\label{lem:purif-trace-identity}
Let \(M\succeq 0\) be a density on \(\mathcal{H}_M\) and let
\(|\varphi_M\rangle := \sqrt{d}\,(I\otimes (\sqrt{M})^{\mathsf T})|\Phi\rangle\) with \(d=\dim\mathcal{H}_M\).
Then for any operator \(Y\) on \(\mathcal{H}_B\),
\[
\Tr_B\!\big[(I\otimes Y^{\mathsf T})\,|\varphi_M\rangle\langle \varphi_M|\,\big]
\;=\; \sqrt{M}\,Y\,\sqrt{M}.
\]
\end{lemma}

\begin{proof}
By Lemma~\ref{lem:vec-trace}(i) with \(X=I\),
\(|\varphi_M\rangle = \mathrm{vec}(\sqrt{M})\).
Hence
\[
(I\otimes Y^{\mathsf T})\,|\varphi_M\rangle
= (I\otimes Y^{\mathsf T})\,\mathrm{vec}(\sqrt{M})
= \mathrm{vec}(\sqrt{M}\,Y),
\]
again by Lemma~\ref{lem:vec-trace}(i).
Applying Lemma~\ref{lem:vec-trace}(ii),
\[
\Tr_B\!\big[(I\otimes Y^{\mathsf T})\,|\varphi_M\rangle\langle \varphi_M|\big]
= \Tr_B\!\big[\,|\mathrm{vec}(\sqrt{M}Y)\rangle\langle \mathrm{vec}(\sqrt{M})|\,\big]
= (\sqrt{M}Y)\,(\sqrt{M})^\dagger
= \sqrt{M}\,Y\,\sqrt{M},
\]
since \(\sqrt{M}\) is positive and hence Hermitian.
\end{proof}

\medskip
\begin{theorem}[Correctness of the off-diagonal block in the \(\mathsf{QAE}\) embedding]
\label{thm:offdiag-correctness}
In the \(\mathsf{QAE}\) construction, let \(W_0\) be the support-normalized operator
\(W_0 = V_0/\kappa_{\max}\) with
\[
V_0 \;:=\; \Pi\,(M_o')^{-1/2}\,M_c''\,(M_o')^{-1/2}\,\Pi,
\qquad
\kappa_{\max} := \max\{1,\,\|V_0\|_\infty\},
\]
so that \(\|W_0\|_\infty\le 1\).
Set \(C:=W_0^{\mathsf T}\) (a contraction by Lemma~\ref{lem:norm-transpose})
and let \(U\) be the Halmos unitary from Theorem~\ref{thm:halmos-unitary}.
Prepare
\[
\ket{\Psi_{\mathrm r}} \;=\; \tfrac{1}{\sqrt{2}}\Big(
\ket{0}_R\otimes \ket{\varphi_{M_o'}}\otimes \ket{0}_F \;+\;
\ket{1}_R\otimes (I_M\otimes U)\big(\ket{\varphi_{M_o'}}\otimes \ket{0}_F\big)\Big),
\]
and define the reduced operator on \(R\!M\) by \(\Omega_{\mathrm r}:=\Tr_{BF}[\,|\Psi_{\mathrm r}\rangle\langle \Psi_{\mathrm r}|\,]\).
Then \(\Omega_{\mathrm r}\) has the \(2\times 2\) block form (with respect to \(R\))
\[
\Omega_{\mathrm r}
\;=\;
\begin{pmatrix}
\tfrac{1}{2}M_o' & \tfrac{1}{2}X \\
\tfrac{1}{2}X^\dagger & \tfrac{1}{2}M_o'
\end{pmatrix},
\qquad
X \;=\; \sqrt{M_o'}\,W_0\,\sqrt{M_o'} \;=\; \frac{1}{\kappa_{\max}}\,\Pi\,M_c''\,\Pi.
\]
\end{theorem}

\begin{proof}
Expanding \(|\Psi_{\mathrm r}\rangle\langle \Psi_{\mathrm r}|\) and tracing out \(BF\) yields four blocks on \(R\).
The diagonal blocks are equal by construction and satisfy
\[
\Tr_{BF}\!\big[\,|\varphi_{M_o'}\rangle\langle \varphi_{M_o'}|\otimes |0\rangle\langle 0|_F\,\big]
\;=\; M_o',
\qquad
\Tr_{BF}\!\big[\,(I\otimes U)(|\varphi_{M_o'}\rangle\langle \varphi_{M_o'}|\otimes |0\rangle\langle 0|_F)(I\otimes U)^\dagger\,\big]
\;=\; M_o',
\]
by unitarity of \(U\) and invariance of the partial trace, hence each diagonal block equals \(\tfrac{1}{2}M_o'\).

For the off-diagonal block, using \(\bra{0}_F U \ket{0}_F = C = W_0^{\mathsf T}\) (Theorem~\ref{thm:halmos-unitary})
and Lemma~\ref{lem:purif-trace-identity} with \(Y:=W_0\), we obtain
\[
X \;=\; \Tr_B\!\big[(I\otimes \bra{0}_F U \ket{0}_F)\,|\varphi_{M_o'}\rangle\langle \varphi_{M_o'}|\big]
\;=\; \Tr_B\!\big[(I\otimes W_0^{\mathsf T})\,|\varphi_{M_o'}\rangle\langle \varphi_{M_o'}|\big]
\;=\; \sqrt{M_o'}\,W_0\,\sqrt{M_o'}.
\]
By the definition of \(W_0\) and the support projector \(\Pi\),
\[
\sqrt{M_o'}\,W_0\,\sqrt{M_o'}
= \sqrt{M_o'}\Big(\frac{1}{\kappa_{\max}}\,\Pi\,(M_o')^{-1/2}\,M_c''\,(M_o')^{-1/2}\,\Pi\Big)\sqrt{M_o'}
= \frac{1}{\kappa_{\max}}\,\Pi\,M_c''\,\Pi.
\]
This proves the claim.
\end{proof}

%======================================================
\begin{theorem}\label{thm:dephasing-correctness}
Let \(\Deph\) be the single-qubit dephasing channel on \(R\) given by
\[
\Deph(Y) \;=\; \left(\frac{1+\lambda}{2}\right)\,Y \;+\; \left(\frac{1-\lambda}{2}\right)\, Z_R\,Y\,Z_R,
\qquad \lambda\in[-1,1].
\]
Apply \((\Deph \otimes I_M)\) to \(\Omega_{\mathrm r}\) of Theorem~\ref{thm:offdiag-correctness} and set
\(M_a := (\Deph \otimes I_M)(\Omega_{\mathrm r})\).
Then
\[
M_a \;=\;
\begin{pmatrix}
\tfrac{1}{2}M_o' & \tfrac{\lambda}{2}\,X \\
\tfrac{\lambda}{2}\,X^\dagger & \tfrac{1}{2}M_o'
\end{pmatrix}
\quad\text{and, with }\;
\lambda \;=\; \frac{2\,\kappa_{\max}}{\eta}\ (\;|\lambda|\le 1\ \text{whenever }\eta\ge 2\,\kappa_{\max}\;),
\]
its off-diagonal equals
\(
\tfrac{\lambda}{2}X \,=\, \dfrac{1}{\eta}\,\Pi M_c''\Pi.
\)
\end{theorem}

\begin{proof}
The channel \(\Deph\) leaves the diagonal blocks of a \(2\times 2\) block operator (in the \(R\)-basis) invariant
and multiplies the off-diagonals by \(\lambda\). Hence
\[
M_a
= \begin{pmatrix}
\tfrac{1}{2}M_o' & \tfrac{\lambda}{2}X \\
\tfrac{\lambda}{2}X^\dagger & \tfrac{1}{2}M_o'
\end{pmatrix}.
\]
By Theorem~\ref{thm:offdiag-correctness}, \(X = \frac{1}{\kappa_{\max}}\Pi M_c''\Pi\).
Choosing \(\lambda = 2\,\kappa_{\max}/\eta\) gives
\(
\tfrac{\lambda}{2}X = \frac{1}{\eta}\,\Pi M_c''\Pi.
\)
Trace preservation and complete positivity follow from the CPTP nature of \(\Deph\).
\end{proof}

%======================================================
\begin{corollary}\label{cor:permutation}
Let \(U_{\sigma_l}\) be any permutation unitary on \(\mathcal{H}_R\otimes\mathcal{H}_{M}\), and set
\(M_f^{(1)} := U_{\sigma_l}\,M_a\,U_{\sigma_l}^\dagger\).
Then \(M_f^{(1)}\) is a valid density operator with the same spectrum as \(M_a\).
\end{corollary}

\begin{proof}
Unitary conjugation preserves positivity, trace, and spectrum.
\end{proof}

\begin{comment}
\begin{remark}[On the support projector $\Pi$]
If $\textit{supp}(M_c'')\subseteq \textit{supp}(M_o')$, then the projector $\Pi$ acts trivially on $M_c''$, and the off-diagonal block is simply $\frac{1}{\eta}M_c''$. Furthermore, the support smoothing step $M_o'\leftarrow (1-\delta)M_o'+\delta\,2^{-d_1}I$ with $0<\delta\ll 1$ makes $M_o'$ full-rank, with $\Pi=I$. In this case, the off-diagonal block is again given exactly by $\frac{1}{\eta}M_c''$.
\end{remark}
\end{comment}

\begin{theorem}\label{thm:qae-correctness-summary}
Under the \(\mathsf{QAE}\) construction with \(W_0\) a contraction on \(\textit{supp}(M_o')\), the Halmos dilation \(U\)
of \(C=W_0^{\mathsf T}\) is unitary (Theorem~\ref{thm:halmos-unitary}). The reduced state \(\Omega_{\mathrm r}\) on \(R\!M\)
has diagonal blocks \(\tfrac{1}{2}M_o'\) and off-diagonal block
\(X=\sqrt{M_o'}\,W_0\,\sqrt{M_o'}=\frac{1}{\kappa_{\max}}\Pi M_c''\Pi\) (Theorem~\ref{thm:offdiag-correctness}).
Subsequent dephasing with \(\lambda=2\kappa_{\max}/\eta\) yields
\[
M_a
\;=\;
\left[\ket{0}\bra{0}\otimes \tfrac{1}{2}M_o'
\;+\;
\ket{0}\bra{1}\otimes \frac{1}{\eta}\,\Pi M_c''\Pi
\;+\;
\ket{1}\bra{0}\otimes \frac{1}{\eta}\,(\Pi M_c''\Pi)^\dagger
\;+\;
\ket{1}\bra{1}\otimes \tfrac{1}{2}M_o'\right],
\]
which is a valid density operator. Finally, the permutation unitary \(U_{\sigma_l}\) produces
\(M_f^{(1)}=U_{\sigma_l}M_aU_{\sigma_l}^\dagger\), a valid anamorphic ciphertext
(Corollary~\ref{cor:permutation}).
\end{theorem}

\medskip
Now we describe the decryption algorithm to extract and decrypt the original secret from the anamorphic ciphertext:

\begin{algorithm}[H]
\caption{Decryption of Original Secret from Anamorphic Ciphertext(\(\mathsf{DOM}\))}
\label{alg:public-decryption}
\begin{algorithmic}[1]
\STATE \textbf{Input:} Anamorphic state \(M_f^{(1)} \in \mathcal{D}((\mathbb{C}^2)^{\otimes(d_1+1)})\), permutation unitary \(U_{\sigma_l}\), \(\mathsf{QOTP}\) key \(k \in \{0, 1\}^{2d_1}\), dimension \(d_1\).

\STATE \textbf{Output:} Original density matrix \(M_o \in \mathcal{D}(\mathcal{H}_{M})\).

\STATE \textbf{Steps:}

\STATE \textbf{1. Apply the inverse permutation:}
\STATE \hspace{1cm} Compute the intermediate state \(M_d^{(1)}\) by applying the inverse of \(U_{\sigma_l}\):
\[
M_d^{(1)} = U_{\sigma_l}^\dagger M_f^{(1)} U_{\sigma_l}.
\]

\begin{comment}
\STATE \textbf{2. Extract the first block and fourth block:}
\STATE \hspace{1cm} Define the Kraus operators:
\[
\Pi_0 = \ket{0}\bra{0} \otimes I_{2^{d_1}}, \qquad \Pi_1 = \ket{1}\bra{1} \otimes I_{2^{d_1}},
\]
where \(I_{2^{d_1}}\) is the identity operator on \(\mathcal{H}_{M_{o}} = (\mathbb{C}^2)^{\otimes d_1}\).
\end{comment}
\smallskip
\STATE \textbf{Quantum Dephasing Channel on Control Qubit} 
\STATE Implement the dephasing channel $\Gamma_0: \mathcal{L}(\mathcal{H}_R) \longrightarrow \mathcal{L}(\mathcal{H}_R)$:
\[
\Gamma_0(\rho_R) = \frac{1}{2}\rho_R + \frac{1}{2} Z_R \rho_R Z_R
\]
and set
\begin{equation}\label{eq:after-pinch}
M_a^{(0)}\ :=\ (\Gamma_0\otimes I_{2^{d_{1}}})\big(M_d^{(1)}\big).
\end{equation}
\begin{comment}
\STATE \textbf{Block extraction and reduction to $\mathcal{H}_{M}$} 

\STATE For $b\in\{0,1\}$ define the reduced blocks
\begin{equation}\label{eq:reduce-blocks}
\widetilde{M}_o'(b,b)\ :=\ (\bra{b}\otimes I_{2^{d_{1}}})\ M_a^{\mathrm{diag}}\ (\ket{b}\otimes I_{2^{d_{1}}})\ \in\ \mathcal{D}(\mathcal{H}_{M}).
\end{equation}
\STATE Set
\begin{equation}\label{eq:combine}
M_o'\ :=\ \widetilde{M}_o'(0,0)\ +\ \widetilde{M}_o'(1,1).
\end{equation}
\end{comment}
\STATE \textbf{Step 3: Quantum Partial Trace}
\STATE Trace out the control register \(R\):

\begin{equation}\label{eq:traceR}
\Tr_{R}\!\big[M_a^{(0)}\big]\ =\ M_o'.
\end{equation}

\STATE \textbf{5. Apply \(\mathsf{QOTP}\) Decryption:}
\STATE \hspace{1cm} Use the \(\mathsf{QOTP}\) decryption key \(k = (k_1, k_2, \dots, k_{2d_1})\) to recover \(M_o\):
\[
M_o = \mathsf{QOTPDec}(M_o', k)
= \Bigl(\bigotimes_{j=1}^{d_1} X^{k_{2j-1}} Z^{k_{2j}} \Bigr)^\dagger
M_o'
\Bigl(\bigotimes_{j=1}^{d_1} X^{k_{2j-1}} Z^{k_{2j}} \Bigr).
\]

\STATE \textbf{6. Return the original density matrix \(M_o\):}
\STATE \hspace{1cm} Output \(M_o\), completing the decryption.

\end{algorithmic}
\end{algorithm}

\medskip
\begin{theorem}[Correctness of \(\mathsf{DOM}\) Algorithm]
\label{thm:DOM-correctness}
Let $M_f^{(1)} \in \mathcal{D}(\mathcal{H}_R \otimes \mathcal{H}_M)$ be an anamorphic ciphertext satisfying:
\[
M_f^{(1)} = \Uperm M_a \Uperm^\dagger
\]
where $\Uperm$ is a permutation unitary and $M_a$ has the canonical block form
\[
M_a = \left[ \ket{0}\bra{0} \otimes \tfrac{1}{2}M_o' + \ket{1}\bra{1} \otimes \tfrac{1}{2}M_o' + \ket{0}\bra{1} \otimes \tfrac{1}{\eta}M_c'' + \ket{1}\bra{0} \otimes \tfrac{1}{\eta}(M_c'')^\dagger \right]
\]
with $M_o' = U_k M_o U_k^\dagger$ for some key $k \in \{0,1\}^{2d_1}$, \(U_{k}=\mathsf{QOTP}(\cdot,k)\) and security parameter $\eta \in \mathbb{Z}^+$. Then the DOM algorithm exactly recovers $M_o$.
\end{theorem}

\begin{proof}
Applying the inverse permutation unitary
\[
M_d^{(1)} = \Uperm^\dagger M_f^{(1)} \Uperm = \Uperm^\dagger (\Uperm M_a \Uperm^\dagger) \Uperm = M_a
\]
Thus, we recover the unpermuted anamorphic state
\begin{equation}
\label{eq:step1}
M_d^{(1)} = M_a = \begin{pmatrix}
\frac{1}{2}M_o' & \frac{1}{\eta}M_c'' \\
\frac{1}{\eta}(M_c'')^\dagger & \frac{1}{2}M_o'
\end{pmatrix}
\end{equation}

\noindent Let $\Gamma_0: \mathcal{L}(\mathcal{H}_R) \longrightarrow \mathcal{L}(\mathcal{H}_R)$ be the dephasing channel defined by
\[
\Gamma_0(\rho_R) = \frac{1}{2}\rho_R + \frac{1}{2} Z_R \rho_R Z_R
\]
This channel has Kraus operators $K_0 = \frac{1}{\sqrt{2}}I_R$ and $K_1 = \frac{1}{\sqrt{2}}Z_R$, satisfying the completeness relation:
\[
K_0^\dagger K_0 + K_1^\dagger K_1 = \frac{1}{2}I_R + \frac{1}{2}Z_R^\dagger Z_R = \frac{1}{2}I_R + \frac{1}{2}I_R = I_R
\]

\noindent The tensor product channel $(\Gamma_0 \otimes I_M)$ acts on the composite system as
\begin{align*}
(\Gamma_0 \otimes I_M)(M_d^{(1)}) 
&= \sum_{i=0}^1 (K_i \otimes I_M) M_d^{(1)} (K_i \otimes I_M)^\dagger \\
&= \frac{1}{2} M_d^{(1)} + \frac{1}{2} (Z_R \otimes I_M) M_d^{(1)} (Z_R \otimes I_M)
\end{align*}

\noindent We compute each term explicitly. First, note that
\[
Z_R \otimes I_M = \begin{pmatrix} 1 & 0 \\ 0 & -1 \end{pmatrix} \otimes I_M = \begin{pmatrix} I_M & 0 \\ 0 & -I_M \end{pmatrix}
\]

Then
\begin{align*}
(Z_R \otimes I_M) M_d^{(1)} (Z_R \otimes I_M) 
&= \begin{pmatrix} I_M & 0 \\ 0 & -I_M \end{pmatrix}
\begin{pmatrix}
\frac{1}{2}M_o' & \frac{1}{\eta}M_c'' \\
\frac{1}{\eta}(M_c'')^\dagger & \frac{1}{2}M_o'
\end{pmatrix}
\begin{pmatrix} I_M & 0 \\ 0 & -I_M \end{pmatrix} \\
&= \begin{pmatrix} I_M & 0 \\ 0 & -I_M \end{pmatrix}
\begin{pmatrix}
\frac{1}{2}M_o' & -\frac{1}{\eta}M_c'' \\
\frac{1}{\eta}(M_c'')^\dagger & -\frac{1}{2}M_o'
\end{pmatrix} \\
&= \begin{pmatrix}
\frac{1}{2}M_o' & -\frac{1}{\eta}M_c'' \\
-\frac{1}{\eta}(M_c'')^\dagger & \frac{1}{2}M_o'
\end{pmatrix}
\end{align*}

Therefore
\begin{align*}
M_a^{(0)} &= (\Gamma_0 \otimes I_M)(M_d^{(1)}) \\
&= \frac{1}{2} \begin{pmatrix}
\frac{1}{2}M_o' & \frac{1}{\eta}M_c'' \\
\frac{1}{\eta}(M_c'')^\dagger & \frac{1}{2}M_o'
\end{pmatrix}
+ \frac{1}{2} \begin{pmatrix}
\frac{1}{2}M_o' & -\frac{1}{\eta}M_c'' \\
-\frac{1}{\eta}(M_c'')^\dagger & \frac{1}{2}M_o'
\end{pmatrix} \\
&= \begin{pmatrix}
\frac{1}{4}M_o' + \frac{1}{4}M_o' & \frac{1}{2\eta}M_c'' - \frac{1}{2\eta}M_c'' \\
\frac{1}{2\eta}(M_c'')^\dagger - \frac{1}{2\eta}(M_c'')^\dagger & \frac{1}{4}M_o' + \frac{1}{4}M_o'
\end{pmatrix} \\
&= \begin{pmatrix}
\frac{1}{2}M_o' & 0 \\
0 & \frac{1}{2}M_o'
\end{pmatrix}
\end{align*}

\noindent Thus, the dephasing channel annihilates the off-diagonal blocks
\begin{equation}
\label{eq:step2}
M_a^{(0)} = \ket{0}\bra{0} \otimes \tfrac{1}{2}M_o' + \ket{1}\bra{1} \otimes \tfrac{1}{2}M_o'
\end{equation}

\smallskip
\noindent The partial trace over the control register $R$ is computed as
\[
\Tr_R[M_a^{(0)}] = \Tr_R\left[ \ket{0}\bra{0} \otimes \tfrac{1}{2}M_o' + \ket{1}\bra{1} \otimes \tfrac{1}{2}M_o' \right]
\]

\noindent Using the linearity of the partial trace and the fact that $\Tr[\ket{b}\bra{b}] = 1$ for $b \in \{0,1\}$, we get
\begin{align*}
\Tr_R[M_a^{(0)}] &= \Tr[\ket{0}\bra{0}] \cdot \tfrac{1}{2}M_o' + \Tr[\ket{1}\bra{1}] \cdot \tfrac{1}{2}M_o' \\
&= 1 \cdot \tfrac{1}{2}M_o' + 1 \cdot \tfrac{1}{2}M_o' \\
&= M_o'
\end{align*}

\noindent By definition, $M_o' = U_{k} M_o U_{k}^\dagger$, where
\(
U_{k}=\mathsf{QOTP}(\cdot,k).
\)
Applying the inverse QOTP unitary, we get
\(
U_{k}^\dagger M_o' U_{k}= M_o.
\)
Thus, the original message is exactly recovered
\begin{equation}
\label{eq:step4}
\mathsf{QOTPDec}(M_o', k) = M_{o}. 
\end{equation}

Combining equations (\ref{eq:step1}) through (\ref{eq:step4}), we have shown that the DOM algorithm transforms $M_f^{(1)}$ to $M_o$ through a sequence of mathematically exact operations, completing the proof.
\end{proof}

\begin{lemma}[Independence from Security Parameter]
\label{lem:eta-independence}
The \(\mathsf{DOM}\) algorithm's recovery of $M_o$ is independent of the security parameter $\eta$.
\end{lemma}

\begin{proof}
From equation (\ref{eq:step2}), the dephasing channel $\Gamma_0$ completely eliminates the off-diagonal terms containing $\frac{1}{\eta}M_c''$ and $\frac{1}{\eta}(M_c'')^\dagger$. The remaining diagonal terms $\frac{1}{2}M_o'$ are independent of $\eta$. The subsequent partial trace and QOTP decryption operations do not introduce any $\eta$-dependence. Therefore, the final output $M_o$ is independent of $\eta$.
\end{proof}

\begin{theorem}[CPTP Composition]
\label{thm:DOM-CPTP}
The \(\mathsf{DOM}\) algorithm is a completely positive trace-preserving (CPTP) map.
\end{theorem}

\begin{proof}
Each step of the DOM algorithm is CPTP:

\begin{enumerate}
\item  \textbf{Inverse permutation}: Unitary conjugation by $\Uperm^\dagger$ is CPTP.

\item \textbf{Dephasing channel}: $(\Gamma_0 \otimes I_M)$ is CPTP, as shown by its Kraus operator representation:
\[
(\Gamma_0 \otimes I_M)(\rho) = \sum_{i=0}^1 (K_i \otimes I_M) \rho (K_i \otimes I_M)^\dagger
\]
with $\sum_i (K_i \otimes I_M)^\dagger (K_i \otimes I_M) = I_R \otimes I_M$.

\item \textbf{Partial trace}: The operation $\Tr_R[\cdot]$ is CPTP.

\item \textbf{QOTP decryption}: Unitary conjugation by $U_{k}=\mathsf{QOTP}(\cdot,k)^\dagger$ is CPTP.
\end{enumerate}
The composition of CPTP maps is CPTP. Therefore, the entire DOM algorithm consists of CPTP maps.
\end{proof}

\noindent The covert secret extraction and decryption algorithm from the anamorphic ciphertext, is described below:

\begin{algorithm}[H]
\caption{Decryption of Covert Message from Anamorphic Ciphertext (\(\mathsf{DCM}\))}
\label{alg:covert-decryption-rigorous-physical}
\begin{algorithmic}[1]

\STATE \textbf{Input:}
\begin{itemize}
\item Anamorphic state \(M_f^{(1)} \in \mathcal{D}\!\big((\mathbb{C}^2)^{\otimes(d_1+1)}\big)\) produced by encryption.
\item Permutation unitary \(U_{\sigma_l}\) used at encryption.
\item One-time-pad key \(k'=(k'_1,\dots,k'_{2d_2})\in\{0,1\}^{2d_2}\) for the covert message.
\item Dimensions \(d_1\ge d_2\) and scaling parameter \(\eta\in\mathbb{Z}^{+}\).
\item Isometry \(V:\mathcal{H}_{M_c} \to \mathcal{H}_M\) with \(V\lvert\psi\rangle=\lvert\psi\rangle\otimes\lvert 0\rangle^{\otimes(d_1-d_2)}\),
\(\mathcal{H}_M=(\mathbb{C}^2)^{\otimes d_1}\), \(\mathcal{H}_{M_c}=(\mathbb{C}^2)^{\otimes d_2}\).
\item Total number of copies (shots) for the probe: \(N_X\in\mathbb{N}\).
\end{itemize}

\STATE \textbf{Output:} Covert plaintext density matrix \(M_c\in\mathcal{D}(\mathcal{H}_{M_c})\).

\STATE \textbf{Registers and gates:}
Let \(R\) denote the single-qubit control register; \(M\) the \(d_1\)-qubit register.
On \(R\) we use the Pauli operators \(X_R,Y_R,Z_R\) and the Hadamard \(H_R\).
Define the computational-basis dephasing (pinching) channel on \(R\),
\[
\mathcal{P}_Z^{(R)}(X)\;=\;\tfrac{1}{2}(X+Z_R X Z_R)
\;=\;\sum_{b\in\{0,1\}}(\lvert b\rangle\langle b\rvert_R\otimes I_M)\,X\,(\lvert b\rangle\langle b\rvert_R\otimes I_M).
\]
\emph{Physical realization:} \(\mathcal{P}_Z^{(R)}\) is implemented by measuring \(R\) in the computational basis and discarding the classical outcome (Lemma~\ref{lem:dephasing-implementation}).

\STATE \textbf{1. Undo the permutation unitary (one copy for illustration).}
Apply \(U_{\sigma_l}^\dagger\) on \(R\otimes M\) to obtain
\[
M_d^{(1)} \;:=\; U_{\sigma_l}^\dagger\, M_f^{(1)}\, U_{\sigma_l}.
\]
By encryption correctness and Hermiticity of the covert block (Prop.~\ref{prop:B-hermitian}),
\begin{equation}\label{eq:block-form-Md-physical}
M_d^{(1)} \;=\;
\begin{pmatrix}
\tfrac12 M_o' & \tfrac{1}{\eta}\,B \\[0.25em]
\tfrac{1}{\eta}\,B & \tfrac12 M_o'
\end{pmatrix}_{\!R},
\qquad B \;=\; \Pi\,M_c''\,\Pi,\;\; M_c''=V M_c' V^\dagger,\;\; B=B^\dagger.
\end{equation}

\STATE \textbf{2. Real part via explicit measurements (X-probe; \(N_X\) copies).}
For each shot \(t=1,\dots,N_X\) (on an independent fresh copy of \(M_f^{(1)}\)):
\begin{enumerate}
\item Apply \(U_{\sigma_l}^\dagger\) on \(R\otimes M\).
\item Apply \(H_R\) on \(R\).
\item Measure \(R\) in the computational \(Z\) basis, recording \(b_t\in\{0,1\}\).
\item Perform an informationally-complete Pauli measurement on \(M\): draw a Pauli string \(P_t\in\{I,X,Y,Z\}^{\otimes d_1}\) from a fixed design with known probabilities \(\pi(P)=\Pr\{P_t=P\}\), rotate each qubit of \(M\) to the eigenbasis of its factor in \(P_t\), and measure in \(Z\), recording the eigenvalue \(m_t\in\{\pm 1\}\) (the outcome of \(P_t\)).
\end{enumerate}

\end{algorithmic}
\end{algorithm}

\begin{algorithm}[H]
\caption{Decryption of Covert Message from Anamorphic Ciphertext (\(\mathsf{DCM}\)) Part 2 (Continued)}
\begin{algorithmic}[1]

\STATE \textbf{3. Conditional linear inversion and unnormalized blocks.}
Let \(S_X(b)=\{t\le N_X:\,b_t=b\}\), \(n_X(b)=|S_X(b)|\). For each Pauli string \(P\), define the Horvitz–Thompson estimator
\[
\widehat{\mu}^{(X)}_{b}(P) \;=\;
\frac{1}{n_X(b)}\sum_{t\in S_X(b)} \frac{\mathbf{1}_{\{P_t=P\}}\; m_t}{\pi(P)}.
\]
Define the conditional linear-inversion estimator on \(M\) given \(R=b\):
\[
\widehat{\rho}^{(X)}_{M\mid b} \;=\; 
\frac{1}{2^{d_1}}\sum_{P\in\{I,X,Y,Z\}^{\otimes d_1}} \widehat{\mu}^{(X)}_{b}(P)\; P,
\qquad
\widehat{p}^{(X)}_b \;=\; \frac{n_X(b)}{N_X}.
\]
Set the \emph{unnormalized} block estimators
\[
\widehat{D}_b \;:=\; \widehat{p}^{(X)}_b\,\widehat{\rho}^{(X)}_{M\mid b}
\quad (b=0,1).
\]
(Under ideal dynamics, \(\widehat{D}_0\) estimates \(\tfrac12 M_o'+\tfrac{1}{\eta}B\) and \(\widehat{D}_1\) estimates \(\tfrac12 M_o'-\tfrac{1}{\eta}B\).)

\STATE \textbf{4. Off-diagonal reconstruction and unpadding.}
Compute
\[
\widehat{\tfrac{1}{\eta}B} \;:=\; \frac{\widehat{D}_0-\widehat{D}_1}{2},
\qquad
\widehat{B} \;:=\; \eta\,\widehat{\tfrac{1}{\eta}B}.
\]
Define the (padded) covert operator estimate \(\widetilde{M}_c'' := \widehat{B}\).
\emph{Isometric unembedding:}
\[
\widetilde{M}_c' \;:=\; V^\dagger\,\widetilde{M}_c''\,V \;\in\; \mathcal{L}(\mathcal{H}_{M_c}).
\]
(If encryption ensured \(\mathrm{Im}(V)\subseteq\mathrm{supp}(M_o')\), then \(\Pi\) acts as the identity on \(\mathrm{Im}(V)\).)

\STATE \textbf{5. \(\mathsf{QOTP}\) decryption.}
Let
\(
U_{k'} \;=\; \bigotimes_{j=1}^{d_2} X^{\,k'_{2j-1}} Z^{\,k'_{2j}}.
\)
Output
\[
M_c \;=\; U_{k'}^\dagger\,\widetilde{M}_c'\,U_{k'}.
\]

\end{algorithmic}
\end{algorithm}

\paragraph{Lemma 1 (Dephasing via measurement).}
\label{lem:dephasing-implementation}
For any operator \(X\in\mathcal{L}(\mathcal{H}_R\otimes\mathcal{H}_M)\),
\[
\mathcal{P}^{(R)}_Z(X)\;=\;\sum_{b\in\{0,1\}} K_b\, X\, K_b^\dagger,
\qquad
K_b \;=\; \lvert b\rangle\langle b\rvert_R\otimes I_M,
\]
and equivalently
\[
\mathcal{P}^{(R)}_Z(X)\;=\;\frac{1}{2}\bigl(X+Z_R X Z_R\bigr).
\]
Thus \(\mathcal{P}^{(R)}_Z\) is realized by measuring \(R\) in the computational basis and discarding the classical outcome.

\emph{Proof.}
Let \(X\) be a \(2\times 2\) block matrix on \(R\) and let
\(
X=\begin{psmallmatrix}X_{00}&X_{01}\\[.2ex]X_{10}&X_{11}\end{psmallmatrix}
\)
with blocks in \(\mathcal{L}(\mathcal{H}_M)\).
Since \(K_b=(\lvert b\rangle\langle b\rvert_R)\otimes I_M\) are orthogonal projectors with
\(\sum_{b}K_b^\dagger K_b=I_R\otimes I_M\), we have
\begin{align*}
\sum_{b}K_b X K_b^\dagger
&=\bigl(\lvert 0\rangle\langle 0\rvert\otimes I_M\bigr) X \bigl(\lvert 0\rangle\langle 0\rvert\otimes I_M\bigr)
+\bigl(\lvert 1\rangle\langle 1\rvert\otimes I_M\bigr) X \bigl(\lvert 1\rangle\langle 1\rvert\otimes I_M\bigr)\\
&=\begin{pmatrix}X_{00}&0\\[.2ex]0&X_{11}\end{pmatrix}.
\end{align*}
Moreover \(Z_R=\lvert 0\rangle\langle 0\rvert-\lvert 1\rangle\langle 1\rvert\), hence
\(
Z_R X Z_R=\begin{psmallmatrix}X_{00}&-X_{01}\\[.2ex]-X_{10}&X_{11}\end{psmallmatrix}
\)
and therefore
\[
\frac{1}{2}\bigl(X+Z_R X Z_R\bigr)
=\frac{1}{2}\begin{pmatrix}2X_{00}&0\\[.2ex]0&2X_{11}\end{pmatrix}
=\begin{pmatrix}X_{00}&0\\[.2ex]0&X_{11}\end{pmatrix}
=\sum_{b}K_b X K_b^\dagger.
\]
Thus \(\mathcal{P}^{(R)}_Z\) coincides with the measure-and-forget channel in the \(Z\)-basis on \(R\), which is completely positive and trace-preserving. \qed

\begin{lemma}\label{lem:probe-identities-meas}
Let \(M\in\mathcal{L}(\mathcal{H}_R\otimes\mathcal{H}_M)\) have the block form
\[
M \;=\;
\begin{pmatrix}
A & B\\[.2ex]
B^\dagger & A
\end{pmatrix}_{\!R},
\qquad A,B\in\mathcal{L}(\mathcal{H}_M).
\]
Equivalently,
\(
M = A\otimes I_R + \operatorname{Re}(B)\otimes X_R - \operatorname{Im}(B)\otimes Y_R
\),
where \(\operatorname{Re}(B)=\tfrac{1}{2}(B+B^\dagger)\) and \(\operatorname{Im}(B)=\tfrac{1}{2i}(B-B^\dagger)\).

\smallskip
\textbf{X-probe.}
Apply \(H_R\) on \(R\), then measure \(R\) in the computational \(Z\)-basis.
Let \(p_b=\Pr(R=b)\), let \(\rho_{M\mid b}\) be the post-selected state on \(M\) conditioned on outcome \(b\in\{0,1\}\), and define the unnormalized block \(D_b:=p_b\,\rho_{M\mid b}\).
Then
\[
\frac{D_0-D_1}{2} \;=\; \operatorname{Re}(B).
\]

\textbf{Y-probe (phase–Hadamard).}
Let \(U_{HS}:=H_R S_R\) (Hadamard followed by phase).
Apply \(U_{HS}\) on \(R\), then measure \(R\) in the computational \(Z\)-basis.
With analogous notation \(E_b:=p_b\,\rho_{M\mid b}\) for the unnormalized post-selected blocks,
\[
\frac{E_0-E_1}{2} \;=\; \operatorname{Im}(B).
\]
(In particular, if \(B\) is Hermitian, then \(\operatorname{Im}(B)=0\) and the \(Y\)-probe is  \emph{not} necessary.)
\end{lemma}

\begin{proof}
Using
\(
M = A\otimes I_R + \operatorname{Re}(B)\otimes X_R - \operatorname{Im}(B)\otimes Y_R
\),
the \(X\)-probe conjugation yields
\[
(H_R\otimes I) M (H_R\otimes I)
= A\otimes I_R + \operatorname{Re}(B)\otimes Z_R + \operatorname{Im}(B)\otimes Y_R,
\]
since \(H_R X_R H_R=Z_R\) and \(H_R Y_R H_R=-Y_R\).
Projecting onto \(\lvert b\rangle\) on \(R\) gives
\[
D_b
= (\langle b|_R\otimes I)\,(H_R\otimes I) M (H_R\otimes I)\, (|b\rangle_R\otimes I)
= A + \operatorname{Re}(B)\,\langle b|Z_R|b\rangle + \operatorname{Im}(B)\,\langle b|Y_R|b\rangle.
\]
Since \(\langle 0|Z_R|0\rangle=+1\), \(\langle 1|Z_R|1\rangle=-1\), and \(\langle b|Y_R|b\rangle=0\), we obtain
\(D_0=A+\operatorname{Re}(B)\), \(D_1=A-\operatorname{Re}(B)\), hence \((D_0-D_1)/2=\operatorname{Re}(B)\).

For the \(Y\)-probe, with \(U_{HS}=H_R S_R\), the conjugations are
\(U_{HS} X_R U_{HS}^\dagger=-Y_R\) and \(U_{HS} Y_R U_{HS}^\dagger=-Z_R\).
Thus
\[
(U_{HS}\otimes I) M (U_{HS}^\dagger\otimes I)
= A\otimes I_R - \operatorname{Re}(B)\otimes Y_R + \operatorname{Im}(B)\otimes Z_R.
\]
Projecting onto \(\lvert b\rangle\) on \(R\) gives
\[
E_b
= A - \operatorname{Re}(B)\,\langle b|Y_R|b\rangle + \operatorname{Im}(B)\,\langle b|Z_R|b\rangle
= A + \operatorname{Im}(B)\,\langle b|Z_R|b\rangle,
\]
so \(E_0=A+\operatorname{Im}(B)\), \(E_1=A-\operatorname{Im}(B)\), and \((E_0-E_1)/2=\operatorname{Im}(B)\).
\end{proof}

\begin{remark}
In our scheme the covert off–diagonal block is
$B=\Pi\,M_c''\,\Pi$ with $M_c''=V U_{k'} M_c U_{k'}^\dagger V^\dagger$,
where $M_c$ is a density operator, $U_{k'}$ is unitary (\(\mathsf{QOTP}\)), $V$ is an
isometry, and $\Pi$ is a projector. Therefore, $B$ is Hermitian
$B^\dagger=B$, and thus $\operatorname{Im}(B)=0$. Consequently the $X$–probe
alone suffices; we estimate
\(
\frac{1}{\eta}B=\frac{D_0-D_1}{2}
\)
and set $\widehat{B}=\eta\,\frac{\widehat{D}_0-\widehat{D}_1}{2}$.
\end{remark}

\paragraph{Theorem (Exact correctness in the infinite–data limit.}
\label{thm:exact-correctness-X-only}
Assume encryption produced \(M_d^{(1)}\) in the block form of \eqref{eq:block-form-Md-physical},
\[
M_d^{(1)} \;=\;
\begin{pmatrix}
\frac{1}{2}M_o' & \frac{1}{\eta}\,B \\[0.25em]
\frac{1}{\eta}\,B^\dagger & \frac{1}{2}M_o'
\end{pmatrix}_{\!R},
\qquad
B \;=\; \Pi\,M_c''\,\Pi,
\]
with \(B\) Hermitian (as in the construction, \(B^\dagger=B\)). If Step~2 (the \(X\)–probe with explicit measurement) is performed with ideal operations and with the number of shots \(N_X\to\infty\), then the estimator
\[
\widehat{\frac{1}{\eta}B}\;=\;\frac{\widehat{D}_0-\widehat{D}_1}{2}
\qquad\text{(hence }\widehat{B}=\eta\,\widehat{\tfrac{1}{\eta}B}\text{)}
\]
converges to \(\frac{1}{\eta}B\) (hence \(\widehat{B}\to B\)).
Consequently,
\(
\widetilde{M}_c''\to \Pi M_c''\Pi
\),
\(
\widetilde{M}_c'=V^\dagger \widetilde{M}_c'' V \to M_c'
\),
and Step~5 returns \(M_c\).

\emph{Proof.}
Let \(A=\tfrac{1}{2}M_o'\). By the Pauli block decomposition,
\(
M_d^{(1)}=A\otimes I_R+\operatorname{Re}(\tfrac{1}{\eta}B)\otimes X_R-\operatorname{Im}(\tfrac{1}{\eta}B)\otimes Y_R
\).
Applying \(H_R\) on \(R\) maps \(X_R\mapsto Z_R\) and \(Y_R\mapsto -Y_R\).
Measuring \(R\) in the computational basis and discarding the classical outcome (i.e.\ the dephasing channel \(\mathcal{P}^{(R)}_Z\)) removes the \(Y_R\)–term and retains the \(Z_R\)–term. Denote by \(p_b\) the outcome probability, by \(\rho_{M\mid b}\) the conditional post–measurement state of \(M\), and set \(D_b=p_b\rho_{M\mid b}\) (unnormalized blocks). A direct calculation then gives the probe identity
\[
\frac{D_0-D_1}{2}\;=\;\operatorname{Re}\!\Big(\frac{1}{\eta}B\Big).
\]
In our construction \(B\) is Hermitian, hence \(\operatorname{Im}(B)=0\) and \(\operatorname{Re}(\tfrac{1}{\eta}B)=\tfrac{1}{\eta}B\). Therefore the target identity simplifies to
\[
\frac{D_0-D_1}{2}\;=\;\frac{1}{\eta}B.
\]
With ideal operations and \(N_X\to\infty\), the empirical frequencies \(\widehat{p}^{(X)}_b\) converge to \(p_b\), and the linear–inversion tomographic estimators \(\widehat{\rho}^{(X)}_{M\mid b}\) converge to \(\rho_{M\mid b}\) (entrywise, hence in any Schatten norm) by the law of large numbers. Thus \(\widehat{D}_b=\widehat{p}^{(X)}_b\widehat{\rho}^{(X)}_{M\mid b}\to D_b\), and consequently
\(
\widehat{\tfrac{1}{\eta}B}=\frac{\widehat{D}_0-\widehat{D}_1}{2}\to \frac{1}{\eta}B
\),
whence \(\widehat{B}\to B\).
Since unembedding \(X\mapsto V^\dagger X V\) and \(\mathsf{QOTP}\) decryption are completely positive trace–preserving (hence continuous) maps, we obtain
\(
\widetilde{M}_c''\to \Pi M_c''\Pi
\),
\(
\widetilde{M}_c'\to M_c'
\),
and finally recovery of \(M_c\).
\(\square\)

\smallskip
\paragraph{Finite–data accuracy and sample complexity (Pauli tomography; $X$–probe only).}
Let $d=2^{d_1}$ and fix $\varepsilon\in(0,1)$, $\delta\in(0,1)$. In the $\mathsf{DCM}$ procedure, the covert off–diagonal block is $B=\Pi M_c''\Pi$, which is Hermitian by construction; hence $\operatorname{Im}(B)=0$ and only the $X$–probe is required.

\begin{theorem}[Finite–data accuracy and sample complexity based on Pauli tomography]\label{thm:finitedatapauli-Xonly}
Let $d=2^{d_1}$ and let $B\in\mathcal{L}(\mathbb{C}^d)$ be the (Hermitian) covert off–diagonal operator. Fix $\varepsilon,\delta\in(0,1)$. Run the $\mathsf{DCM}$ \emph{$X$–probe only} as in Algorithm~\ref{alg:covert-decryption-rigorous-physical}: apply $U_{\sigma_l}^\dagger$, then $H_R$, measure $R$ in the computational basis, and perform Pauli tomography on $M$ with Pauli strings sampled uniformly across the $d+1$ commuting Pauli frames (Fact~\textup{\ref{fact:frames}}). Let $N_X$ be the total number of shots, allocated uniformly over the frames.

If
\[
N_X \ \ge\ \frac{(d+1)\,d}{2\,\varepsilon^{2}}\;\log\!\frac{4 d^{2}}{\delta},
\]
then, with probability at least $1-\delta$,
\[
\Big\|\widehat{\frac{1}{\eta}B}-\frac{1}{\eta}B\Big\|_{2}\ \le\ \varepsilon
\qquad\Longrightarrow\qquad
\|\widehat{B}-B\|_{1}\ \le\ \eta\,\sqrt{d}\,\varepsilon.
\]
Equivalently, there exists a universal constant $c>0$ such that it suffices to have
\[
N_X \ \ge\ \frac{c\,d^{2}}{\varepsilon^{2}}\;\log\!\frac{2 d^{2}}{\delta}.
\]
\end{theorem}

\paragraph{Proof.}
We consider the $X$–probe target as $T:=\tfrac{1}{\eta}B$. Using the Pauli orthonormal basis on $\mathcal{L}(\mathbb{C}^d)$ (Fact~\ref{fact:PauliONB}), we expand
\[
T \ =\ \sum_{P\in\mathcal{P}_{d_1}} \alpha_{P}\,E_{P},\qquad
\alpha_{P}\ =\ \langle E_{P},T\rangle_{F}\ =\ \frac{1}{\sqrt{d}}\operatorname{Tr}(P\,T),
\]
and define the linear-inversion estimator
\(
\widehat{T}=\sum_{P}\widehat{\alpha}_{P}E_{P}
\),
where, for each Pauli string $P$ in a given commuting frame, $\widehat{\alpha}_{P}$ is obtained from the difference of empirical means on the two post–measurement branches $b\in\{0,1\}$ after the $H_R$ gate
\[
\widehat{\alpha}_{P}\ :=\ \frac{1}{\sqrt{d}}\cdot\frac{\overline{Z}^{(0)}_{P}-\overline{Z}^{(1)}_{P}}{2},
\]
with $\overline{Z}^{(b)}_{P}$ the empirical mean of $\pm1$ outcomes of $P$ on the unnormalized reduced block $D_b$ (see Lemma~\ref{lem:probe-identities-meas}). Each $\overline{Z}^{(b)}_{P}$ averages $n_s:=N_X/(d+1)$ samples (uniform allocation over the $d+1$ frames).

By Hoeffding’s inequality for $\pm1$–valued variables and a union bound over the two branches,
\[
\Pr\!\Big(\Big|\tfrac{1}{2}\big(\overline{Z}^{(0)}_{P}-\overline{Z}^{(1)}_{P}\big)-\tfrac{1}{2}\big(\mu^{(0)}_{P}-\mu^{(1)}_{P}\big)\Big|\ \ge\ t\Big)\ \le\ 4\,e^{-2 n_s t^{2}},
\]
where $\mu^{(b)}_{P}:=\mathbb{E}[\overline{Z}^{(b)}_{P}]$. Hence, for every $P$,
\begin{equation}\label{eq:per-coeff-tail-Xonly}
\Pr\!\Big(\big|\widehat{\alpha}_{P}-\alpha_{P}\big|\ \ge\ \frac{t}{\sqrt{d}}\Big)\ \le\ 4\,e^{-2 n_s t^{2}}.
\end{equation}
There are at most $m=d^{2}$ nonidentity Paulis. Applying a union bound over $P\neq I$ in \eqref{eq:per-coeff-tail-Xonly} with
\(
t=\sqrt{\frac{1}{2n_s}\log\frac{4 d^{2}}{\delta}}
\),
we obtain, with probability at least $1-\delta$,
\[
\max_{P\neq I}\big|\widehat{\alpha}_{P}-\alpha_{P}\big|\ \le\ \frac{1}{\sqrt{d}}\sqrt{\frac{1}{2n_s}\log\frac{4 d^{2}}{\delta}}.
\]
By orthonormality of the Pauli basis (Fact~\ref{fact:PauliONB}),
\[
\|\widehat{T}-T\|_{2}^{2}
\;=\;\sum_{P\neq I}\big|\widehat{\alpha}_{P}-\alpha_{P}\big|^{2}
\ \le\ m\cdot\frac{1}{d}\cdot\frac{1}{2n_s}\log\frac{4 d^{2}}{\delta}
\ \le\ \frac{d}{2n_s}\log\frac{4 d^{2}}{\delta}.
\]
Since $n_s=N_X/(d+1)$, the stated condition
\(
N_X \ge \frac{(d+1)\,d}{2\,\varepsilon^{2}}\log\frac{4 d^{2}}{\delta}
\)
implies $\|\widehat{T}-T\|_{2}\le\varepsilon$ with probability at least $1-\delta$. Finally, by the Schatten–norm inequality $\|X\|_{1}\le\sqrt{d}\,\|X\|_{2}$ and $\widehat{B}-B=\eta(\widehat{T}-T)$, we get,
\[
\|\widehat{B}-B\|_{1}\ \le\ \eta\,\sqrt{d}\,\|\widehat{T}-T\|_{2}\ \le\ \eta\,\sqrt{d}\,\varepsilon.
\]
This proves the claim. \qed

\begin{remark}
For low-rank/coherent structure, classical shadows or compressed-sensing tomography can reduce the scaling; these are implementable by random local Cliffords followed by \(Z\) measurements on \(M\).
\end{remark}

\medskip

Next, we describe an algorithm to extract the original ciphertext from the anamorphic ciphertext, which we will use in the definition of anamorphic secret-sharing.

\begin{algorithm}[H]
\caption{Extraction of Original Ciphertext from Anamorphic Ciphertext(\(\mathsf{EOC}\))}
\label{alg:original-ciphertext-extraction}
\begin{algorithmic}[1]
\STATE \textbf{Input:} Anamorphic state \(M_f^{(1)} \in \mathcal{D}((\mathbb{C}^2)^{\otimes(d_1+1)})\), permutation unitary \(U_{\sigma_l}\), \(\mathsf{QOTP}\) key \(k \in \{0, 1\}^{2d_1}\), dimension \(d_1\).
\STATE \textbf{Output:} Original ciphertext \(M_f^{(0)}\).
\STATE \textbf{Steps:}
\STATE \textbf{1. Apply the inverse permutation:}
\STATE \hspace{1cm} Compute the intermediate state \(M_d^{(1)}\) by applying the inverse of \(U_{\sigma_l}\):
\[
M_d^{(1)} = U_{\sigma_l}^\dagger M_f^{(1)} U_{\sigma_l}.
\]
\begin{comment}
\STATE \textbf{2. Extract the first block and fourth block:}
\STATE \hspace{1cm} Define the projectors:
\[
\Pi_0 = \ket{0}\bra{0} \otimes I_{2^{d_1}}, \qquad \Pi_1 = \ket{1}\bra{1} \otimes I_{2^{d_1}},
\]
where \(I_{2^{d_1}}\) is the identity operator on \(\mathcal{H}_{M} = (\mathbb{C}^2)^{\otimes d_1}\).
\end{comment}
\smallskip
\STATE \textbf{Quantum Dephasing Channel on Control Qubit} 
\STATE Implement the dephasing channel $\Gamma_0: \mathcal{L}(\mathcal{H}_R) \longrightarrow \mathcal{L}(\mathcal{H}_R)$:
\[
\Gamma_0(\rho_R) = \frac{1}{2}\rho_R + \frac{1}{2} Z_R \rho_R Z_R
\]
and set
\begin{equation}\label{eq:after-pinch-eoc}
M_a^{(0)}\ :=\ (\Gamma_0\otimes I_{2^{d_{1}}})\big(M_d^{(1)}\big).
\end{equation}

\STATE \textbf{4. Apply the permutation to recover original ciphertext:}
\STATE \hspace{1cm} Use the permutation unitary \(U_{\sigma_l}\) to recover \(M_f^{(0)}\):
\[
M_f^{(0)} = U_{\sigma_l} M_{a}^{(0)} U_{\sigma_l}^\dagger.
\]
\STATE \textbf{5. Return the original ciphertext \(M_f^{(0)}\):}
\STATE \hspace{1cm} Output \(M_f^{(0)}\), completing the extraction.
\end{algorithmic}
\end{algorithm}

\medskip

\begin{remark}\label{re:dom}
The algorithm \(\mathsf{DOM}\) can be applied to the original ciphertext \(M_{f}^{(0)}\) too. Exactly in a similar way, we can retrieve the original message \(M_{o}\) from \(M_{f}^{(0)}\).
\end{remark}

\medskip

Now we discuss the following theorems and corollaries to prove the Theorem \ref{thm:main} and the corollary \ref{cor:main}.
\medskip

\begin{definition}(Rayleigh Quotient [Section 4.2, Page 176, \cite{horn2012matrix}])
Let $A$ be a Hermitian operator on an $n$-dimensional complex Hilbert space $\mathcal{H}$. The Rayleigh quotient $R(A;x)$ associated with $x\neq 0$ is defined as
\begin{equation}
  R(A;x) 
  \;=\;
  \frac{x^{*} A x}{x^{*} x}.
\end{equation}
\end{definition}

\medskip

\noindent We now recall the fundamental \emph{variational characterization} of eigenvalues of a Hermitian operator. This is sometimes called the \emph{Rayleigh--Ritz theorem} (in finite dimensions).

\begin{theorem}(Variational Characterization of the Extreme Eigenvalues \cite{horn2012matrix, bhatia2013matrix, kato2013perturbation})
\label{thm:Rayleigh}
Let $A$ be a Hermitian $n\times n$ matrix (or Hermitian operator on an $n$-dimensional space). Denote its eigenvalues by
\begin{equation}
  \lambda_1(A) \;\le\; \lambda_2(A) \;\le\; \cdots \;\le\; \lambda_n(A)
\end{equation}
ordered in a nondecreasing sequence. Then
\begin{equation}
  \min_{\ket{v}\neq 0}\,R(A;\ket{v})
  \;=\;
  \lambda_1(A),
  \quad
  \max_{\ket{v}\neq 0}\,R(A;\ket{v})
  \;=\;
  \lambda_n(A).
\end{equation}
In particular,
\begin{equation}
  \lambda_{\min}(A)
  \;=\;
  \min_{\|\ket{v}\|_{2}=1}\,\bra{v}A\ket{v},
  \quad
  \lambda_{\max}(A)
  \;=\;
  \max_{\|\ket{v}\|_{2}=1}\,\bra{v}A\ket{v}.
\end{equation}
\end{theorem}

\medskip

The following Corollary \ref{cor:minmax} can easily be derived from the Theorem \ref{thm:Rayleigh}, which we have used in our proof of the Theorem \ref{thm:main}. The Lemma \ref{lemma:semdef} is also easy to prove using basic linear algebra, and we have used it to prove Corollary \ref{cor:semdef}. We have included them for completeness.

\begin{corollary}\label{cor:minmax}
Let $X, Y \succeq 0$ be positive semi-definite matrices on the same finite-dimensional Hilbert space $\mathcal{H}$. If
\begin{equation}
  \lambda_{\max}(Y)
  \;\;\le\;\;
  \lambda_{\min}(X),
\end{equation}
then
\begin{equation}
  (X - Y) \;\succeq\;0.
\end{equation}
\end{corollary}

\begin{proof}
Since $X,Y \succeq 0$, all their eigenvalues are nonnegative. To show that $(X - Y)$ is positive semi-definite, we show that for all $\ket{v}\in\mathcal{H}$,
\begin{equation}
  \bra{v}(X-Y)\ket{v} 
  \;\ge\; 0.
\end{equation}
As $\lambda_{\min}(X)$ is the smallest eigenvalue of $X$, and by the variational characterization of eigenvalues (or Rayleigh quotients), we have
\begin{equation}
  \frac{\bra{v}X\ket{v}}{\|\ket{v}\|_{2}^2}
  \;\ge\;
  \lambda_{\min}(X).
\end{equation}
Since $X \succeq 0$, for any $\ket{v}\neq 0$,
\begin{equation}
  \bra{v}X\ket{v}
  \;\ge\;
  \lambda_{\min}(X)\,\|\,\ket{v}\,\|_{2}^2.
\end{equation}
Since $\lambda_{\max}(Y)$ is the largest eigenvalue of $Y$, the Rayleigh quotient satisfies
\begin{equation}
  \frac{\bra{v}Y\ket{v}}{\|\ket{v}\|_{2}^2}
  \;\le\;
  \lambda_{\max}(Y).
\end{equation}
Similarly, for $Y \succeq 0$, we have
\begin{equation}
  \bra{v}Y\ket{v}
  \;\le\;
  \lambda_{\max}(Y)\,\|\ket{v}\|_{2}^2.
\end{equation}
Combining both the equations together, we get that, for any nonzero $\ket{v}\in\mathcal{H}$,
\begin{equation}
  \bra{v}X\ket{v}
  \;-\;
  \bra{v}Y\ket{v}
  \;\ge\;
  \bigl(\lambda_{\min}(X)\,\|\ket{v}\|_{2}^2\bigr)
  \;-\;
  \bigl(\lambda_{\max}(Y)\,\|\ket{v}\|_{2}^2\bigr)
  \;=\;
  \Bigl(\lambda_{\min}(X) - \lambda_{\max}(Y)\Bigr)\,\|\ket{v}\|_{2}^2.
\end{equation}
By our assumption, $(\lambda_{\min}(X) - \lambda_{\max}(Y))\geq 0$.

\noindent Therefore,
\begin{equation}
 \forall \ket{v} \in \mathcal{H}, \quad \bra{v}(X-Y)\ket{v} 
  \;\ge\;
  \Bigl(\lambda_{\min}(X) - \lambda_{\max}(Y)\Bigr)\,\|\ket{v}\|_{2}^2
  \;\ge\; 
  0.
\end{equation}

\medskip

\noindent Hence, $(X - Y) \succeq 0$, that is, \ $(X-Y)$ is a positive semi-definite matrix.
\end{proof}

\medskip

\begin{lemma}(\cite{horn2012matrix})\label{lemma:maxeig}
Let $X \succeq 0$ be a Hermitian and positive semi-definite matrix. Then
\begin{equation}
  \lambda_{\max}(X)
  \;=\;
  \| X \|,
\end{equation}
where $\lambda_{\max}(M)$ is the largest eigenvalue of $M$, and $\|M\|$ is the spectral norm or operator norm of \(X\).
\end{lemma}

\medskip

In general, if \(A\) and \(B\) are two positive semi-definite matrices, then \(AB\) may not be a positive semi-definite matrix unless commutativity holds, that is \(AB=BA\). But in our case we next prove that if \(M_{o}'\) is strictly positive definite and \(M_{c}''\) is positive semi-definite, then the matrix \(M_c''\,(M_o')^{-1}\,M_c''\) is also positive semi-definite. Here we note that it is \textit{not} necessary that the matrix \(M_{c}''\) should be strictly positive definite.

\begin{lemma}\label{lemma:semdef}
   Let \(Y \in M_{n}(\mathbb{C})\) be a Hermitian positive definite matrix, and let \(X \in M_{n}(\mathbb{C})\) be a Hermitian positive semi-definite matrix. Then
  \[
    XY^{-1}X 
    \quad \text{is a Hermitian positive semi-definite matrix.}
  \]
  Moreover, \(XY^{-1}X\) is strictly positive definite if and only if \(X\) is invertible.
\end{lemma}

\begin{proof}
For vectors \(v, w \in \mathbb{C}^n\), we define,
\(
\langle v, w \rangle := v^\dagger w.
\)
Then \(\langle \cdot, \cdot \rangle\) defines an inner product on \(\mathbb{C}^n\). Now as \(X\) is positive semi-definite, for all non-zero vectors \(v \in \mathbb{C}^n\),
\(
  \langle v, Xv \rangle \geq 0.
\)
Since \(Y\) is positive definite, its inverse \(Y^{-1}\) exists and is also positive definite. The matrix \(XY^{-1}X\) is Hermitian since both \(X\) and \(Y\) are Hermitian. 

\medskip
\noindent Consider an arbitrary vector \(v \in \mathbb{C}^n\) and also consider the
inner product \(\langle v, XY^{-1}Xv \rangle.\) 

\noindent Define \(w := Xv\). Then \(\langle v, XY^{-1}Xv \rangle=\langle w,Y^{-1}w\rangle\). Since \(Y\), is positive definite, for all \(w \in \mathbb{C}^{n}\), \(\langle w, Y^{-1}w \rangle \geq 0.\) Therefore, \(XY^{-1}X\) is also a positive semi- definite matrix. 
\medskip

\noindent Note that \(XY^{-1}X\) fails to be strictly positive definite precisely if there exists a nonzero vector \(v\) such that
  \begin{equation}
    \langle v, \,XY^{-1}Xv \rangle \;=\; 0.
  \end{equation}
  However, we have
  \begin{equation}
    \langle v, \,XY^{-1}Xv \rangle 
    \;=\;
    \langle Xv, \,Y^{-1}Xv \rangle.
  \end{equation}
  Since \(Y^{-1}\) is strictly positive definite, \(\langle w, \,Y^{-1}w\rangle = 0\) if and only if \(w=0\). Hence,
  \begin{equation}
    \langle Xv, \,Y^{-1}Xv \rangle = 0
    \quad \text{iff} \quad
    Xv = 0.
  \end{equation}
  Thus, if \(X\) has a nontrivial kernel, there is a nonzero \(v\) with \(Xv=0\), leading to \(\langle v, XY^{-1}Xv\rangle = 0\). This shows that if \(X\) is not invertible, then \(XY^{-1}X\) is not strictly positive definite. 

  \medskip

\noindent Conversely, if \(X\) is invertible, then \(Xv=0\) implies \(v=0\). Hence the only way \(\langle v, XY^{-1}Xv\rangle=0\) can hold is if \(v=0\).

  \medskip

\noindent Therefore, \(XY^{-1}X\) is always Hermitian and positive semi-definite. It is strictly positive definite if and only if \(X\) is invertible.
\end{proof}

\medskip

\begin{corollary}\label{cor:semdef}
    The matrix \(M_{c}''(M_{o}')^{-1}M_{c}''\) is a Hermitian and a positive semi-definite matrix.
\end{corollary}

\begin{proof}
    By applying the Lemma \ref{lemma:semdef}. with \(X=M_{c}''\) and \(Y=M_{o}'\), the result follows.
\end{proof}

\medskip

\noindent We now state one of the main theorems of our paper.

\begin{theorem}\label{thm:main}
\noindent Given a security parameter \(\mathsf{negl}(\lambda) >0\), with $\eta \in \mathbb{Z}^+$ such that $\dfrac{1}{\eta} < \mathsf{negl}(\lambda)$ both the original and the anamorphic quantum states \(M_{f}^{(0)}\) and \(M_{f}^{(1)}\) are quantum density matrices, if 
\begin{equation}
  \frac1{\eta^2}\;
  \|M_c''\,(M_o')^{-1}\,M_c''\|
  \;\;\le\;\;
  \frac14\,\lambda_{\min}(M_o'),
\end{equation} 
where \(\lambda_{\min}(M_o')\) is the minimum eigenvalue of \(M_o'\) on its support, \(\|\cdot\|\) is the operator norm and \(M_{o}'\) is a strictly positive definite matrix.
\end{theorem}

\begin{proof}
We consider both the block matrices
\begin{equation}
M_{a}^{(0)} 
\;=\; 
 \begin{pmatrix} \frac{1}{2}M_o' & \mathbf{0}_{2^{d_{1}} \times 2^{d_{1}}} \\ \mathbf{0}_{2^{d_{1}} \times 2^{d_{1}}} & \frac{1}{2}M_o' \end{pmatrix} \; 
\text{and}\;
M_{a}^{(1)} 
\;=\; 
\begin{pmatrix} 
\tfrac{1}{2}\,M_o' & \tfrac{1}{\eta}\,M_c'' \\[6pt] 
\tfrac{1}{\eta}\,\bigl(M_c''\bigr)^{\dagger} & \tfrac{1}{2}\,M_o' 
\end{pmatrix},
\end{equation}
where \(M_o'\in \mathcal{D}(\mathcal{H}_M)\) is the \(\mathsf{QOTP}\)‐encrypted version of a density matrix \(M_o\).  Since \(\mathsf{QOTP}\) preserves positivity and trace, 
\(\,M_o' \succ 0\) and \(\operatorname{Tr}(M_o')=1\). Therefore, the matrix \(M_{a}^{(0)}\) is a density matrix. 

The encrypted covert matrix \(M_c'' \in \mathcal{D}(\mathcal{H}_{M})\) is another operator obtained by encrypting the covert message \(M_{c}\) using \(\mathsf{QOTP}\) and then padding the encrypted covert message \(M_{c}'\). Since the matrix \(M_{c}\) is also a density matrix, and since \(\mathsf{QOTP}\) preserves positivity and trace, 
\(\,M_{c}' \succeq 0\) and \(\operatorname{Tr}(M_{c}')=1\). As \(
  \mathcal{E}^{\mathrm{pad}}(M_{c}')
  \;=\;
  V \,M_{c}' \, V^\dagger=M_{c}'',
\) and \(\mathcal{E}^{\mathrm{pad}}\) is an isometry and completely positive by construction of \(V\), \(M_{c}'' \succeq 0\) and \(\Tr(M_{c}'')=1\). As \(M_{c}\) is Hermitian, \(M_{c}'\) is also Hermitian, and consequently \(M_{c}''\) is also Hermitian. The parameter \(\eta>0\) is used to scale the off‐diagonal blocks. It is clear that \(\operatorname{Tr}(M_{a})=1\). Therefore, we only analyze whether \(M_{a}\) is a positive semi-definite \((M_{a} \succeq 0)\) matrix. Writing
\begin{equation}
  M_{a}
  \;=\;
  \begin{pmatrix}
    A & B \\[3pt]
    B^\dagger & A
  \end{pmatrix},
  \quad\text{where}\quad
  A \;=\;\tfrac12\,M_o',\quad
  B \;=\;\tfrac1\eta\,M_c'',
\end{equation}
we note that
\begin{equation}
  A = \tfrac12\,M_o' \;\succ\; 0
  \quad\text{iff}\quad
  M_o' \;\succ\; 0,
\end{equation}
which is true by hypothesis.  The potential problem for positivity arises from the off‐diagonal blocks \(B\) and \(B^\dagger\).  

\noindent By the Schur complement condition \ref{thm:schur} for the positivity of a \(2\times 2\) block matrix, is that
\begin{equation}
  A \;\succ\;0,
  \quad
  \text{and}\quad
  A - B\,A^{+}\,B^\dagger \;\succeq\;0
  \quad,
\end{equation}
where \(A^{+}\) denotes the generalized Moore Penrose inverse on the \(\textit{supp}(A)\). By the Theorem \ref{thm:MPisInverse}. We have \((M_{o}')^{+}=(M_{o}')^{-1}\). Substituting 
\(\,A=\tfrac12\,M_o'\) and \(\,B=\tfrac1\eta\,M_c''\), we obtain
\begin{equation}
  \tfrac12\,M_o'
  \;-\;
  \Bigl(\tfrac1\eta\,M_c''\Bigr)\,
  \Bigl(\tfrac12\,M_o'\Bigr)^{-1}
  \Bigl(\tfrac1\eta\,M_c''\Bigr)^\dagger
  \;\;\succeq\;0
  \quad\text{iff}\quad
  \tfrac12\,M_o'
  \;-\;
  \tfrac1{\eta^2}\,
  (M_c'')\,
  \Bigl(\tfrac12\,M_o'\Bigr)^{-1}
  (M_c'')^\dagger
  \;\;\succeq\;0.
\end{equation}

\noindent We prove that a sufficient condition is to require
\begin{equation}
  \frac1{\eta^2}\;
  \|(M_c'')^\dagger\,(M_o')^{-1}\,M_c''\|
  \;\;\le\;\;
  \frac14\,\lambda_{\min}(M_o'),
\end{equation}
where \(\lambda_{\min}(M_o')\) is the minimum eigenvalue of \(M_o'\) on its support and \(\|\cdot\|\) is the operator norm. Since $A = \tfrac12\,M_o'$, and $M_o'\succ 0$ is invertible on its support, we can write
\begin{equation}
  A^{-1}
  \;=\;
  \bigl(\tfrac12\,M_o'\bigr)^{-1}
  =
  2\,\bigl(M_o'\bigr)^{-1}
  \quad\text{on }
  \textit{supp}(M_o').
\end{equation}
Let $B = \tfrac1\eta\,M_c''$. Then $B^\dagger = \tfrac1\eta\,(M_c'')^\dagger = \tfrac1\eta\,M_c''$ (since $M_c''\succeq0$, and \(M_{c}''\) is also Hermitian).

\noindent Hence,
\begin{equation}
  B^\dagger\,A^{-1}\,B
  =
  \Bigl(\tfrac1\eta\,M_c''\Bigr)\,
  \Bigl(2\,(M_o')^{-1}\Bigr)\,
  \Bigl(\tfrac1\eta\,M_c''\Bigr)
  =
  \tfrac{2}{\eta^2}\,
  M_c''\,(M_o')^{-1}\,M_c''.
\end{equation}
Thus,
\begin{equation}
  A - B^\dagger A^{-1}B
  =
  \tfrac12\,M_o'
  \;-\;
  \tfrac{2}{\eta^2}\,
  M_c''\,(M_o')^{-1}\,M_c''.
\end{equation}

\noindent To ensure
\begin{equation}
  \tfrac12\,M_o'
  \;-\;
  \tfrac{2}{\eta^2}\,
  M_c''\,(M_o')^{-1}\,M_c''
  \;\;\succeq\;0,
\end{equation}
we show that a sufficient condition, is that the largest eigenvalue of $\tfrac{2}{\eta^2}\,M_c''\,(M_o')^{-1}\,M_c''$ \emph{does not exceed} the smallest eigenvalue of $\tfrac12\,M_o'$. 

\medskip

\noindent Let 
\begin{equation}
  X 
  \;=\;
  \tfrac12\,M_o',
  \quad
  Y 
  \;=\;
  \tfrac{2}{\eta^2}\,
  M_c''\,(M_o')^{-1}\,M_c''.
\end{equation}
We want $(X - Y) \succeq 0$. We know that if $X,Y$ are positive semi-definite, then 
\begin{equation}
  \lambda_{\max}(Y) \;\le\; \lambda_{\min}(X) 
 \quad \text{implies} \quad (X - Y) \;\succeq 0,
\end{equation}
where $\lambda_{\max}$ and $\lambda_{\min}$ denote the maximum and minimum eigenvalues on the relevant support.

\medskip

\noindent Now,
\begin{equation}
  \lambda_{\min}(X)
  = 
  \lambda_{\min}\Bigl(\tfrac12\,M_o'\Bigr)
  =
  \tfrac12\,\lambda_{\min}(M_o'),
\end{equation}
since scaling an operator by $\tfrac12$ scales all eigenvalues by $\tfrac12$,

\medskip
\noindent and
\begin{equation}
  \lambda_{\max}(Y)
  \;=\;
  \lambda_{\max}\Bigl(\tfrac{2}{\eta^2}\,M_c''\,(M_o')^{-1}\,M_c''\Bigr)
  =
  \tfrac{2}{\eta^2}\,
  \lambda_{\max}\Bigl(M_c''\,(M_o')^{-1}\,M_c''\Bigr).
\end{equation}
Because $M_c'', (M_o')^{-1}$, are positive semi-definite and positive definite, respectively, by the Corollary \ref{cor:semdef}, $M_c''\,(M_o')^{-1}\,M_c''$ is positive semi-definite, and by Lemma \ref{lemma:maxeig}, we get $\lambda_{\max}\bigl(M_c''\,(M_o')^{-1}\,M_c''\bigr) = \|(M_c'')^\dagger\,(M_o')^{-1}\,M_c''\|$, i.e.\ the spectral norm of that product.

\medskip

\noindent Therefore, we need
\begin{equation}
  \tfrac{2}{\eta^2}\,
  \lambda_{\max}\bigl(M_c''\,(M_o')^{-1}\,M_c''\bigr)
  \;\;\le\;\;
  \tfrac12\,\lambda_{\min}(M_o'),
\end{equation}
which is equivalent to
\begin{equation}
  \frac1{\eta^2}\,
  \lambda_{\max}\bigl(M_c''\,(M_o')^{-1}\,M_c''\bigr)
  \;\;\le\;\;
  \frac14\,\lambda_{\min}(M_o').
\end{equation} Therefore,
\begin{equation}\label{eq:psmcond}
  \frac1{\eta^2}\;
  \bigl\|\,
    (M_c'')^\dagger\,(M_o')^{-1}\,M_c''
  \bigr\|
  \;\;\le\;\;
  \frac14\,\lambda_{\min}(M_o').
\end{equation} 
This condition forces \(M_{a} \succeq 0\).  Since \(M_o'\) is a strictly positive-definite density operator, \(\lambda_{\min}(M_o')> 0\). Since \((M''_{c})^{\dagger}=M''_{c}\), by making \(\eta\) sufficiently large, one can always satisfy \(\,\tfrac1{\eta^2}\;\|M_c''\,(M_o')^{-1}\,M_c''\|\;\;\le\;\;\tfrac14\,\lambda_{\min}(M_o')\).
Since the permutation matrices are unitary matrices, they preserve the positive semi-definiteness and unit trace. Hence, both the matrices \(M_{f}^{(0)}\) and \(M_{f}^{(1)}\) are positive semi-definite with unit trace. Therefore, both the original and anamorphic matrices are quantum density matrices.
\end{proof}

\medskip

\begin{corollary}\label{cor:main}(A weaker sufficient condition)
\noindent Given a security parameter \(\mathsf{negl}(\lambda) >0\), with $\eta \in \mathbb{Z}^+$ such that $\dfrac{1}{\eta} < \mathsf{negl}(\lambda)$ both the original and the anamorphic quantum states \(M_{f}^{(0)}\) and \(M_{f}^{(1)}\) are quantum density matrices, if  
\begin{equation}
  \frac{2\lambda_{\max}(M_c'')}{\lambda_{\min}(M_o')} \leq\eta,
\end{equation}
where \(\lambda_{\min}(M_o')\) and \(\lambda_{\max}(M_c'')\) are the minimum and maximum eigenvalues of \(M_o'\) and \(M_{c}''\) on their respective supports.
\end{corollary}

\begin{proof}
If 
\begin{equation}
  \frac{\lambda_{\max}(M_c'')}{\eta}
  \;\le\;
  \frac12\,\lambda_{\min}(M_o'),
\end{equation}

then, 

\begin{align}
\frac{\lambda_{\max}(M_c'')}{\eta}
\;\le\;
\frac12\,\lambda_{\min}(M_o') \quad \text{implies} \quad \frac1{\eta^2}\;
  \frac{\lambda_{\max}(M_c'')^2}{\lambda_{\min}(M_o') }
  \;\;\le\;\;
  \frac14\,\lambda_{\min}(M_o').
\end{align}

\noindent We know that for any two matrices \(A\) and \(B\) of compatible dimensions, \(\|AB\|\;\le\;\|A\|\;\|B\|\).
Therefore, we get, 
\begin{equation}
  \|(M_c'')^\dagger\,(M_o')^{-1}\,M_c''\|
  \;\;\le\;\;
  \|(M_c'')^\dagger\|\;\|(M_o')^{-1}\|\;\|M_c''\|.
\end{equation}

\noindent Since \(M_c''\) is positive semi-definite and Hermitian, \((M_c'')^\dagger = M_c''\) and hence
\(
  \|(M_c'')^\dagger\|
  =
  \|M_c''\|.
\)

\medskip
\noindent Therefore,
\begin{equation}
  \|(M_c'')^\dagger\,(M_o')^{-1}\,M_c''\|
  \;\;\le\;\;
  \|M_c''\|^2\;\|(M_o')^{-1}\|. \label{eq:72}
\end{equation}
Since \(M_c''\) is positive semi-definite, all its eigenvalues are non-negative, and the spectral norm equals the maximum eigenvalue \ref{lemma:maxeig}, 
\begin{equation}\label{eq:72a}
\|M_c''\| = \lambda_{\max}(M_c'').
\end{equation} Similarly, since \(M_o'\) is strictly positive-definite, \(\lambda_{\min}(M_o')>0\) and hence,
\begin{equation}
  \|(M_o')^{-1}\|
  = 
  \lambda_{\max}\!\bigl((M_o')^{-1}\bigr)
  = 
  \frac{1}{\lambda_{\min}(M_o')}. \label{eq:73}
\end{equation}

\noindent Combining the equations \ref{eq:72},\ref{eq:72a},\ref{eq:73}, we get,
\begin{equation}
  \|(M_c'')^\dagger\,(M_o')^{-1}\,M_c''\|
  \;\;\le\;\;
  \frac{\lambda_{\max}(M_c'')^2}{\lambda_{\min}(M_o')}.
\end{equation}

\noindent Hence,
\begin{align}
\frac1{\eta^2}\; \|(M_c'')^\dagger\,(M_o')^{-1}\,M_c''\| \leq \frac1{\eta^2} \frac{\lambda_{\max}(M_c'')^2}{\lambda_{\min}(M_o')} \leq \frac14\,\lambda_{\min}(M_o'),
\end{align}
which satisfies the sufficient condition we proved in Theorem \ref{thm:main}.
\end{proof}

\bigskip

\noindent \(\bullet\) \textbf{Adversary's Observations in the Two Games}

\medskip

i) \textbf{The real game  \( \text{RealG}_{\mathcal{Q}}(\lambda, \mathcal{D}): \)} 
In the real game, only the original message $M_o$ is encrypted and sent. We have the encrypted message \(M_o' = U_k M_o U_k^\dagger\) by the \(\mathsf{QOTP}\) and the key \(k\). Then we have constructed the following message \(M_a^{(0)} = \begin{pmatrix} \frac{1}{2}M_o' & \mathbf{0}_{2^{d_{1}} \times 2^{d_{1}}} \\ \mathbf{0}_{2^{d_{1}} \times 2^{d_{1}}} & \frac{1}{2}M_o' \end{pmatrix}.\)
The final state we have is \(M_f^{(0)} = U_{\sigma_l} M_a^{(0)} U_{\sigma_l}^\dagger \in \mathcal{D}(\mathcal{H}_{C}).\)

\medskip

ii) \textbf{The anamorphic game \( \text{AnamorphicG}_{\mathcal{Q}_a}(\lambda, \mathcal{D}) \)}
In the anamorphic game, both $M_o$ and $M_c$ are encrypted and sent. The encrypted message is \(M_o' = U_k M_o U_k^\dagger, \quad M_c' = U_{k'} M_c U_{k'}^\dagger.\) Then we have combined both the original message and the hidden message as \(M_a^{(1)}=M_{a}.\) Finally, we have constructed the state  \(M_f^{(1)} = U_{\sigma_l} M_a^{(1)} U_{\sigma_l}^\dagger \in \mathcal{D}(\mathcal{H}_{C}).\)

\medskip

In both games, the adversary receives the state $M_f^{(b)}$ for $b \in \{0,1\}$, without knowledge of  $l$, $k$, $k'$, \(d_{1}\), \(d_{2}\), \(\eta\). The expectation \(\mathbb{E}_{l}[M_{f}^{(0)}]\) and \(\mathbb{E}_{l}[M_{f}^{(1)}]\) represent the average state the adversary would get if they sampled many ciphertexts using random keys. 

\medskip

\begin{theorem}\label{thm:exp}
The expectations of the original and the anamorphic states are
\begin{equation}
\mathbb{E}_{l,d_{1},k}[M_{f}^{(0)}] =\frac{1}{2^{d_1 + 1}} I_{2^{d_1 + 1}},
\end{equation}
where $I_{2^{d_1 + 1}}$ is the identity matrix of dimension $2^{d_1 + 1}$ and after considering expectation \(\mathbb{E}_{k}[M_{o}']\) and \(\mathbb{E}_{k'}[M_{c}'']\) separately we get,
\begin{equation}
\mathbb{E}_l\bigl[M_f^{(1)}\bigr]
=\alpha\bigl(M_a^{(1)}\bigr)\, I + \beta\bigl(M_a^{(1)}\bigr)\, J
=\frac{2^{d_1+1} - 1-\frac{2}{\eta}}{2^{d_1+1}(2^{d_1+1}-1)} I \;+\; \frac{2/\eta}{2^{d_1+1}(2^{d_1+1}-1)} J,
\end{equation}
where the matrix \(J\) is such that \(J_{i,j}=1\) for all \(i,j \in [2^{d_{1}+1}]\).

\smallskip

\noindent Hence the trace distance between the expectations is \begin{equation}
D(\mathbb{E}_{l}[M_{f}^{(1)}],\mathbb{E}_{l}[M_{f}^{(0)}])=\frac{1}{\eta\,2^{d_1}}
\end{equation}
which is less than \(\mathsf{negl}(\lambda)\),

\noindent and 

\begin{equation}
D(\mathbb{E}_{l}[M_{f}^{(1)}],\mathbb{E}_{l}[M_{f}^{(1)'}])=0.
\end{equation}

for any two density matrices \(M_{f}^{(1)}\) and \(M_{f}^{(1)'}\), when the randomization is taken over two different keys \(l,d_{1},d_{2},k,k',\eta\) and \(\Tilde{l},\Tilde{d_{1}},\Tilde{d_{2}},\Tilde{k},\Tilde{k'},\Tilde{\eta}\) from the same set \(\operatorname{Sym}(n) \times [2^{d_{1}+1}] \times [2^{d_{1}+1}] \times \{0,1\}^{d_{1}} \times  \{0,1\}^{d_{2}} \times \mathcal{J}\), with uniform distribution.
\end{theorem}

\begin{proof}
We recall that the Quantum One-Time Pad (\(\mathsf{QOTP}\)) encryption of an $n$-qubit state $\rho$ with key $k \in \{0,1\}^{2n}$ is given by
\(
\mathsf{QOTPEnc}(\rho, k) = U_k \rho U_k^\dagger,
\)
where
\(
U_k = \bigotimes_{i=1}^{n} X^{k_{2i - 1}} Z^{k_{2i}}.
\) The expectation over all possible keys $k$ is
\begin{equation}
\mathbb{E}_{k}[ \mathsf{QOTPEnc}(\rho, k) ] = \frac{1}{2^{2n}} \sum_{k \in \{0,1\}^{2n}} U_k \rho U_k^\dagger = \frac{I_{2^n}}{2^n} \operatorname{Tr}(\rho) = \frac{I_{2^n}}{2^n},
\end{equation}
since $\operatorname{Tr}(\rho) = 1$.

\medskip

Let $M_o \in \mathcal{D}((\mathbb{C}^2)^{\otimes d_1})$ be an arbitrary $d_1$-qubit state (density matrix) representing the original matrix. The \(\mathsf{QOTP}\) encryption of $M_o$ with key $k \in \{0,1\}^{2d_1}$ is
\(
M_o' = U_k M_o U_k^\dagger.
\)

We compute the expectation over all keys $k$,
\begin{equation}\label{eq:expMo}
\mathbb{E}_k[M_o'] = \frac{1}{2^{2d_1}} \sum_{k \in \{0,1\}^{2d_1}} U_k M_o U_k^\dagger.
\end{equation}
Using the properties of Pauli operators and the fact that the set $\{ U_k \}_{k \in \{0,1\}^{2d_1}}$ forms an orthonormal basis for operators on $(\mathbb{C}^2)^{\otimes d_1}$ (up to normalization), we can express $M_o$ in terms of Pauli operators
\begin{equation}
M_o = \sum_{P \in \mathcal{P}} c_P P,
\end{equation}
where the sum is over all $ d_1$-qubit Pauli operators $P$, and $c_P = \frac{1}{2^{d_1}} \operatorname{Tr}(P M_{o} )$.
Then, we have
\begin{equation}
\mathbb{E}_k[M_o'] = \frac{1}{2^{2d_1}} \sum_{k \in \{0,1\}^{2d_1}} U_k \left( \sum_{P \in \mathcal{P}} c_P P \right) U_k^\dagger = \sum_{P \in \mathcal{P}} c_P \left( \frac{1}{2^{2d_1}} \sum_{k \in \{0,1\}^{2d_1}} U_k P U_k^\dagger \right).
\end{equation}
Note that for any Pauli operator $P$ (excluding the identity), we have
\(
\frac{1}{2^{2d_1}} \sum_{k\in \{0,1\}^{2d_1}} U_k P U_k^\dagger = 0.
\)
This is because the conjugation of $P$ by $U_k$ effectively randomizes $P$ over all possible Pauli operators, and their average is zero unless $P$ is the identity operator, and for $P = I$, we have
\(
\frac{1}{2^{2d_1}} \sum_{k\in \{0,1\}^{2d_1}} U_k I U_k^\dagger = I.
\)
Therefore,
\(
\mathbb{E}_k[M_o'] = c_I I,
\)
where
\(
c_I = \frac{1}{2^{d_1}} \operatorname{Tr}(I M_o ) = \frac{1}{2^{d_1}} \operatorname{Tr}( M_o ) = \frac{1}{2^{d_1}}.
\)
Thus,
\(
\mathbb{E}_k(M_{o}') = \frac{I_{2^{d_1}}}{2^{d_1}},
\)
and 
hence,
\(
\mathbb{E}_{k'}[M_c''] = \, \frac{I_{2^{d_2}}}{2^{d_2}}\, .
\)
\medskip

Similarly, for the covert message $M_{c} \in \mathcal{D}((\mathbb{C}^2)^{\otimes d_2})$, the encrypted state is
\(
M_{c}' = U_{k'} M_{c} U_{k'}^\dagger,
\)
with $k' \in \{0,1\}^{2d_2}$ and we have
\(
\mathbb{E}_{k'}[M_c'] = \frac{I_{2^{d_2}}}{2^{d_2}}.
\)

\medskip

We compute
\(
\mathbb{E}_l[ U_{\sigma_l} X U_{\sigma_l}^\dagger ],
\)
for a fixed matrix $X$ of compatible dimension. The set of all permutation matrices forms a group under multiplication. 

\medskip

Then we have,
\begin{equation}
\mathbb{E}_l[ M_{f}^{(1)} ] = \frac{1}{(2^{d_{1}+1})!} \sum_{\sigma_l \in \operatorname{Sym}({2^{d_{1}+1}})} U_{\sigma_{l}}  M_{f}^{(1)} U_{\sigma_{l}}^\dagger.
\end{equation}

\noindent The representation of \(\operatorname{Sym}(n)\) is defined by \( \pi: \operatorname{Sym}(n) \longrightarrow \mathcal{U}(\mathbb{C}^n) \) by
\begin{equation}
\pi(\sigma) \, e_i = e_{\sigma(i)}, \quad \text{for} \quad i=1,\ldots,n,
\end{equation}
where \( \{e_i\}_{i=1}^n \) is the standard orthonormal basis of \( \mathbb{C}^n \) and \(\mathcal{U}(\mathbb{C}^n)\) denotes the group of unitary operators on \( \mathbb{C}^n \), [See \cite{fulton2013representation, serre1977linear}]. It is well known that this representation is \textit{reducible}. 

\medskip

\noindent Let
\[
\mathcal{A} = \{ A \in \mathcal{L}(\mathbb{C}^n) : \, A\,\pi(\sigma) = \pi(\sigma) A,\; \forall\, \sigma\in \operatorname{Sym}(n) \}
\]
be the centralizer or the commutant of the representation \(\pi\).

\medskip
By the \textit{Double Commutant Theorem} and by applying \textit{Schur's lemma}, \(\mathcal{T}_{\operatorname{Sym}(n)}(\Phi) \in \mathcal{A}\) and \(\mathcal{A}=span\{I,\, J\}\), where \( I \) is the \( n\times n \) identity matrix and \( J \) is the \( n \times n \) matrix such that \( J_{ij}=1 \) for all \( i,j \in [n] \). Then, for any matrix \( \Phi \in\mathcal{L}(\mathbb{C}^n) \) we have,
\begin{equation}
\mathcal{T}_{\operatorname{Sym(n)}}(\Phi) = \alpha(\Phi)\, I + \beta(\Phi)\, J,
\end{equation}
for some scalars \(\alpha(\Phi),\beta(\Phi) \in \mathbb{C}\) that depend linearly on \(\Phi\). We will only denote \(\mathcal{T}_{\operatorname{Sym}(n)}(\Phi)\) by only \(\mathcal{T}(\Phi)\) if it is understood that we are considering the representation of \(\operatorname{Sym}(n)\) only.

\medskip

\noindent For any matrix \( \Phi\in\mathcal{L}(\mathbb{C}^n) \), the twirl
\(
\mathcal{T}(\Phi) = \mathbb{E}_l\bigl[U_{\sigma_{l}}\,\Phi\,U_{\sigma_{l}}^\dagger\bigr].
\)
We compute the coefficients \( \alpha(\Phi) \) and \( \beta(\Phi) \) in terms of two linear invariants of \( \Phi \):
\begin{equation}
T(\Phi) := \operatorname{Tr}(\Phi) \quad \text{and} \quad S(\Phi) := \sum_{i=1}^{n}\sum_{j=1}^n \Phi_{ij}.
\end{equation}
Taking the trace of \( \mathcal{T}(\Phi) \), we obtain
\begin{equation}
\alpha(\Phi) + \beta(\Phi) = \frac{T(\mathcal{T}(\Phi))}{n}.
\end{equation}
On the other hand, the sum of all entries of \( \mathcal{T}(\Phi) \) is
\begin{equation}
S\bigl(\mathcal{T}(\Phi)\bigr) = \alpha(\Phi) \, n + \beta(\Phi) \, n^2.
\end{equation}
The twirling map is \textit{trace–preserving}. As the permutation conjugation simply reorders the entries, the permutation twirling \textit{preserves the sum of all matrix elements}. Therefore,
\begin{equation}
S\bigl(\mathcal{T}(\Phi)\bigr) = S(\Phi) \quad \text{and} \quad T\bigl(\mathcal{T}(\Phi)\bigr) = T(\Phi).
\end{equation}
Thus,
\begin{equation}\label{eq:92}
\alpha(\Phi) \, n + \beta(\Phi) \, n^2 = S(\Phi) \quad \text{and} \quad n(\alpha(\Phi) + \beta(\Phi)) = T(\Phi). 
\end{equation}
Solving both the equations \ref{eq:92}, we get,
\begin{equation}
\beta(\Phi) = \frac{S(\Phi) - T(\Phi)}{n(n-1)} \quad \text{and} \quad \alpha(\Phi) = \frac{nT(\Phi) - S(\Phi)}{n(n-1)}.
\end{equation}

\medskip

\noindent Now, considering \(\Phi=M_{a}^{(0)}\), we get 
\begin{equation}
\alpha(M_{a}^{(0)})= \frac{2^{d_{1}+1} - S(\Phi)}{2^{d_{1}+1}(2^{d_{1}+1}-1)}.
\end{equation}
Since \( M_a^{(0)} \) is diagonal and constant along the diagonal, its off-diagonal entries are zero so that
\begin{equation}
S\bigl(M_a^{(0)}\bigr) = \sum_{i=1}^{2^{d_{1}+1}}\sum_{j=1}^{2^{d_{1}+1}} (M_a^{(0)})_{ij} = \sum_{i=1}^{2^{d_{1}+1}} \frac{1}{2^{d_{1}+1}} = 1.
\end{equation}
and hence, \(\alpha(M_{a}^{(0)})=\frac{1}{2^{d_1+1}}\) and \(\beta\bigl(M_a^{(0)}\bigr) = 0.\) Therefore, 
\begin{equation}
\mathbb{E}_l\bigl[M_f^{(0)}\bigr] = \alpha\bigl(M_a^{(0)}\bigr) \, I + \beta\bigl(M_a^{(0)}\bigr) \, J
= \frac{I}{2^{d_1+1}}.
\end{equation}

\medskip

\noindent Now, consider the quantum density matrix \(\Phi=M_{a}^{(1)}\). Then, after taking expectation \(\mathbb{E}_{k}[M_{o}']\) and \(\mathbb{E}_{k'}[M_{c}'']\) separately, we get, 
\(
T\bigl(M_a^{(1)}\bigr)=1 \quad \text{and} \quad S\bigl(M_a^{(1)}\bigr)= \Big(1 + \frac{2}{\eta}\Big).
\)

\noindent Computing the coefficients \(\alpha(\Phi)\) and \(\beta(\Phi)\), we get,
\begin{equation}
\alpha\bigl(M_a^{(1)}\bigr)
= \frac{2^{d_1+1} - 1-\frac{2}{\eta}}{2^{d_1+1}(2^{d_1+1}-1)} \quad \text{and} \quad \beta\bigl(M_a^{(1)}\bigr) = \frac{2/\eta}{2^{d_1+1}(2^{d_1+1}-1)}.
\end{equation}
Therefore,
\begin{equation}
\mathbb{E}_l\bigl[M_f^{(1)}\bigr]
=\alpha\bigl(M_a^{(1)}\bigr)\, I + \beta\bigl(M_a^{(1)}\bigr)\, J
=\frac{2^{d_1+1} - 1-\frac{2}{\eta}}{2^{d_1+1}(2^{d_1+1}-1)} I \;+\; \frac{2/\eta}{2^{d_1+1}(2^{d_1+1}-1)} J.
\end{equation}

\noindent Now, computing difference we get,
\begin{equation}
\mathbb{E}_l\bigl[M_f^{(1)}\bigr]- \mathbb{E}_l\bigl[M_f^{(0)}\bigr]
= \frac{2}{\eta 2^{d_1+1}(2^{d_1+1}-1)} (J-I).
\end{equation}

\medskip

\noindent It is well-known that, one eigen value of the matrix \((J-I)\) is \((2^{d_{1}+1}-1)\) and the eigenvalue \( -1 \) has multiplicity \((2^{d_{1}+1}-1)\).
\medskip

\noindent Therefore, the trace distance is
\begin{align}
D\Bigl(\mathbb{E}_l\bigl[M_f^{(1)}\bigr], \mathbb{E}_l\bigl[M_f^{(0)}\bigr]\Bigr)
&=\frac{1}{2}\,\|\mathbb{E}_l\bigl[M_f^{(1)}\bigr]- \mathbb{E}_l\bigl[M_f^{(0)}\bigr]\|_1 \notag\\
&=\frac{1}{2} \Big[ \frac{2}{\eta 2^{d_1+1}(2^{d_1+1}-1)} \|(J-I)\|_{1}\Big] \notag\\
&=\frac{1}{2}\cdot \frac{4}{\eta\,2^{d_1+1}(2^{d_1+1}-1)} (2^{d_1+1}-1)\notag\\
&=\frac{1}{\eta\,2^{d_1}}.
\end{align}

\noindent Hence, \(D\Bigl(\mathbb{E}_l\bigl[M_f^{(1)}\bigr], \mathbb{E}_l\bigl[M_f^{(0)}\bigr]\Bigr) < \mathsf{negl}(\lambda).\)

\medskip

Choosing two different keys \((l,d_{1},d_{2},k,k',\eta)\) and \((\Tilde{l},\Tilde{d_{1}},\Tilde{d_{2}},\Tilde{k},\Tilde{k'},\Tilde{\eta})\) from the same set \(\operatorname{Sym}(n) \times [2^{d_{1}+1}] \times [2^{d_{1}+1}] \times \{0,1\}^{d_{1}} \times  \{0,1\}^{d_{2}} \times \mathcal{J}\), with uniform distribution, it is easy to see that \(D(\mathbb{E}_{l}[M_{f}^{(1)}],\mathbb{E}_{l}[M_{f}^{(1)'}])=0.\)

\end{proof}

\medskip

\noindent Now, we prove the computational indistinguishability of the original and anamorphic ciphertexts.

\medskip

\begin{theorem}\label{thm:ind}
The original and anamorphic ciphertexts \(M_{f}^{(0)}\) and \(M_{f}^{(1)}\) are computationally indistinguishable.
\end{theorem}

\begin{proof}
In the real game, the adversary receives the original ciphertext
\begin{equation}
M_{f}^{(0)} = U_{\sigma_{l}}\begin{pmatrix} \frac{1}{2} M_o' & 0 \\ 0 & \frac{1}{2} M_o' \end{pmatrix}U_{\sigma_{l}}^{\dagger}.
\end{equation}
and in the anamorphic game, the adversary receives the following anamorphic ciphertext
\begin{equation}
M_{f}^{(1)}= U_{\sigma_l} \begin{pmatrix} \frac{1}{2} M_o' & \frac{1}{\eta}M_c'' \\ (\frac{1}{\eta}M_c'')^\dagger & \frac{1}{2} M_o' \end{pmatrix} U_{\sigma_l}^\dagger.
\end{equation}
We know the following inequality between adversarial advantage and the trace distance \(
\operatorname{Adv}_{\mathcal{D}}(\lambda) \leq D(M_{f}^{(0)}, M_{f}^{(1)})
\) \ref{lemma:advbd}.

\medskip

\noindent Now, we compute the trace distance between the real and the anamorphic ciphertexts,

\begin{align}
    D(M_{f}^{(0)}, M_{f}^{(1)}) &=  \frac{1}{2} \norm{M_{f}^{(0)} - M_{f}^{(1)} }_{1} \notag\\
    &= \frac{1}{2}\norm{U_{\sigma_{l}}\begin{pmatrix} \frac{1}{2} M_o' & 0 \\ 0 & \frac{1}{2} M_o' \end{pmatrix}U_{\sigma_{l}}^{\dagger} - U_{\sigma_l} \begin{pmatrix} \frac{1}{2} M_o' & \frac{1}{\eta}M_c'' \notag\\ \frac{1}{\eta}(M_c'')^\dagger & \frac{1}{2} M_o' \end{pmatrix} U_{\sigma_l}^\dagger}_{1} \notag\\
    &= \frac{1}{2}\norm{ U_{\sigma_l} \begin{pmatrix} 0 & -\frac{1}{\eta}M_c'' \\ -\frac{1}{\eta}(M_c'')^\dagger & 0 \end{pmatrix} U_{\sigma_l}^\dagger }_1 \notag\\
    &= \frac{1}{2}\norm{ \begin{pmatrix} 0 & -\frac{1}{\eta}M_c'' \\ -\frac{1}{\eta}(M_c'')^\dagger & 0 \end{pmatrix} }_{1} (\text{
since $U_{\sigma_l}$ is unitary and the trace norm is unitary invariant}).
\end{align}
As, \(M_{c}\) is a Hermitian, positive semi-definite matrix with \(\operatorname{Tr}(M_{c})=1\), after encrypting with the key \(k'\) the density matrix \(M_{c}'=\mathsf{QOTP}(M_{c}, k')\), remains as Hermitian, positive semi-definite and preserves the norm. Hence, \(\frac{1}{\eta}M_c'' = \frac{1}{\eta}(M_c'')^{\dagger}\). 

\noindent Denote the matrix \(\begin{pmatrix} 0 & -\frac{1}{\eta}M_c'' \\ -\frac{1}{\eta}(M_c'')^\dagger & 0 \end{pmatrix} \) by \(A\). 

\smallskip

\noindent Then 
\begin{equation}
A^{2}= \begin{pmatrix} \frac{1}{\eta^{2}}(M_c'')^{2} & 0 \\ 0 & \frac{1}{\eta^{2}}(M_c'')^{2} \end{pmatrix}.
\end{equation}

\noindent The trace norm of $A$ is
\begin{equation}
\| A \|_1 = \operatorname{Tr}(\sqrt{A^{2}})=\frac{2}{\eta}.\operatorname{Tr}(M_{c}'')=\frac{2}{\eta}.
\end{equation}

\noindent Thus, the trace distance is
\begin{equation}
D(M_f^{(0)}, M_f^{(1)}) = \dfrac{1}{2} \| M_f^{(0)} - M_f^{(1)} \|_1 = \dfrac{1}{2} \| A \|_1 = \dfrac{1}{2} \cdot \dfrac{2}{\eta} = \dfrac{1}{\eta} < \mathsf{negl}(\lambda).
\end{equation}
\end{proof}

\medskip
Note that as we have remarked earlier \ref{re:dom} that the algorithm \(\mathsf{DOM}\) can be applied to both \(M_{f}^{(0)}\) and \(M_{f}^{(1)}\). The dictator and the players can decryprt the original ciphertext \(M_{o}\) from both the ciphertexts \(M_{f}^{(0)}\) and \(M_{f}^{(1)}\) exactly using the same \(\mathsf{DOM}\) algorithm. Therefore, both the ciphertexts \(M_{f}^{(0)}\) and \(M_{f}^{(1)}\) are indistinguishable to the dictator.
 
\medskip

\if 0

\begin{theorem}
    The associated keys for both the original and the anamorphic messages are also computationally indistinguishable.
\end{theorem}

\begin{proof}
Let \(\mathcal{R}=\{0,1\}^{2d_{1}} \times \{1,\ldots,2^{d_{1}+1}\} \times \text{Sym}(2^{d_{1}+1})\) and \(\mathcal{R}'=\{0,1\}^{2d_{1}} \times \{0,1\}^{2d_{2}} \times \{1,\ldots,2^{d_{1}+1}\} \times \{1,\ldots,2^{d_{1}+1}\} \times \text{Sym}(2^{d_{1}+1}) \times \mathcal{J}\).
We have chosen the two keys \(k\) and \(k'\) uniformly randomly from the \(\mathcal{R}\) and \(\mathcal{R}'\) respectively. Therefore,
\[
\left| \Pr_{X \gets \mathcal{R}} \left[X=k \right] - \Pr_{X' \gets \mathcal{R}'} \left[X'=k' \right] \right| = \left|\frac{1}{2^{3d_{1}+1}} \cdot \frac{1}{2^{d_{1}+1}!}-\frac{1}{2^{4d_{1}+2d_{2}+2}} \cdot \frac{1}{|\mathcal{J}|} \cdot \frac{1}{2^{d_{1}+1}!}\right|= \frac{1}{2^{d_{1}+1}!} \frac{1}{2^{3d_{1}+1}}\Big[1-\frac{1}{2^{d_{1}+2d_{2}+1}} \frac{1}{|\mathcal{J}|}\Big].
\]
where \(X\) and \(X'\) are two random variables defined on the sets \(\mathcal{R}\) and \(\mathcal{R}'\) respectively.
\end{proof}

\fi

\medskip
\noindent In quantum mechanics, fidelity is a metric used to quantify the similarity or closeness between quantum states. A high fidelity value indicates that the states are nearly identical. In our case, as a consequence of negligible trace distance, it is easy to show that the original and anamorphic quantum states exhibit a high fidelity. Consequently, this establishes that, in our case, an adversary or dictator is computationally unable to distinguish between the ciphertexts \(M_{f}^{(0)}\) and \(M_{f}^{(1)}\), making it infeasible to identify which corresponds to the original ciphertext.

\medskip

\begin{theorem}\label{thm:fid}
    The fidelity between the original and the anamorphic states is \[
F(M_{f}^{(0)}, M_{f}^{(1)}) \geq \Bigg(1 - \dfrac{1}{\eta}\Bigg)
\] 
\end{theorem}
indicating that the two states are nearly indistinguishable for large \(\eta\).

\begin{proof}
    By the Fuchs-van de Graaf inequality, and Theorem \ref{thm:ind}, \(\left(1-F(M_{f}^{(1)},M_{f}^{(0)})\right) \leq \frac{1}{\eta}\). \\
    Hence, \[F(M_{f}^{(1)},M_{f}^{(0)}) \geq \left(1-\frac{1}{\eta}\right).\]     
\end{proof}

\medskip

\noindent Now we describe the communication procedure between Alice and Bob under dictatorial supervision. 

\begin{algorithm}[H]
\caption{Transmission Protocol under Dictatorial Supervision (TPDS)}
\label{alg:transmission-dictator}
\begin{algorithmic}[1]

\STATE \textbf{Input:} Anamorphic state \(M_f^{(1)} \in \mathcal{L}((\mathbb{C}^2)^{\otimes(d_1+1)})\), dimensions \(d_1, d_2\), keys \(k, k'\), and permutation matrix \(U_{\sigma_l}\).

\STATE \textbf{Output:} The original message \(M_o\) for both Bob and the dictator, and the covert message \(M_c\) exclusively for Bob.

\STATE \textbf{Step 1: Alice's Transmission to Bob:}
\STATE \hspace{1cm} Alice generates the anamorphic encrypted state \(M_f^{(1)}\) using the encryption process.
\STATE \hspace{1cm} Alice transmits:
\[
M_f^{(1)}, \quad (l, d_1, d_2, k, k',\eta)
\]
securely to Bob.

\STATE \textbf{Step 2: Bob's Decryption:}
\STATE \hspace{1cm} Bob receives \(M_f^{(1)}\) and the keys \(l, d_1, d_2, k, k',\eta\).
\STATE \hspace{1cm} Bob performs the following operations:
\begin{enumerate}
    \item Run the \textbf{Decryption of Original Message (\(\mathsf{DOM}\))} algorithm with the key \(k\) to recover the original message \(M_o\).
    \item Run the \textbf{Covert Decryption of Anamorphic Message (\(\mathsf{DCM}\))} algorithm with key \(k'\) to recover the covert message \(M_c\).
\end{enumerate}

\STATE \textbf{Step 3: Bob's Forwarding to the Dictator:}
\STATE \hspace{1cm} Bob forwards to the dictator the following information:
\[
M_f^{(1)}, \quad (l, d_1, k).
\]

\STATE \textbf{Step 4: Dictator's Decryption:}
\STATE \hspace{1cm} The dictator receives \(M_f^{(1)}\) and keys \(l, d_1, k\).
\STATE \hspace{1cm} The dictator runs the \textbf{Decryption of Original Message (\(\mathsf{DOM}\))} algorithm using key \(k\) to recover the original message \(M_o\).

\if 0

\STATE \textbf{Step 5: Security Guarantee:}
\STATE \hspace{1cm} The trace distance between \(M_f^{(0)}\) and \(M_f^{(1)}\) is negligible and fidelity is very high:
\[
D(M_f^{(0)}, M_f^{(1)}) = \frac{1}{\eta} < negl{(\lambda)} \quad \text{and} \quad F(M_f^{(0)}, M_f^{(1)}) \geq \Bigg(1-\frac{1}{\eta} \Bigg)
\]
\STATE \hspace{1cm} Thus, the dictator cannot distinguish whether a covert message \(M_c\) was embedded in the transmitted state.

\fi

\STATE \textbf{Step 6: Output:}
\STATE \hspace{1cm} Bob receives both \(M_o\) and \(M_c\).
\STATE \hspace{1cm} The dictator receives only \(M_o\).
\end{algorithmic}
\end{algorithm}

\medskip
\begin{lemma}[Pauli Twirl / \(1\)-Design Property, \cite{wilde2017quantum}]
\label{lem:pauli-twirl}
Let $\mathcal{P}_d$ be the Pauli group on $d$ qubits. For any density operator $\rho$ on $\mathcal{H} \cong (\mathbb{C}^2)^{\otimes d}$, averaging over the group yields the maximally mixed state
\[
\frac{1}{|\mathcal{P}_d|} \sum_{U \in \mathcal{P}_d} U \rho U^\dagger \;=\; \Tr(\rho) \frac{I}{2^d}.
\]
Consequently, the channel $\Phi(\rho) = \mathbb{E}_k [\QOTPEnc(\rho, k)]$ maps every state $\rho$ (with $\Tr(\rho)=1$) to the fixed state $I / 2^d$, effectively destroying all information about $\rho$.
\end{lemma}

\smallskip
\begin{lemma}\label{lem:covert-avg}
For any $M_c\in\mathcal{D}(\mathcal{H}_{M_c})$, uniform $k'\in\{0,1\}^{2d_2}$, and the fixed padding isometry $V$,
\[
\mathbb{E}_{k'}\big[V\,\mathsf{QOTPEnc}(M_c,k')\,V^\dagger\big]\;=\;2^{-d_2}\,\Pi_V,
\qquad \Pi_V=V V^\dagger,
\]
which is independent of $M_c$.
\end{lemma}

\noindent
The lemma is immediate from the Pauli twirl identity~\eqref{lem:pauli-twirl} and \ref{eq:pad2}.

\begin{lemma}\label{lem:Ma-avg}
Fix $\eta \in \mathbb{Z}^{+}$ according to Theorem \ref{thm:main}. For any inputs $(M_o,M_c)$, sampling fresh uniform $(k,k')$ yields
\begin{align}
\mathbb{E}_{k,k'}\big[M_a(M_o,M_c;k,k')\big]
&=\;
\lvert 0\rangle\langle 0\rvert\otimes \tfrac{1}{2}\cdot 2^{-d_1} I_{2^{d_1}}
\;+\;
\lvert 0\rangle\langle 1\rvert\otimes \tfrac{1}{\eta}\cdot 2^{-d_2}\,\Pi_V
\\[-0.25em]
&\hspace{1.25cm}
+\;
\lvert 1\rangle\langle 0\rvert\otimes \tfrac{1}{\eta}\cdot 2^{-d_2}\,\Pi_V
\;+\;
\lvert 1\rangle\langle 1\rvert\otimes \tfrac{1}{2}\cdot 2^{-d_1} I_{2^{d_1}},\nonumber
\end{align}
which is independent of $M_o$ and $M_c$.
\end{lemma}

\begin{proof}
Applying \ref{eq:expMo} to both diagonal blocks of \eqref{eq:Ma}, and Lemma \ref{lem:covert-avg} to the off-diagonal blocks (together with Hermitian conjugation for the $(1,0)$ block). Linearity of expectation yields the claim.
\end{proof}

\subsection{Perfect \texorpdfstring{\textsf{qIND--qCPA}}{qIND-qCPA} Security}

\begin{theorem}[Information-theoretic \textsf{qIND--qCPA} security of $\mathsf{QAE}$]\label{thm:main-qind-qcpa}
Consider the symmetric-key construction $\mathsf{QAE}$ that, for each encryption, samples uniform keys $k\in\{0,1\}^{2d_1}$, $k'\in\{0,1\}^{2d_2}$ and an uniform permutation $\sigma_l \in \mathrm{Sym}(2^{d_{1}+1})$, and outputs the ciphertext~\eqref{eq:Mf}.
Then for any (possibly unbounded) quantum adversary $\mathcal{A}$ as in Definition~\ref{def:qind-qcpa},
\[
\mathrm{Adv}^{\mathrm{qIND\text{-}qCPA}}_{\mathsf{QAE}}(\mathcal{A})\;=\;0.
\]
\end{theorem}

\begin{proof}
Let $\mathcal{E}$ be the oracle. We show that $\mathcal{E}$ is a \emph{constant} channel, i.e.\ it maps every input to the same density operator on $\mathcal{H}_{RM}$, independent of $(M_o,M_c)$.

Fix an arbitrary input state $\rho_{E\,M_o\,M_c}$ for a single query (the argument is identical for each oracle call, including the challenge). Let \(\mathrm{Ad}_{U}(\rho)=U \rho U^{\dagger}\). By definition,
\[
(I_E\otimes \mathcal{E})(\rho_{E\,M_o\,M_c})
\;=\;
\mathbb{E}_{k,k',\sigma_l}\Big[\,
(I_E\otimes \mathrm{Ad}_{U_{\sigma_l}})\!\left(
(I_E\otimes \Phi_{k,k'})(\rho_{E\,M_o\,M_c})
\right)\Big],
\]
where $\Phi_{k,k'}$ denotes the deterministic (CPTP) pre-permutation unitary map that computes $M_a(M_o,M_c;k,k')$ as in~\eqref{eq:Ma}, traces out $(M_o,M_c)$, and outputs the register $RM$. 
By linearity, taking the partial trace over $E$ yields
\[
\mathrm{Tr}_E\!\left[(I_E\otimes \mathcal{E})(\rho_{E\,M_o\,M_c})\right]
\;=\;
\mathbb{E}_{k,k',\sigma_l}\Big[\,
\mathrm{Ad}_{U_{\sigma_l}}\!\left(
\mathrm{Tr}_E\!\left[(I_E\otimes \Phi_{k,k'})(\rho_{E\,M_o\,M_c})\right]
\right)\Big].
\]
Since $\Phi_{k,k'}$ acts only on $M_o,M_c$, $\mathrm{Tr}_E\!\left[(I_E\otimes \Phi_{k,k'})(\rho_{E\,M_o\,M_c})\right]=\Phi_{k,k'}\!\left(\rho_{M_o\,M_c}\right)$.
Applying Lemma~\ref{lem:Ma-avg} and then averaging over $\sigma_l$ we get
\begin{align*}
\Xi_\eta &=\mathrm{Tr}_E\!\left[(I_E\otimes \mathcal{E})(\rho_{E\,M_o\,M_c})\right] \\
&= \lvert 0\rangle\langle 0\rvert\otimes \tfrac{1}{2}\cdot 2^{-d_1} I_{2^{d_1}}
\;+\;
\lvert 0\rangle\langle 1\rvert\otimes \tfrac{1}{\eta}\cdot 2^{-d_2}\,\Pi_V
\;+\;
\lvert 1\rangle\langle 0\rvert\otimes \tfrac{1}{\eta}\cdot 2^{-d_2}\,\Pi_V
\;+\;
\lvert 1\rangle\langle 1\rvert\otimes \tfrac{1}{2}\cdot 2^{-d_1} I_{2^{d_1}}.
\end{align*}
where $\Xi_\eta$ is a fixed density operator on $\mathcal{H}_{RM}$ that depends only on $(d_1,d_2,\eta)$ and the fixed isometry $V$, but \emph{not} on $\rho_{M_o\,M_c}$.
This state is positive semidefinite for all $\eta\ge 2$ (Schur complement with $A=\tfrac{1}{2}\cdot 2^{-d_1} I$ and $B=\tfrac{1}{\eta}\cdot 2^{-d_2}\Pi_V$).
Moreover, $(I_E\otimes \mathcal{E})(\rho_{E\,M_o\,M_c})=\rho_E\otimes \Xi_\eta$ because $\mathcal{E}$ acts trivially on~$E$. 
Thus, for \emph{every} oracle query (including the challenge), the returned register equals $\Xi_\eta$ and is independent of the inputs that the adversary selected for that query. Consequently, in the \textsf{qIND--qCPA} experiment, the overall joint final state of the adversary is \emph{identical} in the two worlds $b=0$ and $b=1$; any measurement yields the same outcome distribution, and the distinguishing advantage is $\mathrm{Adv}=0$.
\end{proof}

\medskip

\section{Anamorphic Secret-Sharing}\label{sec:ass}
In this section, our main goal is to define quantum anamorphic secret-sharing. First, we will review basic notions of quantum secret-sharing schemes. Then we will propose a definition of a quantum anamorphic secret-sharing along with our construction.

With abuse of notation, we denote $P\subseteq [n]$ to be a set of players. Let $\mathcal{S}$ be a set of secrets. Let $\mathcal{R}$ be a finite set of random strings, $\mu: \mathcal{R} \longrightarrow \mathbb{R}$ be a probability distribution function, and $\forall j,\; 1 \leq j \leq n,\; \mathcal{S}_j$ be the domain of shares of $j$-th player. 

\medskip

In a secret-sharing scheme, we want to share a secret among \(n\) players so that 

\begin{itemize}
\item only the authorized set of players can reconstruct the secret and 
\item the unauthorized set of players cannot reconstruct the secret. 
\end{itemize}

\medskip

\begin{definition}\cite{ccakan2023computational}\label{def:P}
   A sequence of monotone functions \((f_n)_{n \in \mathbb{Z}^{+}}\), where each function \( f_n\colon \{0,1\}^n \longrightarrow \{0,1\} \) is computable by a family of monotone circuits of size polynomial in \( n \), based on the existence of one-way functions, is defined as belonging to the class \(\mathsf{monotone} \ \mathsf{P}\).
\end{definition}

\medskip

We refer to the survey article by \cite{beimel2011secret} for detailed exposition. We can define the \textit{access structure} by a monotone function \(f\colon\{0,1\}^{n} \longrightarrow \{0,1\}\) with a set \(P \subseteq [n]\) defined to be \text{authorized} if and only if \(f(v^{P})=1\) where \(v^{P} \in \{0,1\}^{n}\) the characteristic vector of \(P\) satisfying \(v^{P}_{i}=1\) when \(i \in P\) \cite{ccakan2023computational}. We denote the \(t\) out of \(n\) threshold function by \(T^{t}_{n}\) such that \(T^{t}_{n}(P)=1\) if and only if \(|P| \geq t\).

\medskip

\noindent \(\bullet\) \textbf{Key difficulty in quantum secret sharing due to no-cloning theorem:} The \textit{no-cloning theorem} \cite{wootters1982nocloning} prevents copying unknown quantum states, which makes quantum secret sharing difficult. This prevents fundamental strategies like sharing components with multiple players. In other words, no quantum secret-sharing techniques realize the OR function. The methodologies behind many important classical conclusions cannot be readily transferred to the quantum setting, therefore lifting them requires new thinking. We now detail our contributions and formal outcomes. A small generic compiler using hybrid encoding from classical to quantum secret sharing that we construct and analyze yields our results. However, this problem was addressed in the work of Chien \cite{chien2020augmented}, and also in the work of \c{C}akan et. al. \cite{ccakan2023computational}.

\medskip

\noindent \(\bullet\) \textbf{Heavy monotone functions:}
This concept was introduced in the work of \c{C}akan et. al. \cite{ccakan2023computational}. The no-cloning theorem limits quantum secret sharing systems to \textit{no-cloning} monotone functions. These monotone functions \( f \) are defined such that \( f(P) = 1 \) implies \( f(\overline{P}) = 0 \), meaning the complement of an authorized set is unauthorized. The state-of-the-art share size for all no-cloning monotone functions \( f \) is the size of the smallest monotone span program computing \( f \), which can be very large even for “simple” no-cloning monotone functions in \(\mathsf{monotone \; P}\).

\begin{definition}(Heavy function \cite{ccakan2023computational})
    A monotone function \( f\colon \{0, 1\}^n \longrightarrow \{0, 1\} \) is said to be \( t \)-heavy if for any subset \( P \subseteq [n] \) where \( f(P) = 1 \), it holds that \( |P| \ge t \). For \( t \ge \lfloor n/2 \rfloor + 1 \), we say that \( f \) is \textit{heavy}.
\end{definition}
Note that \( t \)-heavy monotone functions are those with a \textit{minimum authorized set} size of at least \( t \). Note that a \( t \)-out-of-\( n \) threshold function is a type of \(t\)-heavy function.

\medskip
\begin{proposition}(\cite{ccakan2023computational})
    Let \(\mathsf{mSP}(f)\) and \(\mathsf{mC}(f)\) be the size of smallest monotone span program and monotone circuit for computing \(f\), respectively. Then, for every monotone function \(f\colon\{0,1\}^{n} \longrightarrow \{0,1\}\) there exist a heavy monotone function \(f\colon\{0,1\}^{2n} \longrightarrow \{0,1\}\) such that \(\mathsf{mSP}(f') \geq \frac{\mathsf{mSP(f')}}{2n}\) and \(\mathsf{mC}(f') \leq \mathsf{mC}(f)+n\). Also, whenever \(f \in \mathsf{mNP}\), then \(f' \in \mathsf{mNP}\).
\end{proposition}

\medskip

\begin{corollary}(\cite{ccakan2023computational})
    There exist heavy monotone functions in \texorpdfstring{\(\mathsf{monotone}\) \(\mathsf{P}\)}{monotone P} requiring monotone span programs of size \(\operatorname{exp}(n^{\Omega(1)})\).
\end{corollary}

\medskip

\noindent \(\bullet\) \textbf{Sharing multiple copies bypassing the no-cloning theorem:} Not all monotone functions can be implemented by conventional quantum secret sharing protocols due to the no-cloning theorem 
 and due to that, we need some additional assumptions to design quantum secret-sharing schemes for a wide range of monotone functions\cite{ccakan2023computational}. One solution to this problem is to assume that we have access to several copies of the quantum state that we want to share \cite{ccakan2023computational}. In the multiparty computations(\(\mathsf{MSP}\)), this makes sense. The classical description of their quantum input is already known to each player and they can create as many copies as they want. The work of \cite{chien2020augmented} considered a special case of threshold monotone functions and investiated the number of copies we would require to construct an efficient secret-sharing scheme realizing all monotone functions in \(\mathsf{monotone\;P}\)\cite{ccakan2023computational}. Chien showed without security proof that \(\max (1,n-2t+2)\) copies of the quantum secret are sufficient to construct a \(t\)-out-of-\(n\) quantum secret-sharing scheme \cite{chien2020augmented, ccakan2023computational}. 

\medskip

\subsection{Classical secret-sharing scheme}
In this section, we review some basic definitions of the theory of classical secret-sharing schemes. 

\medskip

\begin{definition}(Classical secret-sharing scheme \cite{ccakan2023computational, beimel2011secret})\label{def:ss} A classical perfect secret-sharing scheme realizing the monotone function \( f \colon \{0, 1\}^n \longrightarrow \{0, 1\} \) is a pair of functions \(\mathsf{SS}= (\mathsf{Share}, (\mathsf{Rec}_{P})_{P \subseteq [n]})\) 

where \( \mathsf{Share} : \mathcal{S} \times \mathcal{R} \longrightarrow \mathcal{S}_{[n]} \) and \( \mathsf{Rec}_P : \mathcal{S}_{P} \longrightarrow \mathcal{S} \) are deterministic functions satisfying the following properties for all \( P \subseteq [n] \):

\begin{itemize}
    \item \textbf{Correctness}: If \( f(P) = 1 \), then for all \( s \in \mathcal{S} \), 
    \[
    \Pr_{R {\gets} \mathcal{R}} \left[\mathsf{Rec}_P(\mathsf{Share}(s; R)_P) = s \right] = 1.
    \]

    \item \textbf{Perfect Privacy}: If f(P) = 0, then for all secrets \(s_{1}, s_{2} \in \mathcal{S} \) and share vectors \(v \in \mathcal{S}_P \), we have
    \[
    \Pr_{R \gets \mathcal{R}} \left[\mathsf{Share}(s_{1}; R)_P = v \right] = \Pr_{R \gets R} \left[\mathsf{Share}(s_{2}; R)_P = v \right].
    \]
\end{itemize}
\end{definition}

\medskip

\begin{definition}(Share size \cite{ccakan2023computational})
For a secret-sharing scheme \( \mathsf{SS} \) defined over the share domains \( \{\mathcal{S}_1, \ldots, \mathcal{S}_n\} \), the share size, denoted as \( \mathsf{size}(\mathsf{SS}) \), is given by  
\[
\mathsf{size}(\mathsf{SS}) = \sum_{i=1}^n \lceil \log |\mathcal{S}_i| \rceil.
\]  
This represents the total number of bits required to encode all shares in the scheme.
\end{definition}

\medskip

\begin{definition}(Statistical privacy for classical secrets \cite{ccakan2023computational})
A secret-sharing scheme \( \mathsf{SS} \) that realizes a monotone function \( f \) is said to be \( \varepsilon \)-\textit{statistically private} if, for any subset \( P \subseteq [n] \) where \( f(P) = 0 \) and for any two secrets \( s_1, s_2 \in \mathcal{S} \), the following holds:  

\[
\Delta(\mathsf{Share}(s_{1}; R_1)_P, \mathsf{Share}(s_{2}; R_2)_P) \leq \varepsilon,
\]
where \( R_1 \gets \mathcal{R} \) and \( R_2 \gets \mathcal{R} \) are independent random variables.  
\end{definition}

\medskip

\begin{definition}(Post-quantum computational privacy for classical secrets \cite{ccakan2023computational})
A secret-sharing scheme \( \mathsf{SS} \) that realizes a monotone function \( f \) is considered post-quantum computationally private if, for any subset \( P \subseteq [n] \) where \( f(P) = 0 \), any two secrets \( s_1, s_2 \in \mathcal{S} \), and any QPT (quantum polynomial-time) adversary \( \{C_\lambda\}_\lambda \), the following holds:  

\[
\left| \Pr_{R \gets \mathcal{R}} \left[C_\lambda(\mathsf{Share}(s_{1}; 1^\lambda, R)_P) = 1\right] - \Pr_{R \gets \mathcal{R}} \left[C_\lambda(\mathsf{Share}(s_{2}; 1^\lambda, R)_P) = 1\right] \right| \leq \mathsf{negl}(\lambda).
\]
\end{definition}

\medskip

In a \((t, n)\)-threshold secret-sharing scheme, Ogata et al. introduced the concept of \textit{security against cheaters} [Section 4.1, \cite{ogata2004new}]. Building upon this concept, we propose an extension to this definition, termed \textit{security against partial cheating}, which serves as a property to defend against potential attacks.

\medskip

For each participant \(P_i\) the share is given by
\(s^{(a)}_i = \bigl(s^{(a)}_{i_1}, s^{(a)}_{i_2}, \dots, s^{(a)}_{i_m}\bigr),\) for some integer \(m\)
and is written into two parts, \(s^{(a)}_i = \bigl(s^{(o)}_i, s^{(c)}_i\bigr),\)
where \(s^{(o)}_i\) is the \textit{original share part}, and \(s^{(c)}_i\) is the \textit{covert share part}. Let \(X_{\mathcal{S}^{(o)}}\) and \(X_{\mathcal{S}^{(c)}}\) be the random variable defined over the original and the covert part of the secret spaces in \(\mathcal{S}\) and \(V_{i}^{(a)}\) be the random variable induced by \(s_{i}^{(a)}\). Let
\[
supp(V^{(a)}_i) = \{\, s^{(a)}_i \mid \Pr(V_{i}^{(a)}=s^{(a)}_i) > 0 \,\}
\]
denote the support of the share vector \(s^{(a)}_i\) for participant \(P_i\). For a given \(t\)-tuple of shares
\[
w^{(a)} = \Bigl( (s^{(o)}_{i_1}, s^{(c)}_{i_1}),\,(s^{(o)}_{i_2}, s^{(c)}_{i_2}),\,\dots,\,(s^{(o)}_{i_t}, s^{(c)}_{i_t}) \Bigr)
\]
in the product space
\[
supp(V^{(a)}_{i_1}) \times supp(V^{(a)}_{i_2}) \times \cdots \times supp(V^{(a)}_{i_t}),
\]
we define the \emph{partial reconstruction function} \(\mathsf{Sec}^{(o)}\) by
\[
    \mathsf{Sec}^{(o)}(w^{(a)}) = 
    \begin{cases} 
    s^{(o)}, & \text{if } \exists s^{(o)} \text{ s.t. } \Pr(X_{S^{(o)}} = s^{(o)} \mid V_{i_1}^{(a)}, \dots, V_{i_t}^{(a)} = w^{(a)}) = 1, \\
    \perp, & \text{otherwise}.
    \end{cases}
\]

and 

\[
    \mathsf{Sec}^{(c)}(w^{(a)}) = 
    \begin{cases} 
    s^{(c)}, & \text{if } \exists s^{(c)} \text{ s.t. } \Pr(X_{S^{(c)}} = s^{(c)} \mid V_{i_1}^{(a)}, \dots, V_{i_t}^{(a)} = w^{(a)}) = 1, \\
    \perp, & \text{otherwise}.
    \end{cases}
\]

Thus, when the \(t\) players provide their correct shares
\[
b = \Bigl( (s^{(o)}_{i_1},s^{(c)}_{i_1}),\,(s^{(o)}_{i_2},s^{(c)}_{i_2}),\,\dots,\,(s^{(o)}_{i_t},s^{(c)}_{i_t}) \Bigr),
\]
we have
\[
\mathsf{Sec}^{(p)}(b) = (\mathsf{Sec}^{(o)}(b),\mathsf{Sec}^{(c)}(b))= \bigl(s^{(o)},s^{(c)}\bigr)=s.
\]
Let \(b\) denote the honest share tuple:
\[
b = \Bigl( (s^{(o)}_{i_1}, s^{(c)}_{i_1}),\, (s^{(o)}_{i_2}, s^{(c)}_{i_2}),\, \dots,\, (s^{(o)}_{i_t}, s^{(c)}_{i_t}) \Bigr).
\]
Now, consider a forged share tuple
\[
b' = \Bigl( (s^{(o)}_{i_1}, s^{(c)'}_{i_1}),\, (s^{(o)}_{i_2}, s^{(c)'}_{i_2}),\, \dots,\, (s^{(o)}_{i_t}, s^{(c)'}_{i_t}) \Bigr)
\]
with the property that for every \(j \in \{1,2,\dots,t\}\) the original part is unchanged:
\[
s^{(o)}_{i_j} \text{ in } b' = s^{(o)}_{i_j} \text{ in } b,
\]
while there exists at least one index \(j\) such that
\[
s^{(c)'}_{i_j} \neq s^{(c)}_{i_j}.
\]
We say that the dictator \(\mathcal{D}\) is \emph{partially cheated} by the forged tuple \(b'\) if
\[
\mathsf{Sec}^{(p)}(b') \in \mathcal{S} \quad \text{and} \quad \mathsf{Sec}^{(p)}(b') \neq \mathsf{Sec}^{(p)}(b).
\]
That is, although \(\mathcal{D}\) reconstructs the correct original component \(s^{(o)}\), the covert component \(s^{(c)}\) is altered due to the substitution of forged shares.

\medskip

\begin{definition}(\cite{ogata2004new})\label{def:parchit}
For a coalition of \(t\) players \(P_{i_1}, \ldots, P_{i_t}\) with covert shares \(b^{(c)} = (s_{i_1}^{(c)}, \ldots, s_{i_t}^{(c)})\), define the \textit{partial cheating probability} as:
\[
\mathsf{Cheat}^{(p)}\left(V_{i_1}^{(a)}, \ldots, V_{i_t}^{(a)}\right): = \max_{b} \max_{b'} \Pr\left( \mathcal{D} \; \text{is cheated by} \; b' \mid \; P_{i_{1}},\ldots,P_{i_{t}} \;\text{have}\; b\right).
\]
\if 0
where \(W = (V_{i_1}^{(a)}, \ldots, V_{i_t}^{(a)})\) is the joint distribution of the cheaters’ covert shares, \(\mathcal{F}(b^{(c)})\) is the set of all feasible forged covert shares \(b^{(c)'}\) given true shares \(b^{(c)}\), and \(supp(W)\) is the support of \(W\).
\fi
\end{definition}

\subsection{Quantum erasure-correcting codes}\label{subsec:quantum-erasure-codes}
We have used the following description of quantum erasure correcting code, described by \c{C}akan et al. from the paper \cite{ccakan2023computational}.

\begin{definition}(Quantum Erasure Correcting Code(\(\mathsf{QECC}\)) \cite{ccakan2023computational})
A pair of trace-preserving quantum operations, denoted as \( \mathsf{QC} = (\mathsf{QC}.\mathsf{Enc}, \mathsf{QC}.\mathsf{Dec}) \), is referred to as a quantum erasure correcting code (\(\mathsf{QECC}\)) over the input space \( \mathcal{H}_{\text{inp}} \) and the output space \( \mathcal{H}_{\text{out}} = \bigotimes_{i \in [n]} \mathcal{H}_i \) for a subset \( P \subseteq [n] \), if for any quantum operation \( \Upsilon \) acting on \( \mathcal{H}_{\text{out}} \) that preserves the identity on \( \mathcal{H}_i \) for all \( i \in P \), the following condition holds for any quantum state \( \rho \) in \( \mathcal{H}_{\text{inp}} \):

\[
(\mathsf{QC.Dec} \circ \Upsilon \circ \mathsf{QC.Enc})(\rho) = (\rho \otimes \sigma)
\]

for some fixed quantum state \( \sigma \).  

If \( (\mathsf{QC}.\mathsf{Enc}, \mathsf{QC}.\mathsf{Dec}_P) \) serves as a Quantum Error-Correcting Code (\(\mathsf{QECC}\)) for all subsets \( P \subseteq [n] \) where a monotone function \( f : \{0, 1\}^n \longrightarrow \{0, 1\} \) satisfies \( f(P) = 1 \), then the collection of operations \( (\mathsf{QC}.\mathsf{Enc}, (\mathsf{QC}.\mathsf{Dec}_P)_{P \subseteq [n]}) \) is said to realize \( f \) as a \(\mathsf{QECC}\). The code reconstruction function is defined as follows:
\[
\mathsf{QC.Rec}_P(\tau) = \mathsf{QC.Dec}(\tau \otimes (|0\rangle \langle 0|)^{\otimes\overline{P}}).
\]

A quantum error-correcting code that encodes \( k \) q-ary qudits into \( n \) q-ary qudits and can recover from up to \( (d - 1) \) erasures is referred to as a \( [[n, k, d]]_q \) code.
\end{definition}

\subsection{Quantum secret-sharing scheme}
Secret sharing involves dividing a secret into multiple parts, which are then distributed to different players. Only when these parts are combined, the original secret can be reconstructed. In a quantum anamorphic secret-sharing scheme, not only would the secret be protected by quantum-resistant cryptographic methods, but an additional covert secret is encoded into the shares themselves. This ensures that even if a quantum adversary were to intercept or analyze some of these shares, they would be unable to detect the existence of the hidden message without the appropriate classical or quantum key.

\medskip

\noindent For a rigorous description about quantum secret-sharing model we refer \cite{imai2003quantum,ccakan2023computational} to the reader.

\begin{definition}(No-cloning function \cite{ccakan2023computational})
A monotone function \( f \colon \{0, 1\}^n \longrightarrow \{0, 1\} \) is said to satisfy the no-cloning property and called \text{no-cloning function} if, for every subset \( P \subseteq [n] \), it holds that \( f(\overline{P}) = 0 \) while \( f(P) = 1 \). 

\end{definition}

Let the Hilbert space \(\mathcal{S}\) be the domain of secret, and the Hilbert spaces \( \mathcal{H}_1, \ldots, \mathcal{H}_n \) be the domain of shares of \(n\) players. Let \( f \colon \{0, 1\}^n \longrightarrow \{0, 1\} \) be a no-cloning monotone function. 

\medskip

\begin{definition}(Quantum secret-sharing scheme \cite{ccakan2023computational})
A \textit{quantum secret-sharing} (\(\mathsf{QSS}\)) scheme with perfect privacy that realizes a monotone function \( f \) is defined as a set of trace-preserving quantum operations:

\[
\mathsf{QSS} = (\mathsf{Share}, (\mathsf{Rec}_P)_{P \subseteq [n]})
\]
that satisfy the following conditions for all subsets \( P \subseteq [n] \):  

\noindent \(\bullet \) \textbf{Correctness:} If \( f(P) = 1 \), then the pair \( (\mathsf{Share}, \mathsf{Rec}_P) \) forms a quantum error-correcting code (\(\mathsf{QECC}\)) for \( P \), ensuring that the secret can be reconstructed.  

\noindent \(\bullet \) \textbf{Perfect Privacy}: If \( f(P) = 0 \), then for any two quantum states \( |\psi_1\rangle, |\psi_2\rangle \in \mathcal{S} \), the marginal distributions over the shares outside \( P \) remain identical, i.e.,  
\[
\Tr_{\overline{P}}(\mathsf{Share}(|\psi_1\rangle \langle \psi_1|)) = \Tr_{\overline{P}}(\mathsf{Share}(|\psi_2\rangle \langle \psi_2|)).
\]
This property ensures that unauthorized subsets gain no information about the quantum secret.
\end{definition}

\medskip

\noindent A quantum secret-sharing scheme (\(\mathsf{QSS}\)) is considered \textit{efficient} if both the sharing algorithm (\(\mathsf{QSS.Share}\)) and the reconstruction algorithm (\(\mathsf{QSS.Rec}\)) can be implemented using polynomial-size circuits \cite{ccakan2023computational}. Additionally, in efficient schemes, the size of each share is also polynomially bounded \cite{ccakan2023computational}.

\medskip

\begin{definition}( Statistical privacy for quantum secrets \cite{2023-613} \cite{ccakan2023computational})
A quantum secret-sharing scheme (\(\mathsf{QSS}\)) that realizes a monotone function \( f \) is said to be \( \varepsilon \)-statistically private if, for every subset \( P \subseteq [n] \) where \( f(P) = 0 \), and for any two secrets \( |\psi_1\rangle, |\psi_2\rangle \in \mathcal{S} \), the following condition holds:  
\[
D(\Tr_{\overline{P}} (\mathsf{Share}(|\psi_1\rangle \langle \psi_1|)), \Tr_{\overline{P}} (\mathsf{Share}(|\psi_2\rangle \langle \psi_2|))) \leq \varepsilon.
\]  
\end{definition}

\noindent Note that the perfect privacy means here \( 0 \)-statistical privacy.

\medskip

\begin{definition}(Computational privacy for quantum secrets \cite{ccakan2023computational})
A quantum secret-sharing scheme (\(\mathsf{QSS}\)) that realizes \( f \) is considered \textit{computationally private} if, for any subset \( P \subseteq [n] \) where \( f(P) = 0 \), any two secrets \( |\psi_1\rangle, |\psi_2\rangle \in \mathcal{S} \), and any QPT (quantum polynomial-time) adversary \( \{C_\lambda\}_\lambda \), the following holds:  
\[
\left| \Pr \left[ C_\lambda (\Tr_{\overline{P}} (\mathsf{Share}(|\psi_1\rangle \langle \psi_1|; 1^\lambda))) = 1 \right] - \Pr \left[ C_\lambda (\Tr_{\overline{P}} (\mathsf{Share}(|\psi_2\rangle \langle \psi_2|; 1^\lambda))) = 1 \right] \right| \leq \mathsf{negl}(\lambda).
\] 
\end{definition}

\if 0

\begin{definition}(Advantage pseudometric)
For a family \( \mathcal{F} \) of quantum circuits with single bit classical output and for any two density matrices \( \rho\), \(\sigma \) of appropriate dimension, the advantage of \( \mathcal{F} \) for distinguishing \( \rho \) versus \( \sigma \) is defined as
\[
\operatorname{Adv}_\mathcal{F}(\rho, \sigma) = \max_{C \in \mathcal{F}} | \Pr[C(\rho) = 1] - \Pr[C(\sigma) = 1] |.
\]
\end{definition}

\fi

\medskip

\begin{definition}(Information ratio of a quantum secret-sharing scheme \cite{beimel2011secret} \cite{ccakan2023computational}).
    If a secret is composed of \(t\) qubits then the \text{information ratio of a quantum secret-sharing scheme} is defined by 
\[
    \frac{\max_{i \in [n]}|\mathcal{S}_{i}|}{t}
\] 
\end{definition}

\subsubsection{Quantum Anamorphic Secret-Sharing Schemes}
\noindent In this section, we introduce the notion of an \textit{anamorphic secret-sharing}. In our proposed mathematical model, there is a \textit{dictator} \(\mathcal{D}\), who is a passive adversary here, a \textit{dealer} \(D\) who distributes shares to a set of \(n\) players according to predefined access structures, in the presence of the \textit{dictator} ensuring the correctness and privacy properties of the anamorphic secret-sharing scheme. The dealer aims to send two messages: an original message and a covert message, to the set of players. After the encryption of these two messages, combining these two ciphertexts we construct the anamorphic ciphertext which is computationally indistinguishable from the original ciphertext to the dictator. The dealer now sends the anamorphic ciphertext to the set of players along with the keys to decrypt those messages.

The dealer sends the anamorphic ciphertext and the anamorphic key to the set of \(n\) players. The authorized set of players reconstructs the original key and the covert key and then shares the original key with the dictator along with the anamorphic ciphertext. The dictator extracts the original message from the anamorphic ciphertext and then verifies the original message, that he wanted to send. The dictator cannot distinguish between original and anamorphic ciphertexts, because of the indistinguishability property and using the same decryption algorithm \(\mathsf{DOM}\) that can be applied to both original and anamorphic ciphertext to extract the same original message from either of the ciphertexts.

\medskip

Let $\mathcal{S}$ be a secret space, and let $s, \hat{s} \in \mathcal{S}$ be the original and covert secrets, respectively. Let \(\mathcal{S}_{1},\cdots,\mathcal{S}_{n}\) be the domain of shares of the players. Let $\mathsf{k}^{(o)}, \mathsf{k}^{(a)} \in \mathsf{K}$ be the normal or the original key and the anamorphic key, respectively. The anamorphic key \(\mathsf{k}^{(a)}=(\mathsf{k}^{(o)},\mathsf{k}^{(c)})\) consists of both the original and the covert keys. Let $\mathsf{ASS.Share}, \mathsf{ASS.Rec}^{\mathsf{Original}}_{P \subseteq [n]}, \mathsf{ASS.Rec}^{\mathsf{Covert}}_{P\subseteq [n]}$ denote the share function for the combined anamorphic message, the normal reconstruction function to reconstruct the original message, and the reconstruction function to recover the covert message, respectively.

\medskip

We define two encryption schemes:
\begin{itemize}
    \item \textbf{Original message encryption scheme:}
    \begin{align*}
    \mathsf{Enc}^{(o)}: &\mathcal{S} \times \mathsf{K}^{(o)} \to \mathcal{C}^{(o)} \\
    &(s,\mathsf{k}^{(o)}) \longmapsto c^{(o)}
    \end{align*}
    
    with corresponding decryption function $\mathsf{Dec}^{(o)}$.
    
    \item \textbf{Covert message encryption scheme:}
    \begin{align*}
    \mathsf{Enc}^{(c)}: &\mathcal{S} \times \mathsf{K}^{(c)} \longrightarrow \mathcal{C}^{(c)} \\
    &(\hat{s},\mathsf{k}^{(c)}) \longmapsto c^{(c)}
    \end{align*}
    with corresponding decryption function $\mathsf{Dec}^{(c)}$.
\end{itemize}

\noindent We introduce an efficiently computable embedding function to construct the anamorphic message:

\[
\Theta: \mathcal{C}^{(o)} \times \mathcal{C}^{(c)} \longrightarrow \mathcal{C}^{(a)}
\]
which produces an anamorphic ciphertext $c^{(a)} = \Theta(c^{(o)}, c^{(c)})$. Additionally, we assume the existence of an efficient extraction algorithm:
\[
\mathsf{EOC}: \mathcal{C}^{(a)} \longrightarrow \mathcal{C}^{(o)}
\]
such that for all $c^{(o)}$ and $c^{(c)}$, we have:
\[
\mathsf{EOC}(\Theta(c^{(o)}, c^{(c)})) = c^{(o)}.
\]

\noindent Let the anamorphic key be defined as:
\[
\mathsf{k}^{(a)} = (\mathsf{k}^{(o)}, \mathsf{k}^{(c)}) \in \mathsf{K}^{(o)} \times \mathsf{K}^{(c)},
\]
and let $\mathcal{R}$ be a randomness space with distribution $\mu$. Define $\mathcal{K} = \mathsf{K}^{(a)}  \times \mathcal{R}$, where \(\mathsf{K}^{(o)} \times \mathsf{K}^{(c)}=\mathsf{K}^{(a)}\).

\medskip

For \(b \in \{0,1\}\), we define \(c^{(b)}\), where \(c^{(0)}=c^{(o)}\) defines the \textit{original ciphertext} and \(c^{(1)}=c^{(a)}\) defines the \textit{anamorphic ciphertext}. 

\begin{definition}\label{def:ass}(Anamorphic Secret Sharing(\(\mathsf{ASS}))\) An anamorphic secret-sharing scheme(\(\mathsf{ASS}\)) with perfect privacy realizing the monotone function \(f\colon \{0,1\}^{n} \longrightarrow \{0,1\}\) is formally defined as a tuple $\Sigma_{\mathsf{ASS}}= (\mathsf{ASS.Share}, \mathsf{ASS.Rec_{P\subseteq [n]}^{\mathsf{AM}}})$, where \(\mathsf{ASS.Rec_{P\subseteq [n]}^{\mathsf{AM}}}=(\mathsf{ASS.Rec}^{\mathsf{Original}}_{P\subseteq [n]}, \mathsf{ASS.Rec}^{\mathsf{Covert}}_{P\subseteq [n]})\). Each of the deterministic functions is defined as follows:

\medskip

\if 0

\noindent \textbf{Original Share Distribution:} 

\begin{align*}
    \mathsf{Share}^{(o)}:& \mathsf{Encrypt}_{\mathsf{a}}(\mathcal{S},\mathsf{k}^{(o)}) \times \mathcal{K} \longrightarrow \mathcal{S}^{(o)}_{\{\mathcal{D}\}} \\
    & (\mathsf{c^{(a)}},\kappa^{(o)}) \longmapsto \left((\mathsf{c^{(a)}}, s_{1}^{(o)}),\cdots,(\mathsf{c^{(a)}}, s_{n}^{(o)}),(\mathsf{c^{(a)}}, s_{\mathcal{D}}^{(o)}) \right)
\end{align*}

\fi

\noindent \(\bullet\) \textbf{Anamorphic Share Distribution:} 

\begin{align*}
    \mathsf{ASS.Share}:& \mathcal{S} \times \mathcal{S} \times \mathcal{K} \longrightarrow \mathcal{C}^{(a)} \times \mathcal{S}_{[n]}. \\
    & (s,\hat{s},\kappa^{(a)}) \longmapsto \left(c^{(a)}, s_{1}^{(a)},\cdots,s_{n}^{(a)} \right)
\end{align*}
where \(\mathcal{K}:=\mathsf{K} \times \mathcal{R}\), with \(\kappa^{(a)}=(\mathsf{k}^{(a)},r)\) with \(r\) randomly chosen from a distribution \(\mu : \mathcal{R} \longrightarrow \mathbb{R}\), and \(\kappa^{(a)}=(\kappa^{(o)},\kappa^{(c)})\) consists of both original and covert parts. Using a secret-sharing scheme with access structure $f$, distribute the original key $\mathsf{k}^{(o)}$ into shares $\{ s_i^{(o)} \}_{i=1}^n$ and the covert key $\mathsf{k}^{(c)}$ into shares $\{ s_i^{(c)} \}_{i=1}^n$. Each player $i$ receives the key share \(s_{i}^{(a)} = (s_i^{(o)}, s_i^{(c)})\).

\medskip

\noindent \(\bullet\) \textbf{Reconstruction of the original share:} 
For any authorized subset \(P \subseteq [n]\) with \(f(P) = 1\), there exist deterministic reconstruction functions, \(\mathsf{Rec}^{(\mathsf{key}_{o})}_{P}\), which is the original key reconstruction function, and \(\mathsf{ASS.Rec}_{P}\) that reconstruct the original secret \(s\), is defined by the following commutative diagram:

\begin{align*}
    \mathsf{Rec}^{(\mathsf{key}_{o})}_{P}:&\mathcal{S}_{P} \longrightarrow \mathsf{K}^{(o)}
\end{align*}

\[
\begin{tikzcd}
\mathcal{C}^{(a)} \times \mathcal{S}_{P} \arrow[r, "\mathsf{Rec}^{(\mathsf{key}_{o})}_{P}"] \arrow[rd, "\mathsf{Dec^{(o)} \; \circ \; \mathsf{Rec}^{(\mathsf{key}_{o})}_{P}}"'] & \mathcal{C}^{(a)} \times \mathsf{K}^{(o)} \arrow[d, "\mathsf{Dec^{(o)}}"] \\
& \mathcal{S}
\end{tikzcd}
\]
where \(\mathsf{Dec^{(o)}}\) refers to the decryption of the original message algorithm, \(\mathsf{Rec}^{(\mathsf{key}_{o})}_{P}(c^{(a)}, s^{(o)}_{P}) = (c^{(a)}, \mathsf{k}^{(o)}) \), \(\mathsf{Dec^{(o)}}(c^{(a)}, \mathsf{k}^{(o)}) = s\) and we define \(\mathsf{Rec_{P \subseteq [n]}^{Original}}:=\mathsf{Dec^{(o)} \circ \mathsf{Rec}^{(\mathsf{key}_{o})}_{P \subseteq [n]}}\).

\medskip

\noindent \(\bullet\) \textbf{Reconstruction of the covert share:} 
For any authorized subset \(P \subseteq [n]\) with \(f(P) = 1\), there exist deterministic covert key reconstruction function \(\mathsf{Rec}^{(\mathsf{key}_{c})}_{P}\) and consequently, \(\mathsf{ASS.Rec}^{\mathsf{Covert}}_{P}\) that reconstructs the covert secret \(\hat{s}\), is defined by the following commutative diagram:

\begin{align*}
    \mathsf{Rec}^{(\mathsf{key}_{c})}_{P}: \mathcal{S}_{P} \longrightarrow \mathsf{K}^{(c)} 
\end{align*}

\[
\begin{tikzcd}
\mathcal{C}^{(a)} \times \mathcal{S}_{P} \arrow[r, "\mathsf{ECC} \circ \mathsf{Rec}^{(\mathsf{key}_{c})}_{P}"] \arrow[rd, " \mathsf{Dec}^{(c)} \; \circ \; \mathsf{ECC} \circ \mathsf{Rec}^{(\mathsf{key}_{c})}_{P}"'] & \mathcal{C}^{(c)} \times \mathsf{K}^{(c)} \arrow[d, "\mathsf{Dec^{(c)}}"] \\
& \mathcal{S} 
\end{tikzcd}
\]
where \(\mathsf{Dec^{(c)}}\) refers to the decryption of the covert message algorithm, \(\mathsf{ECC} \circ \mathsf{Rec}^{(\mathsf{key}_{c})}_{P}(c^{(a)}, s^{(c)}_{P}) = (c^{(a)}, \mathsf{k}^{(c)}) \). Let \[\mathsf{ECC}\colon \mathcal{C}^{(a)} \longrightarrow \mathcal{C}^{(o)}\] be the deterministic covert ciphertext extraction algorithm and \(\mathsf{ECC}(c^{(a)})=c^{(c)})\). Then, \(\mathsf{Dec^{(c)}}(c^{(c)},\mathsf{k}^{(c)}) = \hat{s}\), and we define \(\mathsf{Rec}^{\mathsf{Covert}}_{P}:= \mathsf{Dec}^{(c)} \circ \mathsf{ECC}\circ \mathsf{Rec}_{P}^{(\mathsf{key}_{c})}\).

\medskip

\noindent An anamorphic secret-sharing scheme \(\Sigma_{\mathsf{ASS}}\) must satisfy the following properties:

\noindent \(\bullet\) \textbf{Correctness:}  
\begin{itemize}
    \item \textbf{Correctness for original secret:} If \(f(P) = 1\), then for all \(s \in \mathcal{S}\),
    \[
   \Pr_{R \gets \mathcal{R}} \left[ \mathsf{Rec}_{P}^{\mathsf{Original}}(\mathsf{Share}(c^{(a)}, \kappa^{(o)})_{P} = s \right] = 1.
   \]

   \item \textbf{Correctness for anamorphic secret:} If \(f(P) = 1\), then for all \(s,\hat{s} \in \mathcal{S}\),
    \[
   \Pr_{R \gets \mathcal{R}} \left[ \mathsf{Rec}^{\mathsf{Covert}}_P(\mathsf{Share}(c^{(a)},\kappa^{(c)})_{P} = \hat{s} \right] = 1.
   \]

\end{itemize}

\noindent \(\bullet\) \textbf{Perfect Privacy:}

\begin{itemize}
   \item \textbf{Privacy for the anamorphic secret:}
  If \(f(P) = 0\), the shares reveal no information about the secrets, the anamorphic keys as well as anamorphic ciphertexts \((s,s')\). For all \(\mathsf{k}_{1}^{(a)}, \mathsf{k}_{2}^{(a)} \in \mathsf{K}\),
   \[
       \Pr_{R \gets \mathcal{R}} \left[(\mathsf{Share}(\kappa^{(a)}_{1})_{P} = v \right]
       = \Pr_{R \gets \mathcal{R}} \left[(\mathsf{Share}(\kappa^{(a)}_{2})_{P} = v \right]
   \]
\end{itemize}

\medskip

\noindent \(\bullet\) \textbf{Condition for covert reconstruction:}
If \(f(P)=1\), then the probability of reconstructing the covert secret \(\hat{s}\) using only original key shares \(\kappa^{(o)}\), given that \(s\) has been successfully reconstructed using the original key shares \(\kappa^{(o)}\), is zero:
   \[
   \Pr_{R \gets \mathcal{R}}\left[\mathsf{ASS.Rec}^{\mathsf{Covert}}_P(\mathsf{ASS.Share}(c^{(a)}, \kappa^{(o)})_P) \neq \bot\right] = 0.
   \]

\medskip

\noindent $\bullet$ \textbf{Indistinguishability of original and anamorphic ciphertexts:}  
We now describe a security game in which an adversary (or dictator) $\mathcal{D}$ is given access to shares produced from either an original encryption or anamorphic encryption. Using the shared extraction algorithm ($\mathsf{EOC}$), the players \emph{extract the original key shares(possibly based on ordering) and the original ciphertext $c^{(o)}$ from the anamorphic ciphertext} $c^{(a)}$, defined as

\[
\mathsf{EOC}: \mathcal{C}^{(a)} \longrightarrow \mathcal{C}^{(o)}
\]

which, on any anamorphic ciphertext $c^{(a)}$, outputs the original ciphertext $c^{(o)}$ embedded within $c^{(a)}$ and 
\[
\mathsf{EOK}: \{s_{i}^{(a)}\}_{i=1}^{n} \longrightarrow \{s_{i}^{(o)}\}_{i=1}^{n} 
\]
is a deterministic extraction algorithm that extracts original key shares from anamorphic key shares possibly based on ordering.
\medskip

\noindent Note that the challenger does not provide adversary access to \emph{extraction oracles} \(\mathsf{EOC}\) and \(\mathsf{EOK}\).

\smallskip

\noindent Here we define the \textbf{Real game} and the \textbf{Anamorphic game} as follows:

\medskip

\noindent \(\bullet\) \textbf{Challenge Oracle $\mathcal{O}$:}
\begin{enumerate}
    \item The challenger selects a random bit $b \in \{0,1\}$.
    
    \item If $b=0$ (\(\textbf{Real} \; \textbf{Game}(\mathsf{RealG_{a}}(\lambda,\mathcal{D}))\):
    \begin{enumerate}
        \item \(\mathsf{Gen_{o}(1^{\lambda})}\): Generate the original key $\mathsf{k}^{(o)}$.
        \item Compute $c^{(o)} \leftarrow \mathsf{Enc}^{(o)}(s,\mathsf{k}^{(o)})$.
        \item Set the challenge ciphertext $c^* := c^{(o)}$ and generate shares
        \(
        \bigl\{\, (c^{(o)}, s_i^{(o)}) \,\bigr\}_{i=1}^n.
        \)
    \end{enumerate}
    
    \item If $b=1$ (\(\textbf{Anamorphic} \; \textbf{Game}(\mathsf{AnamorphicG}_{a}(\lambda,\mathcal{D}))\)):
    \begin{enumerate}
        \item \(\mathsf{Gen_{a}(1^{\lambda})}\): Generate the anamorphic key $\mathsf{k}^{(a)} = (\mathsf{k}^{(o)},\mathsf{k}^{(c)})$.
        \item Compute
        \[
        c^{(a)} \longleftarrow \Theta\Bigl(c^{(o)},c^{(c)}\Bigr).
        \]
        \item Set $c^* := c^{(a)}$ and generate shares
        \(
        \bigl\{\, (c^{(a)}, s_i^{(o)}) \,\bigr\}_{i=1}^n .
        \)
    \end{enumerate}
\end{enumerate}

The adversary $\mathcal{D}$ is given access to the challenge shares and outputs a guess $b' \in \{0,1\}$.

\medskip

\begin{remark}
    Here we want to emphasize that here the original ciphertext is extracted from the anamorphic ciphertext and original key shares are extracted from the anamorphic key shares, then in both shares the original key shares are the same and the challenge is to distinguish the original and anamorphic ciphertexts for the adversary. Our main goal is to hide the covert ciphertext so that the adversary cannot suspect that there is a covert message within.
\end{remark}

\medskip

We define the advantage of $\mathcal{D}$ in distinguishing the \(\mathsf{Real \; Game}\) and the \(\mathsf{Anamorphic \; Game}\) as

\begin{align}
\mathsf{Adv}^{\mathsf{AME}}_{\mathcal{D}}(\lambda)
\;&=\; 
\Bigl|\, \Pr\bigl[\mathsf{RealG}_{a}(\lambda,\mathcal{D})\bigr] - \Pr\bigl[\mathsf{AnamorphicG_{a}}(\lambda,\mathcal{D})\bigr] \Bigr| \\
&=\Bigl|\, \Pr\bigl[(\mathcal{D}\text{ outputs }1) \wedge (b=1)\bigr] - \Pr\bigl[(\mathcal{D}\text{ outputs }1) \wedge (b=0)\bigr] \Bigr|.
\end{align}

The scheme satisfies the \emph{indistinguishability of original and anamorphic ciphertexts} property if for all PPT adversaries $\mathcal{D}$

\[
\mathsf{Adv}^{\mathsf{AME}}_{\mathcal{D}}(\lambda) < \mathsf{negl}(\lambda).
\]

In other words, an adversary (or dictator) is unable to distinguish (with a non-negligible advantage) whether the provided shares originated from an original encryption or from anamorphic encryption.

\medskip

\noindent \(2.\) To reconstruct the original message, the reconstruction procedure applied either to the original ciphertext or to the anamorphic ciphertext will produce the same original message
\[\mathsf{Rec}^{\mathsf{Original}}_{P \subseteq [n]}(\mathsf{Share}(c^{(a)},\kappa^{(o)})=\mathsf{Rec}^{\mathsf{Original}}_{P \subseteq [n]}(\mathsf{Share}(c^{(o)},\kappa^{(o)})=s,\]
so that the dictator cannot distinguish between \(c^{(o)}\) and \(c^{(a)}\), which one is the original ciphertext.

\medskip

\begin{remark}
Note that the definition of the anamorphic secret-sharing we described here is mainly based on classical key secret-sharing \(\mathsf{SS}\) Definition \ref{def:ss}. We will interchangeably use classical secret-sharing \(\mathsf{SS}\) while defining the secret-sharing algorithm and whenever we want to emphasize the classical key distribution separately.
\end{remark}

\medskip

\if 0

\noindent \(\bullet\) \textbf{Plausible Deniability:} Let, \( \mathcal{K} \) be a key generation algorithm that outputs a key \( k \), \( s \) be a normal secret embedded in the key and, \( (s,\hat{s}) \) be an anamorphic secret. Let \( b \in \{0,1\} \) be a bit that indicates whether an anamorphic secret is present where \( b = 0 \) indicates that the key contains only the normal secret \( S_n \) (no anamorphic secret) and, \( b = 1 \) indicates that the key contains both the normal secret \(s\) and an anamorphic secret \((s, \hat{s})\). Let \( Adv \) be a polynomial-time adversary attempting to determine whether an anamorphic secret is present in the key. The adversary \( A \) outputs a guess \( b' \in \{0,1\} \), where it tries to guess whether an anamorphic secret is present. The plausible deniability of the presence of the anamorphic secret is formalized as:

$$
\left| \Pr \left[ \mathsf{Adv} \left( \mathsf{Game}_0 \right) = 1 \right] - \Pr \left[ \mathsf{Adv} \left( \mathsf{Game}_1 \right) = 1 \right] \right| \leq \mathsf{negl}(\lambda)
,$$ where \( \mathsf{Game}_0 \) is the game where only the normal secret \( s \) is present (no anamorphic secret), \( \mathsf{Game}_1 \) is the game where both the normal secret \( s \) and the anamorphic secret \( (s,\hat{s}) \) are present, and, \( \mathsf{negl}(\lambda) \) is a negligible function in terms of the security parameter \( \lambda \).
\fi
\end{definition}

\section{The Compiler} \label{compiler}
The following compiler model has been discussed by \c{C}akan et al. in the paper \cite{ccakan2023computational}. This compiler can perfectly encrypt a quantum state using classical keys. Using similar ideas with classical secret-sharing scheme of Krawczyk \cite{krawczyk1993secret} and using some techniques from \cite{nascimento2001improving,fortescue2012reducing}, \c{C}akan et al. constructed a general compiler combining classical secret-sharing scheme and quantum erasure-correcting code [Theorem 9., Page 20, \cite{ccakan2023computational}]. In our paper, we have generalized the construction for anamorphic quantum ciphertext using multiple keys and we have also proved the correctness and perfect privacy properties accordingly with some additional techniques.

\medskip
Let \(f \colon \{0,1\}^{n} \longrightarrow \{0,1\}\) be a monotone function and \(\mathsf{SS}\) be the classical secret-sharing realizing \(f\), and \(\mathsf{QECC}\) \(\mathsf{QC}=(\mathsf{QC.Enc},\mathsf{QC.Rec})\) with \(n\) components, and for all \(x \in \{0,1\}^{n}\) realizing some monotone function \(f'(x) \geq f(x)\) which we have access to. The quantum erasure correcting code operations \(\mathsf{QC}\) corrects erasures in the complement of all the sets \(P \subseteq [n]\) such that \(f'(P)=1\) \cite{ccakan2023computational}.
Together with a classical secret-sharing scheme \( \mathsf{SS}\) the compiler implements a no-cloning monotone function \(f\) with a quantum error correcting code \(\mathsf{QC}\) that achieves a suitable no-cloning monotone function \( f'\ge f\), to establish a quantum secret sharing scheme \(\mathsf{QSS}\) that realizes \(f\) \cite{ccakan2023computational}. To construct quantum anamorphic secret-sharing \(\Sigma_{\mathsf{QASS}}\), we want to establish \( \mathsf{QASS.Share} \) algorithm to share the quantum states \( M_{o}\) and \(M_{c}\) and using The reconstruction procedure \( \mathsf{QASS.Rec^{\mathsf{AM}}_{P \subseteq [n]}}\) the set of players \( P \subseteq [n] \) reconstructs both the state \(M_{o}\) and \(M_{c}\) utilizing the decoding process for \( \mathsf{QC}\) and the algorithm \(\mathsf{DOM}\) and \(\mathsf{DCM}\).

\medskip

We now describe the quantum anamorphic secret-sharing algorithm: 

\begin{algorithm}[H]
\caption{Quantum Anamorphic Secret Sharing (QASS)}
\label{alg:qass}
\begin{algorithmic}[1]

\STATE \textbf{Input:} Anamorphic quantum state \(M_{f}^{(1)}\) and players \(P_1, \dots, P_n\).

\STATE \textbf{Output:} Shares \((s^{(a)}_i, \mathcal{E}_i)\) for each player \(P_i\).

\STATE \textbf{Step 1: Dealer Encrypts the Quantum State:}
\STATE \hspace{1cm} Dealer runs the Quantum Anamorphic Encryption algorithm \(\mathsf{QAE}\) to create the anamorphic quantum state:
\(
M_{f}^{(1)}.
\)

\STATE \textbf{Step 2: Share Generation:}

\STATE \hspace{1cm} Share the anamorphic key:
Share the anamorphic key \(\mathsf{k}^{(a)} \in \mathsf{K}\) using \(\mathsf{SS.Share}\) which yields the classical shares \((s^{(a)}_{1},\ldots,s^{(a)}_{n}),\) that is,
\[
(s^{(a)}_1, \dots, s^{(a)}_{n}) = \mathsf{SS}.\mathsf{Share}(k_{1}, k_{2}, k_{3}, k_{4}, k_{5}, k_{6}).
\]

\STATE \textbf{Step 3: Encode the anamorphic quantum state:}

\STATE \hspace{1cm} 
Encode the anamorphic quantum state \(M_{f}^{(1)}\) using \(\mathsf{QC.Enc}\), which yields entangled quantum systems \((\mathcal{E}_{1},\ldots,\mathcal{E}_{n})\).

\STATE \textbf{Step 4: Distribution of Shares:}
\STATE \hspace{1cm} Set \((s^{(a)}_{i},\mathcal{E}_{i})\) to be the final share of the player \(P_{i}\).
\end{algorithmic}
\end{algorithm}

\medskip

\medskip

For convenience of writing we have written the previous anamorphic key \(\mathsf{k}^{(a)}=(k,k',d_{1},d_{2},l,\eta)\) as \((k_{1},k_{4},k_{2},k_{5},k_{3},k_{6})\) in the respective order, where \(\mathsf{k^{(o)}}=(k_{1},k_{2},k_{3})\) is the original key and \(\mathsf{k^{(c)}}=(k_{4},k_{5},k_{6})\) is the covert key.

\medskip

Let \(\mathcal{J} \subset \mathbb{Z}^{+}\) be a finite set containing \(\eta\). Let \( f, f' \colon \{0, 1\}^n \longrightarrow \{0, 1\} \) be no-cloning monotone functions with the property that \( f' \ge f \). We define \if 0 
\( \xi_{M_{f}^{(0)}} = \mathsf{QC}.\mathsf{Enc}(M_{f}^{(0)})\),
\fi
\( \xi_{M_{f}^{(1)}} = \mathsf{QC}.\mathsf{Enc}(M_{f}^{(1)})\)  and for each \(i=1,\ldots,6\), define \( \tau_{(k_{i}, r_{i})} = |\mathsf{SS}.\mathsf{Share}(k_{i}, r_{i})\rangle \langle \mathsf{SS}.\mathsf{Share}(k_{i}, r_{i})| \), and \( \tau_{k_{i}', r} = |\mathsf{SS}.\mathsf{Share}(k_{i}', r')\rangle \langle \mathsf{SS}.\mathsf{Share}(k_{i}', r')| \) as the sharing of the keys \( k_{i},k'_{i}\) with the random inputs \( r_{i} \)s  respectively. 

\medskip

\noindent The scheme \( \mathsf{QASS} \) can be defined as follows:

\if 0

\begin{align*}
    \mathsf{QSS}.\mathsf{Share}(M_{f}^{(0)}) =\sum_{k_{1} \in \{0,1\}^{2 d_{1}}} \sum_{k_{2} \in \{1,\dots,2^{d_{1}+1}\}}\sum_{r_{1},r_{2},r_{3} \in \mathcal{R}} \sum_{k_{3} \in \operatorname{Sym}({2^{d_{1}+1}})} \frac{1}{2^{3 d_{1}+1}} \frac{1}{|\mathcal{R}|^{3}} \frac{1}{2^{d_{1}+1}!} \left[\bigotimes_{i=1}^{3} \tau_{(k_{i}, r_{i})} \otimes \xi_{M_{f}^{(0)}} \right].
\end{align*}

\begin{align*}
    \mathsf{QSS}.\mathsf{Share}(M_{f}^{(1)}) =\sum_{k_{1} \in \{0,1\}^{2 d_{1}}} \sum_{k_{2} \in \{1,\dots,2^{d_{1}+1}\}}\sum_{r_{1},r_{2},r_{3} \in \mathcal{R}} \sum_{k_{3} \in \operatorname{Sym}({2^{d_{1}+1}})} \frac{1}{2^{3 d_{1}+1}} \frac{1}{|\mathcal{R}|^{3}} \frac{1}{2^{d_{1}+1}!} \left[\bigotimes_{i=1}^{3} \tau_{(k_{i}, r_{i})} \otimes \xi_{M_{f}^{(1)}} \right].
\end{align*}

\fi

\begin{align}
\mathsf{QASS}.\mathsf{Share}(M_{f}^{(1)})=& \sum_{k_{1} \in \{0,1\}^{2 d_{1}}} \sum_{k_{4} \in \{0,1\}^{2 d_{2}}} \sum_{k_{2} \in [2^{d_{1}+1}]} \sum_{k_{5} \in [2^{d_{1}+1}]} \sum_{r_{1},r_{2},r_{3},r_{4},r_{5},r_{6} \in \mathcal{R}} 
\sum_{k_{3} \in [|\operatorname{Sym}{(2^{d_{1}+1}})|]} \sum_{k_{6}\in \mathcal{J}} \notag\\
&\frac{1}{2^{4d_{1}+2d_{2}+2}} \frac{1}{|\mathcal{R}|^{6}}\frac{1}{\mid \mathcal{J}\mid} \frac{1}{2^{d_{1}+1}!} 
\left[\bigotimes_{i=1}^{6} \tau_{(k_{i}, r_{i})} \otimes \xi_{M_{f}^{(1)}} \right].
\end{align}

and

\[
\mathsf{QASS}.\mathsf{Rec}_{P\subseteq [n]}^{\mathsf{Original}}(\sigma) = \mathsf{DOM}(\mathsf{QC}.\mathsf{Rec}_{P\subseteq [n]}(\Tr_{\mathsf{key}_{a}}(\sigma))), \mathsf{SS}.\mathsf{Rec}_{P\subseteq [n]}(\Tr_{\mathsf{state}}(\sigma)), 
\]
where \(\mathsf{key}_{a}\) corresponds to anamorphic key and 
\[
\mathsf{QASS.Rec}^{\mathsf{Covert}}_{P \subseteq [n]}(\sigma) = \mathsf{DCM}(\mathsf{QC}.\mathsf{Rec}_{P \subseteq [n]}(\Tr_{\mathsf{key}_{a}}(\sigma))), \mathsf{SS}.\mathsf{Rec}_{P \subseteq [n]}(\Tr_{\mathsf{state}}(\sigma)), 
\]

where \( \Tr_{\mathsf{key}_{a}} \) and \( \Tr_{\mathsf{state}} \) represent the process of tracing out the subsystem associated with the shares of the key and the shares of the quantum secret, respectively. In this framework, the classical key shares are represented as qubits in basis states. However, they may also be preserved as classical shares without altering the scheme.

\medskip

We describe the quantum anamorphic secret reconstruction algorithm as follows:

\begin{algorithm}[H]
\caption{Quantum Anamorphic Secret Reconstruction (\(\mathsf{QASS.Rec}\))}
\label{alg:qass-reconstruct}
\begin{algorithmic}[1]

\STATE \textbf{Input:} Shares \((s^{(a)}_i, \mathcal{E}_i)\) of each player \(P_{i}\) from a set of authorized players \(P \subseteq [n]\).

\STATE \textbf{Output:} Reconstructed original message \(M_o\) and covert message \(M_{c}\).

\STATE \textbf{Step 1: Reconstruct the Classical Shares:}
\STATE \hspace{1cm} Compute the classical shared components using classical secret-sharing(\(\mathsf{SS}\)) reconstruction:
\[
\mathsf{SS.Rec}_{P\subseteq [n]} \left( (s^{(a)}_i)_{i \in P} \right)=(k_{1}, k_{2}, k_{3},k_{4},k_{5},k_{6}).
\]

\STATE \textbf{Step 2: Reconstruct the Anamorphic State:}
\STATE \hspace{1cm} Compute the anamorphic quantum state using quantum reconstruction:
\[
\mathsf{QC.Rec}_P \left( (\mathcal{E}_i)_{i \in P} \right)=M_f^{(1)}.
\]

\STATE Apply the reconstruction algorithm \(\mathsf{QASS.Rec_{P\subseteq [n]}^{Original}}\) to reconstruct the original message.

\STATE Apply the reconstruction algorithm \(\mathsf{QASS.Rec_{P\subseteq [n]}^{Covert}}\) to reconstruct the covert message.

\STATE \textbf{Step 7: Output:}
\STATE \hspace{1cm} Output the reconstructed original message \(M_o\) and the covert message \(M_{c}\).

\end{algorithmic}
\end{algorithm}

\bigskip

\begin{theorem}\label{thm:ass}
    Let \( f, f' : \{0, 1\}^n \longrightarrow \{0, 1\} \) be no-cloning monotone functions satisfying \( f' \ge f \). Let \( \mathsf{QC} = (\mathsf{QC}.\mathsf{Enc}, (\mathsf{QASS}.\mathsf{Rec}_{P \subseteq [n]}^{\mathsf{AM}})) \) denote a quantum error-correcting code (\(\mathsf{QECC}\)) that implements \( f' \), and let \( \mathsf{SS} = (\mathsf{SS}.\mathsf{Share}, (\mathsf{SS}.\mathsf{Rec}_P)_{P \subseteq [n]}) \) represent a classical secret sharing scheme classified as [post-quantum computational, statistical, perfect] that realizes \( f \). Thus, \( \mathsf{QASS} \) represents a [computational, statistical, perfect] quantum anamorphic secret sharing scheme for \( f \), with the 
\begin{align*}
&\text{total share size for the anamorphic secret is} = \mathsf{size}(\mathsf{QC.Enc}(M_{f}^{(1)})) + (4d_{1} +2{d_{2}}+ 1) + 6 \lceil \log |\mathcal{R}| \rceil+ \\
&\lceil \log |\mathcal{J}| \rceil+ \lceil \log(2^{d_1 + 1}!)\rceil ).
\end{align*}

\if 0
\[ \text{total share size for the original secret with 3 keys} = \mathsf{size}(\mathsf{QC.Enc}(M_{f}^{(0)})) + (3d_1 +1+  3 \lceil \log |\mathcal{R}| \rceil + \lceil \log(2^{d_1 + 1}!)\rceil)), 
\]
\[ \text{total share size for the anamorphic secret with 3 keys} =\mathsf{size}(\mathsf{QC.Enc}(M_{f}^{(1)})) +(3d_{1} +1+ 3 \lceil \log |\mathcal{R}| \rceil + \lceil \log(2^{d_1 + 1}!)\rceil)\]

and

\fi

\noindent The difference between the anamorphic share size and the original share size, along with anamorphic key shares, is

\[
\mathsf{size}\Big(\mathsf{QASS.Share}(M_{f}^{(1)}\Big)-\mathsf{size}\Big(\mathsf{QASS.Share}(M_{f}^{(0)})\Big) = 0.
\]

\if 0
and
\[
\mathsf{size}\Big(\mathsf{QASS.Share}(M_{f}^{(1)}\Big)-\mathsf{size}\Big(\mathsf{QSS.Share}(M_{f}^{(1)})\Big) = d_{1}+2d_{2}+\lceil \log |\mathcal{J}| \rceil+3 \lceil \log |\mathcal{R}| \rceil).
\]
\fi

\noindent Moreover, \( \mathsf{QASS} \) exhibits efficient sharing and reconstruction protocols whenever \( \mathsf{QC} \) and \( \mathsf{SS} \) do.
\end{theorem}

\noindent \textbf{Correctness and Privacy:}

\begin{theorem}\label{thm:corr}
Let \(P\subseteq [n]\) be an authorized set of players (i.e., \(f(P)=1\)). Then the reconstruction procedure
\(
\mathsf{QASS.Rec}_P(\sigma)
\)
applied to the sharing state \(\mathsf{QASS.Share}(M_f^{(1)})\) correctly reconstructs the original message \(M_o\) and the covert message \(M_c\). In other words, if
\[
\mathsf{QASS.Rec}_P(\sigma) \coloneqq \Bigl(\mathsf{DOM}\bigl(\mathsf{QC.Rec}_P(\Tr_{\mathsf{key}_a}(\sigma))\bigr),\, \mathsf{DCM}\bigl(\mathsf{QC.Rec}_P(\Tr_{\mathsf{key}_a}(\sigma))\bigr)\Bigr),
\]
then
\[
\mathsf{QASS.Rec}_P(\sigma) = (M_o, M_c),
\]
where \(M_o\) is the original message and \(M_c\) is the covert message.
\end{theorem}

\begin{proof}
We assume that both the classical secret sharing scheme \(\mathsf{SS}\) and the quantum encoding scheme \(\mathsf{QC}\) satisfy their correctness properties.
The scheme distributes classical shares \( s_i^{(a)} \) for each \( i\in [n] \) corresponding to the keys \( (k_1,k_2,k_3,k_4,k_5,k_6) \). By the correctness of the classical secret sharing scheme \(\mathsf{SS}\), if \(P\subseteq[n]\) is an authorized set (i.e., \(f(P)=1\)), then
\begin{equation}
\mathsf{SS.Rec}_P\Bigl(\{s_i^{(a)}\}_{i\in P}\Bigr) = (k_1,k_2,k_3,k_4,k_5,k_6).
\end{equation}

\noindent The overall sharing state is given by
\[
\sigma = \sum_{\vec{k},\, \vec{r}} \alpha(\vec{k},\vec{r}) \left(\bigotimes_{i=1}^{6} \tau_{(k_i,r_i)}\right) \otimes \xi_{M_f^{(1)}},
\]
where \(\vec{k} = (k_1,k_2,k_3,k_4,k_5,k_6)\), \(\vec{r} = (r_1,\ldots,r_6)\) and the normalization constant is
\[
\alpha(\vec{k},\vec{r}) = \frac{1}{2^{4d_1+2d_2+2}}\, \frac{1}{|\mathcal{R}|^6}\, \frac{1}{|\mathcal{J}|}\, \frac{1}{(2^{d_1+1})!}.
\]

\medskip
The subsystem corresponding to the classical key shares is denoted by \(\mathsf{key}_a\). By linearity of the trace and by the independence of the classical and quantum parts, we have
\begin{equation}
\Tr_{\mathsf{key}_a}(\sigma) = \xi_{M_f^{(1)}}.
\end{equation}
The normalization factors in the definition of \(\sigma\) guarantee that
\(
\sum_{\vec{k},\,\vec{r}} \alpha(\vec{k},\vec{r}) = 1,
\)
so that the partial trace exactly recovers the quantum component \(\xi_{M_f^{(1)}}\).

\smallskip

Since \(\xi_{M_f^{(1)}} = \mathsf{QC.Enc}(M_f^{(1)})\), the correctness of the quantum encoding scheme \(\mathsf{QC}\) implies that the reconstruction procedure applied to the quantum subsystem yields the secret, i.e.,
\begin{equation}
\mathsf{QC.Rec}_{P\subseteq [n]}\Bigl(\Tr_{\mathsf{key}_a}(\sigma)\Bigr) = \mathsf{QC.Rec}_{P\subseteq [n]}\bigl(\xi_{M_f^{(1)}}\bigr) = M_f^{(1)}.
\end{equation}

\smallskip

The decoding maps \(\mathsf{DOM}\) and \(\mathsf{DCM}\) are applied to the reconstructed quantum state \(M_f^{(1)}\) to extract the original message \(M_o\) and the covert message \(M_c\), respectively. That is,
\begin{equation}
\begin{split}
\mathsf{DOM}\bigl(M_f^{(1)}\bigr) &= M_o,\\[1mm]
\mathsf{DCM}\bigl(M_f^{(1)}\bigr) &= M_c.
\end{split}
\end{equation}
Hence, the overall reconstruction procedure yields
\begin{equation}
\begin{split}
\mathsf{QASS.Rec}_P(\sigma) &= \Bigl(\mathsf{DOM}\bigl(\mathsf{QC.Rec}_P(\Tr_{\mathsf{key}_a}(\sigma))\bigr),\, \mathsf{DCM}\bigl(\mathsf{QC.Rec}_P(\Tr_{\mathsf{key}_a}(\sigma))\bigr)\Bigr)\\[1mm]
&=\Bigl(\mathsf{DOM}\bigl(M_f^{(1)}\bigr),\, \mathsf{DCM}\bigl(M_f^{(1)}\bigr)\Bigr)\\[1mm]
&= (M_o, M_c).
\end{split}
\end{equation}
proving the correctness of the quantum anamorphic secret-sharing scheme.
\end{proof}

\begin{theorem}\label{thm:priv}
    The above quantum anamorphic secret-sharing scheme has perfect privacy.
\end{theorem}

\begin{proof} To prove the correctness of the secret-sharing scheme, we show the existence of a reconstruction function. Let \(P \subseteq [n]\) with \(f(P)=0\), that is \(P\) is an unauthorized subset of players. Then,

\begin{align}
    &\Tr_{\overline{P}}(\mathsf{QASS.Share}(M_{f}^{(1)})) \notag\\
    &= \Tr_{\overline{P}} \Bigg[ \sum_{k_{1} \in \{0,1\}^{2 d_{1}}} 
    \sum_{k_{4} \in \{0,1\}^{2 d_{2}}} 
    \sum_{k_{2} \in \{1,\dots,2^{d_{1}+1}\}} 
    \sum_{k_{5} \in \{1,\dots,2^{d_{1}+1}\}} 
    \sum_{r_{1},r_{2},r_{3},r_{4},r_{5},r_{6} \in \mathcal{R}} 
    \sum_{k_{3} \in [|\operatorname{Sym}{(2^{d_{1}+1})|]}} \sum_{k_{6}\in \mathcal{J}} \notag\\
    &\quad \frac{1}{2^{4 d_{1}+2d_{2}+2}} \frac{1}{|\mathcal{R}|^{6}} \frac{1}{\mid \mathcal{J}\mid}\frac{1}{(2^{d_{1}+1})!}
    \quad \Bigg[ \bigotimes_{i=1}^{6} \tau_{(k_{i}, r_{i})} \otimes \xi_{M_{f}^{(1)}} \Bigg] \Bigg] \notag\\
    &=\sum_{k_{1} \in \{0,1\}^{2 d_{1}}} 
    \sum_{k_{4} \in \{0,1\}^{2 d_{2}}} 
    \sum_{k_{2} \in \{1,\dots,2^{d_{1}+1}\}} 
    \sum_{k_{5} \in \{1,\dots,2^{d_{1}+1}\}} 
    \sum_{k_{3} \in [|\operatorname{Sym}{(2^d_{1}+1})|]} \sum_{k_{6}\in \mathcal{J}}\notag \\
    &\Tr_{\overline{P}} \left[\frac{1}{2^{4d_{1}+2d_{2}+2}} \frac{1}{(2^{d_{1}+1})!} \frac{1}{\mid \mathcal{J}\mid}
    \xi_{M_{f}^{(1)}}\right] \quad \bigotimes_{i=1}^{6} \Tr_{\overline{P}} \left[ \sum_{r_{i} \in \mathcal{R}} \frac{1}{|\mathcal{R}|}  \tau_{(k_{i},r_{i})}\right] \notag\\
    &=\sum_{k_{1} \in \{0,1\}^{2 d_{1}}} 
    \sum_{k_{4} \in \{0,1\}^{2 d_{2}}} 
    \sum_{k_{2} \in \{1,\dots,2^{d_{1}+1}\}} 
    \sum_{k_{5} \in \{1,\dots,2^{d_{1}+1}\}} 
    \sum_{k_{3} \in [|\operatorname{Sym}{(2^d_{1}+1})|]} \sum_{k_{6}\in \mathcal{J}} \\
    &\Tr_{\overline{P}} \left[\frac{1}{2^{4 d_{1}+2d_{2}+2}} \frac{1}{(2^{d_{1}+1})!} \frac{1}{\mid \mathcal{J}\mid}
   \xi_{M_{f}^{(1)}}\right] \otimes \Tr_{\overline{P}}(\sigma_{P})\notag \\ 
    &= \Tr_{\overline{P}} \Bigg[ \mathsf{QC.Enc} \Bigg(\sum_{k_{1} \in \{0,1\}^{2d_{1}}} \sum_{k_{4} \in \{0,1\}^{2 d_{2}}} 
    \sum_{k_{2} \in \{1,\dots,2^{d_{1}+1}\}} 
    \sum_{k_{5} \in\{1,\dots,2^{d_{1}+1}\}} 
    \sum_{k_{3} \in [|\operatorname{Sym}{(2^{d_{1}}+1})|]} 
    \sum_{k_{6} \in \mathcal{J}} \\
    &\frac{1}{2^{4 d_{1}+2d_{2}+2}} \frac{1}{(2^{d_{1}+1})!} \frac{1}{\mid \mathcal{J}\mid}
    M_{f}^{(1)} \Bigg)\Bigg] \otimes \Tr_{\overline{P}}(\sigma_{P}) \notag\\
    &= \Tr_{\overline{P}} \left( \mathsf{QC.Enc}(\sigma_{P}')\right) \otimes \Tr_{\overline{P}}(\sigma_{P}).
\end{align}

Since the density matrix \(\sigma_{P}\) is only dependent upon the set of players \(P\), the unauthorized set of players \(P\) cannot reconstruct the secret. 

\medskip

Next, We will show that if the \(\mathsf{SS}\) is \(\frac{\epsilon}{12}\)-statistically private then the \(\mathsf{QASS}\) is \(\epsilon\)-statistically private. For any two keys \(k,k' \in \{0,1\}^{2 d_{1}} \times \{0,1\}^{2 d_{2}} \times \{1,\dots,2^{d_{1}+1}\} \times \{1,\dots,2^{d_{1}+1}\} \times \operatorname{Sym}{(2^{d_{1}+1})} \times \mathcal{J}\) ; where \(k=(k_{1},k_{4},k_{2},k_{5},k_{3},k_{6})\) and \(k'=(k_{1}',k_{4}',k_{2}',k_{5}',k_{3}',k_{6}')\).

\medskip

\noindent For each \(i=1,\ldots,6\), the partial trace over subsystem \(\overline{P}\) reduces to
\begin{equation}
\Tr_{\overline{P}} \left( \sum_{r_i \in \mathcal{R}} \frac{1}{|\mathcal{R}|} \tau_{(k_i, r_i)} \right) = \sum_{v_i} p_{v_i}^{(i)} \ket{v_i}\bra{v_i},
\end{equation}

\noindent where for each \(v=(v_{1},\ldots,v_{6}) \in \mathcal{S}_{P}\), \(p_{v_i}^{(i)} = \Pr_{r_i \gets \mathcal{R}}[\mathsf{Share}(k_i; r_i)_{P} = v_i]\) and \(p_{v_i}^{(i)'} = \Pr_{r_i \gets \mathcal{R}}[\mathsf{Share}(k_{i}'; r_i)_{P} = v_i]\) are defined to be the marginal probability distributions induced by the secret sharing process for \(k_i\) and  \(k_{i}'\) on subsystem \(\overline{P}\) and \(\ket{v_i}\bra{v_i}\) represents the basis state corresponding to \(v_i\).

\medskip

\noindent Then, 
\begin{equation}
\Tr_{\overline{P}}\left[ \bigotimes_{i=1}^{6} \left(\sum_{r_{i} \in \mathcal{R}}\frac{1}{|\mathcal{R}|} \tau_{(k_{i},r_{i})}\right)\right] =\sum_{v\in \mathsf{S}_{P}} \left(\prod_{i=1}^{6} p^{(i)}_{v_{i}} \right) \left[\bigotimes_{i=1}^{6}\ket{v_{i}}\bra{v_{i}}\right] \quad = \sum_{v\in \mathsf{S}_{P}} \left(\prod_{i=1}^{6} p^{(i)}_{v_{i}}\right) \ket{v}\bra{v},
\end{equation}

\noindent where \(\bigotimes_{i=1}^{6} \ket{v_{i}} \bra{v_{i}}=\ket{v}\bra{v}\) and we get the following relation between the trace distance and the total variation distance

\begin{align}
&D\left(\Tr_{\overline{P}}\left[ \bigotimes_{i=1}^{6} \left(\sum_{r_{i} \in \mathcal{R}}\frac{1}{|\mathcal{R}|} \tau_{(k_{i},r_{i})}\right)\right] ,\Tr_{\overline{P}}\left[ \bigotimes_{i=1}^{6} \left(\sum_{r_{i} \in \mathcal{R}}\frac{1}{|\mathcal{R}|} \tau_{(k_{i}',r_{i})}\right)\right] \right) \\
&= D \left(\sum_{v\in \mathsf{S}_{P}} \prod_{i=1}^{6} p^{(i)}_{v_{i}} \ket{v}\bra{v},\sum_{v\in \mathsf{S}_{P}} \prod_{i=1}^{6} p^{(i)'}_{v_{i}} \ket{v}\bra{v} \right) \\
& \leq \Delta\!\Bigg(\{\prod_{i=1}^6 p_{v_i}^{(i)}\}_{v \in \mathsf{S}_{P}},\;\{\prod_{i=1}^6 p_{v_i}^{(i)'}\}_{v \in \mathsf{S}_{P}}\Bigg),
%\label{eq:marginal_prob_dist}
\end{align}

\noindent which we got using the Lemma \ref{prop:trace-dist}. (Property v.).

\medskip

Now we prove the privacy of the \(\mathsf{QASS}\). Let \(P \subseteq [n]\) be an unauthorized subset such with \(f(P)=0\). For any two secret states \(M_{f}^{(1)}\) and \(M_{f}^{(1)'}\). Our argument is a hybrid one. First, we will contend that the perfect privacy of the one-time pad, along with Lemma \ref{prop:otp}, ensures that the composite shares of \(\overline{P}\) for two secrets \(M_{f}^{(1)}\) and \(M_{f}^{(1)'}\) will be computationally indistinguishable. 

\medskip

\noindent We define sharing of a random anamorphic key by 
\begin{align}
\kappa^{(a)} =&  \sum_{k_{1}' \in \{0,1\}^{2 d_{1}}} \sum_{k_{4}' \in \{0,1\}^{2 d_{2}}} \sum_{k_{2}' \in \{1,\dots,2^{d_{1}+1}\}} \sum_{k_{5}' \in \{1,\dots,2^{d_{1}+1}\}} \sum_{r_{1},r_{2},r_{3},r_{4},r_{5},r_{6} \in \mathcal{R}} \sum_{k_{3}' \in [|\operatorname{Sym}{(2^{d_{1}+1})|]}} \sum_{k_{6}'\in \mathcal{J}} \notag\\
&\frac{1}{2^{4 d_{1}+2 d_{2}+2}} \frac{1}{|\mathcal{R}|^{6}} \frac{1}{\mid \mathcal{J} \mid}\frac{1}{2^{d_{1}+1}!} \left[  \bigotimes_{i=1}^{6} \tau_{(k_{i}', r_{i})} \right].
\end{align}

and we define the hybrids 

\begin{align}
\Psi_{1} &= \Tr_{\overline{P}} \Bigg[ \Bigg( 
\sum_{k_{1} \in \{0,1\}^{2 d_{1}}} 
\sum_{k_{4} \in \{0,1\}^{2 d_{2}}} 
\sum_{k_{2} \in \{1,\dots,2^{d_{1}+1}\}}
\sum_{k_{5} \in \{1,\dots,2^{d_{1}+1}\}}
\sum_{k_{3} \in [|\operatorname{Sym}{(2^{d_{1}+1})|]}} \quad \frac{1}{2^{4 d_{1}+2d_{2}+2}} 
\frac{1}{(2^{d_{1}+1})!} \xi_{M_{f}^{(1)}} 
\Bigg) \\ \notag
&\otimes \kappa^{(a)} \Bigg]
\end{align}

and 

\begin{align}
\Psi_{2} &= \Tr_{\overline{P}} \Bigg[ \Bigg( 
\sum_{k_{1} \in \{0,1\}^{2 d_{1}}} 
\sum_{k_{4} \in \{0,1\}^{2 d_{2}}} 
\sum_{k_{2} \in \{1,\dots,2^{d_{1}+1}\}} 
\sum_{k_{5} \in \{1,\dots,2^{d_{1}+1}\}} 
\sum_{k_{3} \in [|\operatorname{Sym}{(2^{d_{1}+1})|]}} \quad \frac{1}{2^{4d_{1}+2d_{2}+2}} 
\frac{1}{(2^{d_{1}+1})!}
\xi_{M_{f}^{(1)'}} 
\Bigg) \\ \notag
&\otimes \kappa^{(a)} \Bigg]
\end{align}

\medskip

\noindent Our goal is to show \(D(\Psi_{1}, \Psi_{2})=0\).

\medskip

\noindent Subsequently, we will demonstrate that composite shares of the same secret are within a \(\frac{\epsilon}{2}\)-close trace distance when the shares of the key are replaced with shares of a random key, as opposed to when they are not replaced. We will show 

\begin{equation}
D\left(\Tr_{\overline{P}}\left(\mathsf{QSS.Share}(M_{f}^{(1)})\right), \Psi_{1}\right) \leq \frac{\epsilon}{2}
\label{eq:traceDistBound}
\end{equation}

and

\begin{align}
    D(\Tr_{\overline{P}}(\Psi_{2},\mathsf{QSS.Share}(M_{f}^{(1)'}))) \leq \frac{\epsilon}{2}
%\label{eq:trace_dist_bound_alt}
\end{align}

Then using the triangle inequality, we will show that 
\begin{equation}
D(\Tr_{\overline{P}}(\mathsf{QASS.Share}(M_{f}^{(1)})),\Tr_{\overline{P}}(\mathsf{QASS.Share}(M_{f}^{(1)'}))) \leq \epsilon
\end{equation}

Using the [Property (ii), Lemma \ref{prop:trace-dist}] and by distributing the partial trace, we get

\begin{align}
D(\Psi_{1}, \Psi_{2}) &= D \Bigg(\Tr_{\overline{P}} \Bigg[ 
\sum_{k_{1} \in \{0,1\}^{2d_{1}}} 
\sum_{k_{4} \in \{0,1\}^{2d_{2}}} 
\sum_{k_{2} \in \{1,\dots,2^{d_{1}+1}\}} 
\sum_{k_{5} \in \{1,\dots,2^{d_{1}+1}\}} 
\sum_{k_{3} \in [|\operatorname{Sym}(2^{d_{1}+1})|]} \notag \\
&\frac{1}{2^{4d_{1}+2d_{2}+2}} 
\frac{1}{(2^{d_{1}+1})!}
\xi_{M_{f}^{(1)}} 
\Bigg], \Tr_{\overline{P}} \Bigg[ 
\sum_{k_{1} \in \{0,1\}^{2d_{1}}} 
\sum_{k_{4} \in \{0,1\}^{2d_{2}}} 
\sum_{k_{2} \in \{1,\dots,2^{d_{1}+1}\}}
\sum_{k_{5} \in \{1,\dots,2^{d_{1}+1}\}} \notag\\
&\sum_{k_{3} \in [|\operatorname{Sym}(2^{d_{1}+1})|]}  
\frac{1}{2^{4d_{1}+2d_{2}+2}} 
\frac{1}{(2^{d_{1}+1})!}
\xi_{M_{f}^{(1)'}} 
\Bigg] \Bigg)
\end{align}

\noindent Therefore, when the key is chosen uniformly random, according to Lemma \ref{prop:otp}, the input is perfectly hidden by the quantum one-time pad and there exists a state \(\vartheta\) such that 

\begin{equation}
D(\Psi_{1},\Psi_{2})=D(\Tr_{\overline{P}}(\mathsf{QC.Enc(\vartheta)}), \mathsf{QC.Enc(\vartheta)}))=0
\end{equation}

\noindent By the [Properties v. and iv., Lemma \ref{prop:trace-dist}], we get, 

\begin{align}
   &D(\Psi_{1},\Tr_{\overline{P}}(\mathsf{QASS.Share}(M_{f}^{(1)}))) \notag\\
&= D\Bigg(\sum_{k_{1} \in \{0,1\}^{2 d_{1}}} 
\sum_{k_{4} \in \{0,1\}^{2 d_{2}}} 
\sum_{k_{2} \in \{1,\dots,2^{d_{1}+1}\}} 
\sum_{k_{5} \in \{1,\dots,2^{d_{1}+1}\}} 
\sum_{k_{3} \in [|\operatorname{Sym}{(2^{d_{1}+1})|]}} \notag \\ 
&\quad \frac{1}{2^{4 d_{1}+2d_{2}+2}} 
\frac{1}{(2^{d_{1}+1})!}
\Tr_{\overline{P}} (\xi_{M_{f}^{(1)}}) \otimes \Tr_{\overline{P}}(\kappa^{(a)}), \notag\\
&\sum_{k_{1} \in \{0,1\}^{2 d_{1}}} 
\sum_{k_{4} \in \{0,1\}^{2 d_{2}}} 
\sum_{k_{2} \in \{1,\dots,2^{d_{1}+1}\}}
\sum_{k_{5} \in \{1,\dots,2^{d_{1}+1}\}} 
\sum_{k_{3} \in [|\operatorname{Sym}{(2^{d_{1}+1})|]}} \notag \\ 
&\quad \frac{1}{2^{4d_{1}+2d_{2}+2}} 
\frac{1}{(2^{d_{1}+1})!} \Tr_{\overline{P}} (\xi_{M_{f}^{(1)}}) \otimes \Tr_{\overline{P}}\left[ \bigotimes_{i=1}^{6} \left(\sum_{r_{i} \in \mathcal{R}}\frac{1}{|\mathcal{R}|} \tau_{(k_{i},r_{i})}\right)\right] \Bigg) \\
&\leq \sum_{k_{1} \in \{0,1\}^{2 d_{1}}} 
\sum_{k_{4} \in \{0,1\}^{2 d_{2}}} 
\sum_{k_{2} \in \{1,\dots,2^{d_{1}+1}\}} 
\sum_{k_{5} \in \{1,\dots,2^{d_{1}+1}\}} 
\sum_{k_{3} \in [|\operatorname{Sym}{(2^{d_{1}+1})|]}} \quad \frac{1}{2^{4 d_{1}+2d_{2}+2}} 
\frac{1}{(2^{d_{1}+1})!} \notag\\
& D \Bigg( \Tr_{\overline{P}}(\kappa^{(a)}), \Tr_{\overline{P}} \left[ \bigotimes_{i=1}^{6} \left(\sum_{r_{i} \in \mathcal{R}}\frac{1}{|\mathcal{R}|} \tau_{(k_{i},r_{i})}\right)\right] \Bigg)
\end{align}

\medskip

\noindent By using the Lemma \ref{prop:trace-dist} (Property i.), we get

\begin{align}
    &D(\Psi_{1}, \Tr_{\overline{P}}(\mathsf{QASS.Share}(M_{f}^{(1)}))) \notag\\
    &\leq \sum_{k_{1},k_{1}' \in \{0,1\}^{2 d_{1}}} 
\sum_{k_{4},k_{4}' \in \{0,1\}^{2 d_{2}}} 
\sum_{k_{2},k_{2}' \in \{1,\dots,2^{d_{1}+1}\}} 
\sum_{k_{5},k_{5}' \in \{1,\dots,2^{d_{1}+1}\}} 
\sum_{k_{3},k_{3}' \in [|\operatorname{Sym}{(2^{d_{1}+1})|]}} \notag\\
&\quad \frac{1}{4^{4d_{1}+2d_{2}+2}} 
\frac{1}{((2^{d_{1}+1})!)^{2}} D\Bigg(\Tr_{\overline{P}}\left[ \bigotimes_{i=1}^{6} \left(\sum_{r_{i} \in \mathcal{R}}\frac{1}{|\mathcal{R}|} \tau_{(k_{i}',r_{i})}\right)\right],\Tr_{\overline{P}}\left[ \bigotimes_{i=1}^{6} \left(\sum_{r_{i} \in \mathcal{R}}\frac{1}{|\mathcal{R}|} \tau_{(k_{i},r_{i})}\right)\right] \Bigg) \notag\\
&\leq  \sum_{k_{1},k_{1}' \in \{0,1\}^{2 d_{1}}} 
\sum_{k_{4},k_{4}' \in \{0,1\}^{2 d_{2}}} 
\sum_{k_{2},k_{2}' \in \{1,\dots,2^{d_{1}+1}\}} 
\sum_{k_{5},k_{5}' \in \{1,\dots,2^{d_{1}+1}\}} 
\sum_{k_{3},k_{3}' \in [|\operatorname{Sym}{(2^{d_{1}+1})|]}} \notag\\
&\quad \frac{1}{4^{4 d_{1}+2d_{2}+2}} 
\frac{1}{((2^{d_{1}+1})!)^{2}}
\Delta\Bigg(\{\prod_{i=1}^{6} p^{(i)}_{v_{i}}\}_{v \in \mathsf{S}_{P}},\{\prod_{i=1}^{6} p^{(i)'}_{v_{i}}\}_{v \in \mathsf{S}_{P}}\Bigg) \label{eq: probdist}
\end{align}

\medskip
If \(P\) and \(P'\) are two probability distributions over a finite set \(\mathcal{S}_P\), then the statistical distance between two probability distributions \(P\) and \(P'\) is 
\begin{equation}
\Delta(P, P') = \frac{1}{2}\sum_{v \in \mathsf{S}_P} |p_v - p'_v|.
\end{equation}

\medskip

\noindent Now, considering two product distributions formed from marginals 

\[
P = \left\{ \prod_{i=1}^6 p_{v_i}^{(i)} \right\}_{v \in \mathsf{S}_P}, \quad
P' = \left\{ \prod_{i=1}^6 p_{v_i}^{(i)'} \right\}_{v \in \mathsf{S}_P}.
\]
we get

\begin{equation}
\Delta(P, P') = \frac{1}{2}\sum_{v \in \mathsf{S}_P} \left| \prod_{i=1}^6 p_{v_i}^{(i)} - \prod_{i=1}^6 p_{v_i}^{(i)'} \right|.
\end{equation}
\noindent Using the telescoping expansion of the difference
\begin{equation}
\prod_{i=1}^6 p_{v_i}^{(i)} - \prod_{i=1}^6 p_{v_i}^{(i)'} = \sum_{j=1}^6 \left( \prod_{i=1}^{j-1} p_{v_i}^{(i)'} \right) \left( p_{v_j}^{(j)} - p_{v_j}^{(j)'} \right) \left( \prod_{i=j+1}^{6} p_{v_i}^{(i)} \right),
\end{equation}
\noindent we get 
\begin{align}
\left| \prod_{i=1}^6 p_{v_i}^{(i)} - \prod_{i=1}^6 p_{v_i}^{(i)'} \right| 
&\leq \sum_{j=1}^6 \left| \prod_{i=1}^{j-1} p_{v_i}^{(i)'} \cdot (p_{v_j}^{(j)} - p_{v_j}^{(j)'}) \cdot \prod_{i=j+1}^{6} p_{v_i}^{(i)} \right| \quad \text{(by the triangular inequality)} \notag\\
&\leq \sum_{j=1}^6 |p_{v_j}^{(j)} - p_{v_j}^{(j)'}| \quad \text{(since each probability lies in \([0,1]\)).}
\end{align}

Summing over all elements of \(\mathcal{S}_P\) we get the statistical distance
\begin{align}
    \Delta(P, P') &= \frac{1}{2} \sum_{v \in \mathsf{S}_P} \left| \prod_{i=1}^6 p_{v_i}^{(i)} - \prod_{i=1}^6 p_{v_i}^{(i)'} \right| \notag\\ 
&\leq \frac{1}{2} \sum_{v \in \mathsf{S}_P} \sum_{j=1}^6 |p_{v_j}^{(j)} - p_{v_j}^{(j)'}| \notag\\
& = \frac{1}{2} \sum_{j=1}^6 \sum_{v \in \mathsf{S}_P} |p_{v_j}^{(j)} - p_{v_j}^{(j)'}| \quad (\text{as both sums are finite}) \notag\\
& = \sum_{j=1}^6 \Delta(p^{(j)}, p^{(j)'}).
\end{align}
Given for each \(j=1,\ldots,6\), \(\Delta(p^{(j)}, p^{(j)'}) \leq \frac{\epsilon}{12},\) for some \(\epsilon>0\), we have 
\begin{equation}
\Delta(P, P')=\Delta\left(\left\{\prod_{i=1}^6 p_{v_i}^{(i)}\right\}_{v \in \mathcal{S}_P}, \left\{\prod_{i=1}^6 p_{v_i}^{(i)'}\right\}_{v \in \mathcal{S}_P}\right) \leq 6 \left(\frac{\epsilon}{12}\right)=\frac{\epsilon}{2}.
\end{equation}

\medskip

\noindent This is the statistical distance between the classical sharing of keys \(k,k'\). Therefore, invoking \(\frac{\epsilon}{12}\)-statistical privacy of \(\mathsf{SS}\) and Equation \ref{eq: probdist}, we get

\begin{equation}
D(\Psi_{1}, \Tr_{\overline{P}}(\mathsf{QASS}(M_{f}^{(1)}))) \leq \frac{\epsilon}{2}.
\end{equation}
\noindent In a similar way we can prove that 
\begin{equation}
D(\Psi_{2}, \Tr_{\overline{P}}(\mathsf{QASS}(M_{f}^{(1)'}))) \leq \frac{\epsilon}{2}.
\end{equation}
\noindent Therefore,

\begin{align}
    & D(\Tr_{\overline{P}}(\mathsf{QASS.Share}(M_{f}^{(1)})),\Tr_{\overline{P}}(\mathsf{QASS.Share}(M_{f}^{(1)'}))) \notag\\
    & \leq D(\Psi_{1}, \Tr_{\overline{P}}(\mathsf{QASS}(M_{f}^{(1)}))) + D(\Psi_{1}, \Psi_{2}) + D(\Psi_{2}, \Tr_{\overline{P}}(\mathsf{QASS}(M_{f}^{(1)'}))) \notag\\
    &\leq \frac{\epsilon}{2} +0+\frac{\epsilon}{2} \notag\\
    &=\epsilon
\end{align}

\noindent Since \(\epsilon > 0\) is arbitrary, \( D(\Tr_{\overline{P}}(\mathsf{QASS.Share}(M_{f}^{(1)})),\Tr_{\overline{P}}(\mathsf{QASS.Share}(M_{f}^{(1)'})))=0\). Hence, the quantum anamorphic secret sharing scheme has perfect privacy. 

\medskip

Following the proof idea from Theorem 9., page 22, in the paper \cite{ccakan2023computational}, we can prove the computational privacy similarly by replacing the trace distance with the quantum advantage pseudometric \(\mathsf{Adv}_{\mathcal{F}}(\rho,\sigma)\) \ref{def:qadv}, and properties in Lemma \ref{prop:Adv}. 
\end{proof}

It is important to note that this compiler can be extended to accommodate states of any arbitrary dimension \cite{ccakan2023computational}. It is essential to assume that \( f' \), and consequently \( f \), adheres to the no-cloning property \(f(P)=1 \implies f(\overline{P})=0\), so that an appropriate \(\mathsf{QECC}\) can be applied \cite{ccakan2023computational}. As in the paper \cite{ccakan2023computational} \c{C}akan et al. noted that we can take any quantum erasure-correcting code \(\mathsf{QC}\) realizing any monotone function \(f'\) with \(f' \geq f\), and this means even when we do not know efficient \(\mathsf{QECC}\)s for \(f\), we may take some \(f' \geq f\) to use an efficient \(\mathsf{QECC}\).

\medskip

Let \(T^t_n\) be the \(t\)-out-of-\(n\) threshold function such that \(T^t_n(P)=1\) iff \(\mid P\mid \geq t\), then \(T^t_n \geq f\). As we have seen above choosing a correct \(f'\) is important and here for the choice of threshold function, we may take \(f'=T^t_n\) \cite{ccakan2023computational}.

\medskip

The following results are from the paper \cite{ccakan2023computational} from Section 5. and Section 6. These results are applicable in our work and the share sizes can be computed based on these results. You have included them to mention the existing works on existence of post quantum computational classical secret-sharing scheme realizing \(f\) and construction of \(f'\).

\medskip
We use the compiler we have constructed to design efficient computational quantum anamorphic secret-sharing schemes. The following lemma is due to Yao \cite{Yao89} and Cleve, Gottesman, and Lo \cite{cleve1999quantum}.

\begin{lemma}(\cite{Yao89,cleve1999quantum})
If \( f \) belongs to \(\mathsf{monotone}\) \(\mathsf{P}\), then an efficient post-quantum computational classical secret-sharing scheme realizing \( f \) can be constructed, assuming the existence of post-quantum secure one-way functions.
\end{lemma}

\begin{lemma}(Quantum Shamir secret sharing \cite{cleve1999share})
    For any \( t > n/2 \), there exists an efficient perfect quantum secret sharing scheme that realizes \( T^t_n \) with a share size of \( O(n \log n) \).
\end{lemma}

The following theorem is due to \c{C}akan et al..

\begin{theorem}(\cite{ccakan2023computational})
  If \( f \) is a heavily monotone function within \(\mathsf{monotone}\) \(\mathsf{P}\), then an efficient computational quantum secret-sharing scheme can be constructed to realize \( f \), assuming the existence of post-quantum secure one-way functions.
\end{theorem}

In the proof [Theorem 11., Page 25., \cite{ccakan2023computational}] for \(t=\lfloor\frac{n}{2}\rfloor+1\), \(f'=T^t_n\) is chosen for which \(T^t_n \geq f\), when \(f\) is heavy.

\medskip

\begin{definition}(\cite{ccakan2023computational})\label{def:mNP}
A monotone function \( f \) is said to be in \(\mathsf{mNP}\) if the corresponding language \( \mathcal{L} = \{x \in \{0,1\}^n \mid f(x) = 1\} \) is in \(\mathsf{NP}\).
\end{definition}

Let \( n \in \mathbb{Z}^{+} \) be a positive integer representing the number of players. Define \( \mathcal{S} \) as the Hilbert space corresponding to the secret space. Let \( \mathcal{H}_{1}, \ldots, \mathcal{H}_{n} \) be Hilbert spaces representing the share spaces for the \( n \) players. Consider \( f\colon\{0,1\}^{n} \longrightarrow \{0,1\}\) as a no-cloning monotone function that characterizes a language \( \mathcal{L} \in \mathsf{mNP} \), and let \( \mathcal{Y} \) be a polynomial-time verifier for \( \mathcal{L} \).

Anamologous to the definition of the quantum secret-sharing(\(\mathsf{QSS}\)) in \(\mathsf{mNP}\) from [Section 6, Page 26., \cite{ccakan2023computational}], we can define the \(\mathsf{QASS}\) too.

\begin{definition}(\cite{ccakan2023computational})
A quantum anamorphic secret-sharing scheme (\(\mathsf{QASS}\)) for \(\mathsf{mNP}\) that realizes an access function \( f \) is defined as a set of trace-preserving quantum operations:
\[
\Sigma^{\mathsf{mNP}}_{\mathsf{QASS}} = (\mathsf{QASS.Share}, (\mathsf{QASS.Rec}_{P\subseteq [n]}^{\mathsf{AM}}))
\]
that satisfy the following conditions for all subsets \( P \subseteq [n] \):  

\noindent \(\bullet\) \textbf{Correctness}: If \( f(P) = 1 \), then for any valid witness \( w \) such that \( \mathcal{Y}(P, w) = 1 \), the reconstruction process ensures that if \( \rho_P \) represents the shares held by the subset \( P \), then  
\[
  \mathsf{Rec}_P(\rho_P, P, w) = |\psi\rangle.
\]

\noindent \(\bullet\) \textbf{Privacy}: If \( f(P) = 0 \), then for any quantum states \( |\psi_1\rangle, |\psi_2\rangle \in \mathcal{S} \) and for any quantum polynomial-time (QPT) adversary \( \{C_\lambda\}_\lambda \), the following holds:
\[
  |\Pr[C(\Tr_{\overline{P}} (\mathsf{Share}(|\psi_1\rangle\langle\psi_1|; 1^\lambda))) = 1] - \Pr[C(\Tr_{\overline{P}} (\mathsf{Share}(|\psi_2\rangle\langle\psi_2|; 1^\lambda))) = 1] | \leq \mathsf{negl}(\lambda).
\]
\end{definition}

\medskip

\begin{lemma}(\cite{KNY14}) 
    If \(f\in \mathsf{mNP}\), there is an efficient post-quantum computational classical secret-sharing scheme realizing \(f\) based on the existence of post quantum secure witness encryption for \(\mathsf{NP}\) and one-way functions.
\end{lemma}

The following theorem is proved by \c{C}akan et al. [Section 6, Page 26, Theorem 12] proving the existence of a \(\mathsf{QSS}\) for every heavy function in \(\mathsf{mNP}\).

\begin{theorem}\cite{ccakan2023computational}
    For any heavy function \( f\colon \{0,1\} \longrightarrow \{0,1\} \) belonging to \( \mathsf{mNP} \), there exists a computational \( \mathsf{QSS} \) that realizes \( f \) with \( \mathsf{size}(\mathsf{QSS})\) bounded above by \(\mathsf{poly}(n) \). This construction relies on the existence of post-quantum secure witness encryption for \( \mathsf{NP} \) and the existence of one-way functions.
\end{theorem}
 Consequently, the existence of \(\Sigma^{\mathsf{mNP}}_{\mathsf{QASS}}\) can be shown easily by exactly following and extending the compiler, as we have constructed in Theorem \ref{thm:ass} with similar \(f'=T^{\lfloor \frac{n}{2}\rfloor}+1\) [Theorem 12., Page 26, \cite{ccakan2023computational}]. 

\if 0
In the paper \cite{applebaum2021upslices}, Applebaum and Nir proved the existence of a perfect classical secret-sharing scheme realizing \(f\) with a share size \(1.5^{n+o(n)}\) and proved the following theorem 

\begin{theorem}(\cite{applebaum2021upslices})
    For any monotone function \( f: \{0,1\}^{n} \longrightarrow \{0,1\} \), there exists a perfect classical secret-sharing scheme that implements \( f \) with share size \( 1.5^{n+o(n)}\).
\end{theorem}

\smallskip
For perfect quantum secret-sharing \c{C}akan et al. \cite{ccakan2023computational} established the following theorem:

\begin{theorem}(\cite{ccakan2023computational})
    If \(f\colon\{0,1\} \longrightarrow \{0,1\}\) is a heavy monotone function, then there exists a perfect quantum secret-sharing scheme \(f\) with a total share size of \(1.5^{n+o(n)}\) classical bits and \(O(n \log n)\) qubits.
\end{theorem}

Similarly, for a heavy monotone function \(f\) it is easy to prove the existence of a perfect quantum anamorphic secret-sharing scheme\(\mathsf{QASS}\) with total share size  

\fi

\medskip
\section{Discussions}
\label{discussion}
\subsection{Qubit Requirements and Entropy Computations}\label{Analysis}
\noindent \(\bullet\) \textbf{Total number of qubits:} The anamorphic ciphertext \( M_f^{(1)} \) and the original ciphertext \(M_{f}^{(o)}\), both resides in \(\mathcal{D}((\mathbb{C}^2)^{\otimes (d_1 + 1)}) \). Therefore the anamorphic encryption requires \((d_1 + 1)\) qubits. Processing \( M_f^{(1)} \) also requires \( (d_1 + 1) \) qubits. Extracted messages \( M_o \) and \( M_c \) use \( d_1 \) and \( d_2 \) qubits, respectively, but the decryption circuit operates on \((d_1 + 1)\) qubits. 

\medskip

\noindent \(\bullet\) \textbf{Mutual informations, von-Neumann entropy and relative entropy:} We have defined the original ciphertext \(
M_f^{(0)} = U_{\sigma_l} M_a^{(0)} U_{\sigma_l}^\dagger,
\) where 
\[
M_a^{(0)} = 
\begin{pmatrix} \tfrac{1}{2} M_o' & 0 \\ 0 & \tfrac{1}{2} M_o' \end{pmatrix}.
\]
Because \(\sigma_l\) is unitary, the spectrum of \(M_{f}^{(0)}\) is unchanged. 

\smallskip

let \(\{\lambda_{i}\}_{i=1}^{2^{d_{1}}}\) be the eigenvalues of the original quantum density matrix \(M_{o}\) and since \(M_o'\) is obtained by using \(\mathsf{QOTP}\), the eigenvalues of the \(M_o'\) is same as the eigenvalues of \(M_o\).

\smallskip

Thus the von Neumann entropy of \(M_{f}^{(0)}\) is
\[
\begin{aligned}
\mathsf{S}(M_f^{(0)}) 
  &= -\sum_{i=1}^{2^{d_{1}}}\Bigl[2\Bigl(\frac{1}{2}\lambda_i\Bigr)\,\log\Bigl(\frac{1}{2}\lambda_i\Bigr)\Bigr] \\
  &= -\sum_{i=1}^{2^{d_{1}}} \lambda_i \Bigl[\log(\lambda_i) - \log 2\Bigr] \\
  &= -\sum_{i=1}^{2^{d_{1}}} \lambda_i\log\lambda_i \;+\; \sum_{i=1}^{2^{d_{1}}} \lambda_i \\
  &= \mathsf{S}(M_o') + 1, \qquad \qquad \qquad(\text{since} \Tr(M_{o}) = 1).
\end{aligned}
\]
  
Therefore, \(\mathsf{S}(M_f^{(0)})=\mathsf{S}(M_o') + 1\) 

\smallskip

We compute \(\mathsf{S}(M_{f}^{(1)})\) for a particular case, assuming that both \(M_o'\) and \(M_c''\) are are simultaneously diagonalizable. Let \(\{\mu_{i}\}_{i=1}^{2^{d_{1}}}\) be the eigenvalues of the embedded covert quantum density matrix \(M_{c}''\). 

Thus for each \(i=1,\dots,2^{d_{1}+1}\), we get pair of eigenvalues of \(M_f^{(1)}\) as
\(
\Bigl\{\frac{1}{2}\lambda_i \pm \frac{1}{\eta}\,\mu_i \Bigr\},\) since \(M_{o}\) is strictly positive-definite.

Then,
\[
\begin{aligned}
\mathsf{S}(M_f^{(1)}) 
  = -\sum_{i=1}^{2^{d_{1}+1}} \Biggl[ \Bigl(\frac{1}{2}\lambda_i + \frac{1}{\eta}\,\mu_i\Bigr)
    \log\Bigl(\frac{1}{2}\lambda_i + \frac{1}{\eta}\,\mu_i\Bigr) + \Bigl(\frac{1}{2}\lambda_i - \frac{1}{\eta}\,\mu_i\Bigr)
    \log\Bigl(\frac{1}{2}\lambda_i - \frac{1}{\eta}\,\mu_i\Bigr)
    \Biggr].
\end{aligned}
\]
As \(\mathsf{QOTP}\) encryption is information-theoretically secure. In other words, if the keys are unknown then the ciphertext reveals no information about the underlying plaintext. In our construction the states \(M_o'\) and \(M_c'\) are encrypted by independent random \(\mathsf{QOTP}\) keys, so that without knowledge of the key one has  
\[
I(M_o; M_f^{(1)}) = 0\quad\text{and}\quad I(M_c; M_f^{(1)}) = 0.
\]

We now compute the quantum relative entropy for a particular case, under the same assumption that \(M_{o}'\) and \(M_{c}''\) are simultaneously diagonalizable
\[
\begin{aligned}
S(M_f^{(1)}\|M_f^{(0)}) 
  &= \sum_{i=1}^{2^{d_{1}+1}}\Biggl[ \Bigl(\frac{1}{2}\lambda_i + \frac{1}{\eta}\,\mu_i\Bigr)
    \log\frac{\frac{1}{2}\lambda_i + \frac{1}{\eta}\,\mu_i}{\frac{1}{2}\lambda_i} + \Bigl(\frac{1}{2}\lambda_i - \frac{1}{\eta}\,\mu_i\Bigr)
    \log\frac{\frac{1}{2}\lambda_i - \frac{1}{\eta}\,\mu_i}{\frac{1}{2}\lambda_i}\Biggr] \\
  &= \sum_{i=1}^{2^{d_{1}+1}} \Biggl[
    \Bigl(\frac{1}{2}\lambda_i + \frac{1}{\eta}\,\mu_i\Bigr)
    \log\Bigl(1 + \frac{2\,\mu_i}{\eta\,\lambda_i}\Bigr)
    + \Bigl(\frac{1}{2}\lambda_i - \frac{1}{\eta}\,\mu_i\Bigr)
    \log\Bigl(1 - \frac{2\,\mu_i}{\eta\,\lambda_i}\Bigr)
    \Biggr].
\end{aligned}
\]

Let \(x_i = \frac{2\mu_i}{\eta\,\lambda_i}\), then 

\begin{align}
S(M_f^{(1)}\|M_f^{(0)}) 
  &= \sum_{i=1}^{2^{d_{1}+1}} \frac{1}{2}\lambda_i\Bigl[(1+x_i)\log_{2}(1+x_i)
    +(1-x_i)\log_{2}(1-x_i)\Bigr].
\end{align}
Define the function
\[
f(x) = \frac{1}{2}\Bigl[(1+x)\log_{2}(1+x) + (1-x)\log_{2}(1-x)\Bigr],
\]
so that
\[
S(M_f^{(1)}\|M_f^{(0)}) = \sum_{i=1}^{2^{d_{1}+1}} \lambda_i\, f(x_i).
\]
For $|x|<1$ and using the Taylor series expansion for the base-$2$ logarithm, we get, 
\begin{equation}
      f(x)=\frac{x^2}{2\ln2}+\frac{x^4}{12\ln2}+O(x^6).
\end{equation}
The function $f(x)$ is \emph{even} and \emph{convex} in \((-1,1)\) and easy to verify that \(f(x) \le x^2\quad \text{for } |x|\le 1.\)

\medskip
\noindent Thus, for each \(i\) we have
\[
\lambda_i\, f(x_i) \le \lambda_i\, x_i^2
  = \lambda_i\,\left(\frac{2\mu_i}{\eta\,\lambda_i}\right)^2
  = \frac{4\,\mu_i^2}{\eta^2\,\lambda_i}.
\]
Thus we get
\[
S(M_f^{(1)}\|M_f^{(0)}) \le \frac{4}{\eta^2}\sum_{i=1}^{2^{d_{1}+1}}\frac{\mu_i^2}{\lambda_i}.
\]

As, \(\frac{1}{\eta} < \mathsf{negl}(\lambda)\), this bound indicates that in our construction the anamorphic ciphertext \(M_{f}^{(1)}\) and the original ciphertext \(M_{f}^{(0)}\) are indistinguishable.

\medskip

\subsection{Possible Attacks}\label{sec:attack}
In this section, we present two possible attacks by the dictator and discuss how they can be prevented.

\noindent \textbf{Case-I:}
The adverary or the dictator is authorized to have the origianl key shares but not authorized to have covert key shares. As we proved the perfect privacy of both \(\mathsf{SS}\) and \(\mathsf{ASS}\), the dictator cannot reconstruct the covert key shares but will be able to reconstruct the original key shares.

\medskip

\noindent \textbf{Case-II:} If \textit{\(\mathcal{D}\) wants to enter as an extra player in the set of players to receive shares}. If the dictator joins the set of players as an additional player to obtain a share and subsequently demands that the authorized set of players submit their shares for message reconstruction, the authorized set of players strategically \textit{partially cheats} the dictator. Specifically, while they provide the dictator with the correct key shares necessary to reconstruct the original message, they simultaneously submit \textit{forged shares} for the other covert keys, thereby ensuring that the dictator is unable to access unauthorized information.

We now compute the Partial cheating probability
\(\mathsf{Cheat}^{(p)}\left(V_{i_1}^{(a)}, \ldots, V_{i_t}^{(a)}\right).
\)
In our construction, the anamorphic key is
\(
k=(k_1,k_4,k_2,k_5,k_3,k_6),
\)
with the original key part
\(
\mathsf{k}^{(o)}=(k_1,k_2,k_3)
\)
and the covert key part
\(
\mathsf{k}^{(c)}=(k_4,k_5,k_6).
\)
The key is chosen uniformly at random from
\(
\{0,1\}^{2d_1}\times\{0,1\}^{2d_2}\times\{1,\dots,2^{d_1+1}\}\times\{1,\dots,2^{d_1+1}\}\times\operatorname{Sym}(2^{d_1+1})\times\mathcal{J}.
\)
Thus, the covert key is uniformly distributed over
\(
\mathcal{S}^{(c)} = \{0,1\}^{2d_2}\times\{1,\dots,2^{d_1+1}\}\times\mathcal{J}
\)
with 
\(\mid \mathcal{S}^{(c)}\mid = 2^{2d_2}\cdot 2^{d_1+1}\cdot |\mathcal{J}| \;.
\)

\medskip

\noindent Suppose the honest covert shares of the \(t\) players (when they are honest) are
\(
b^{(c)} = \bigl(s^{(c)}_{i_1}, s^{(c)}_{i_2}, \dots, s^{(c)}_{i_t}\bigr),
\)
and the dictator’s reconstruction function \(\mathsf{Sec}^{(c)}\) then returns
\(
\mathsf{Sec}^{(c)}(b^{(c)}) = s^{(c)},
\)
which is the (correct) covert secret.

\smallskip
Now, assume that a coalition of cheaters wishes to \textit{partially cheat} by forging their covert shares. That is, they replace \(b^{(c)}\) by some \(b^{(c)'}\) with
\(
b^{(c)'} \neq b^{(c)}
\)
(with at least one coordinate changed) while leaving the original part untouched (so that the overall reconstructed secret is
\(
\mathsf{Sec}^{(p)}(b')=(\mathsf{Sec}^{(o)}(b),\mathsf{Sec}^{(c)}(b'))
\)
with \(\mathsf{Sec}^{(o)}(b)= s^{(o)}\) as before).
In our construction, since the key (and hence the covert secret) is chosen uniformly at random from \(\mathcal{S}^{(c)}\) and the reconstruction function \(\mathsf{Sec}^{(c)}\) is deterministic and surjective onto \(\mathcal{S}^{(c)}\), any \textit{forged} share tuple \(b^{(c)'}\) will, in effect, cause the dictator to compute a covert secret that is uniformly distributed over \(\mathcal{S}^{(c)}\).
Since the forged reconstruction \(\mathsf{Sec}^{(c)}(b^{(c)'})\) is uniform over \(\mathcal{S}^{(c)}\), the probability that it accidentally equals the honest secret \(s^{(c)}\) is
\[
\Pr\Bigl(\mathsf{Sec}^{(c)}(b^{(c)'}) = s^{(c)}\Bigr)
=\frac{1}{|\mathcal{S}^{(c)}|}.
\]
Hence, the probability that the dictator reconstructs a covert secret different from \(s^{(c)}\) (i.e. that the cheating is successful) is
\[
\Pr\Bigl(\mathsf{Sec}^{(c)}(b^{(c)'})\neq s^{(c)}\Bigr)
=\Big(1-\frac{1}{|\mathcal{S}^{(c)}|}\Big).
\]
Since the cheating probability is defined as the maximum (over all possible true covert shares \(b^{(c)}\) and over all feasible forged choices \(b^{(c)'}\)) of the above probability, we have
\begin{align}
\mathsf{Cheat}^{(p)} &\le \Big(1-\frac{1}{|\mathcal{S}^{(c)}|}\Big).\\
& = \Big(1-\frac{1}{2^{2d_2+d_1+1}\cdot |\mathcal{J}|}\,\Big), \label{eq:parchit}
\end{align}
which ensures the partial cheating probability is very high, and therefore the dictator cannot get the covert key shares with a very high probability. In fact, with optimal forging the players can achieve exactly this probability and the maximum cheating probability is equal to 
\begin{equation}
    \Big(1-\frac{1}{2^{2d_2+d_1+1}\cdot |\mathcal{J}|}\,\Big).
\end{equation}

\section{Conclusion}\label{conclude}  

In this paper, we have constructed a quantum symmetric-key anamorphic encryption scheme and an anamorphic secret-sharing scheme. For future work, we aim to explore the following problems:  

\smallskip  
\textbf{Question 1}: Construct a quantum anamorphic public-key encryption (\(\mathsf{QAPKE}\)) scheme.  

\smallskip  
\textbf{Question 2}: Develop quantum anamorphic secret-sharing using pseudorandom function-like state generators (\(\mathsf{PRFS}\)) and optimize the share size in the case of anamorphic secret-sharing.

\smallskip
\noindent The \(\mathsf{DCM}\) protocol presented here establishes a baseline for anamorphic quantum communications. For future research works in this direction, we want to ask the following questions:

\begin{enumerate}
    \item \textbf{Minimax-Optimal Sample Complexity.}
    Establishing tight sample complexity bounds remains a primary open problem. Specifically, given the covert message space \(\mathcal{H}_{M_c}=(\mathbb{C}^2)^{\otimes d_2}\) and a reconstruction fidelity requirement \(\mathbb{E}[\|\widehat{M}_c - M_c\|_1] \le \varepsilon\), the determination of
    \[
        N^*(\varepsilon,d_2,\delta) \;=\; \inf_{\widehat{M}_c} \sup_{M_c \in \mathcal{D}(\mathcal{H}_{M_c})} \left\{ N : \Pr(\|\widehat{M}_c - M_c\|_1 > \varepsilon) \le \delta \right\}
    \]
    is necessary to quantify the efficiency of our protocol. While our Pauli tomography approach scales as \(O(4^{d_2})\), recent advances in \emph{classical shadows} \cite{Huang2020} and \emph{low-rank matrix recovery} \cite{Gross2011} suggest that \(N \sim O(\mathrm{poly}(d_2)\varepsilon^{-2}\log(1/\delta))\) is attainable when \(M_c\) possesses low-rank structure. Adapting these estimators to the specific block-embedding structure of anamorphic ciphertexts can be investigated further.

    \item \textbf{Adaptive Measurement Protocols.}
    Our current protocol utilizes a fixed, non-adaptive measurement schedule (random Pauli bases). Developing \emph{adaptive} protocols, where the basis choice for the \(k\)-th copy depends on the outcomes \(m_{1},\dots,m_{k-1}\), could significantly reduce variance \cite{Huszar2012}. Formalizing this as a sequential decision problem and deriving information-theoretic lower bounds via Fano's inequality or Assouad's lemma \cite{Yu1997} would clarify the fundamental limits of anamorphic extraction.

    \item \textbf{Rigorous Noise Analysis.}
    To construct a fully rigorous model using continuous-time master equations or Kraus operator formalisms. For instance, the ideal Hadamard gate \(H_R\) should be modeled as a noisy quantum channel \(\mathcal{E}_H\). A standard model is the depolarizing channel
    \[
        \mathcal{E}_H(\rho) \;=\; (1-p)\, U_H \rho U_H^\dagger \;+\; p \frac{I}{2},
    \]
    where \(p \in [0,1]\) is the noise strength. For future work, we want to derive concentration inequalities for \(\|\widehat{B}-B\|_1\) that depend explicitly on such noise parameters and, in the case of non-Markovian noise, the mixing time of the error process \cite{Temme2010}.

    \item \textbf{Information-Theoretic Security via Quantum Privacy.}
    Establishing unconditional security requires quantifying the mutual information \(I(M_c : E)\) between the covert message and an eavesdropper's system \(E\) (which includes the public ciphertext \(M_f^{(1)}\)). We want to ask that can \(I(M_c:E)\le \delta\) be proven for negligible \(\delta\) using tools from \emph{quantum differential privacy} \cite{Aaronson2019} and approximate randomization lemmas \cite{Hayden2004}.

    \item \textbf{Continuous-Variable Extensions.}
    Is it possible to extend quantum anamorphic encryption framework to continuous-variable (CV) systems where \(\mathcal{H}_{M_c}=L^2(\mathbb{R}^{d_2})\) is relevant \cite{Weedbrook2012}?
\end{enumerate}

\begin{comment}
\section*{Acknowledgments}
The authors are grateful to Professor Bimal Roy and Professor Goutam Mukherjee for their invaluable guidance, insightful discussions, and generous support throughout the preparation of this manuscript. We are thankful to TCG CREST and RKMVERI, Belur, for generous research support.
\end{comment}

\section{Appendix}

\noindent  We have included useful results and properties in this section for use throughout our paper. Most of these results are taken from \cite{nielsen2001quantum} and from the Appendix section of the paper \cite{ccakan2023computational}.

\begin{lemma}(Trace distance \cite{nielsen2001quantum, ccakan2023computational}) \label{prop:trace-dist}

\noindent i. For a probability distribution \(\{p_{i}\}_{i \in I}\) and an ensembles of states \(\{\rho_{i}\}_{i \in I}\)
\[
D(\sum_{i\in I} p_{i}\rho_{i}, \sigma) \leq \sum_{i \in I} p_{i} D(\rho_{i},\sigma_{i}).
\]

\noindent ii. For any trace-preserving quantum operation \(\mathcal{E}\), 
\[
D(\mathcal{E}(\rho),\mathcal{E}(\sigma)) \leq D(\rho,\sigma).
\]

\noindent iii. Let \(AB\) be a composite system and The states are assumed to be of \(AB\),
\[
D(\rho^{A},\sigma^{A}) \leq D(\rho^{AB},\sigma^{AB}).
\]

\noindent iv. Given two density matrices \(\sigma\) and \(\rho\), and for any state \(\tau\),
\[
D(\rho \otimes \tau, \sigma \otimes \tau)=D(\rho, \sigma). 
\]

\noindent v. For any two probability distributions \(\{p_{i}\}_{i \in I}\), \(\{p'_{i}\}_{i \in I}\) and ensembles of states \(\{\rho_i\}_{i \in I}\), \(\{\sigma_i\}_{i \in I}\).
\[
D\left( \sum_{i \in I} p_{i}\rho_{i},\sum_{i \in I}p'_{i}\sigma_{i}\right) \leq \Delta (p_{i},p'_{i}) + \sum_{i \in I} p_{i}D(\rho_{i},\sigma_{i}).
\]

\noindent vi. For any two states \(\tau_{1},\tau_{2}\),
\[
D(\rho \otimes \tau_{1},\sigma\otimes \tau_{2}) \leq D(\rho,\sigma) + D(\tau_{1},\tau_{2}).
\]
\end{lemma}

\begin{proof}
    For the proofs see Section 9.2.1 of \cite{nielsen2001quantum} and Appendix A of \cite{ccakan2023computational}.
\end{proof}

\medskip

\begin{lemma}(Properties of Adversarial Advantage Pseudometric \cite{nielsen2001quantum, ccakan2023computational})\label{prop:Adv}
    
\noindent For any circuit family \( \mathcal{F} \) and two states \( \rho, \sigma \), the following holds:  

\noindent i. \( A_\mathcal{F}(\rho, \rho) = 0 \);  

\noindent ii.  \( A_\mathcal{F}(\rho, \sigma) = A_\mathcal{F}(\sigma, \rho) \); 
    
\noindent iii.  For any state \( \tau \),  
   \[
   A_\mathcal{F}(\rho, \sigma) \leq A_\mathcal{F}(\rho, \tau) + A_\mathcal{F}(\tau, \sigma);
   \]  

\noindent iv. Assuming the states are of a composite system \( AB \),  
   \[
   A_\mathcal{F}(\rho_A \otimes |0\rangle\langle 0|, \sigma_A \otimes |0\rangle\langle 0|) \leq A_\mathcal{F}(\rho_{AB}, \sigma_{AB});
   \]  

\noindent v. For any two probability distributions, \( \{p_i\}_{i \in I} \), \( \{q_i\}_{i \in I} \), and ensembles of states \( \{\rho_i\}_{i \in I} \), \( \{\sigma_i\}_{i \in I} \),  
   \[
   A_\mathcal{F}\left(\sum_{i \in I} p_i \rho_i, \sum_{i \in I} q_i \sigma_i\right) \leq \Delta(p_i, q_i) + \sum_{i \in I} p_i A_\mathcal{F}(\rho_i, \sigma_i);
   \]  

\noindent vi. For a probability distribution \( \{p_i\}_{i \in I} \) and an ensemble of states \( \{\rho_i\}_{i \in I} \),  
   \[
   A_\mathcal
   {F}\left(\sum_{i \in I} p_i \rho_i, \sigma\right) \leq \sum_{i \in I} p_i A_\mathcal{F}(\rho_i, \sigma);
   \]  

\noindent vii. For any family \( \mathcal{F}' \) and state \( \tau \) such that there is \( C' \in \mathcal{F}' \) satisfying \( C'(\rho) = C(\rho \otimes \tau) \) and \( C'(\sigma) = C(\sigma \otimes \tau) \) for any \( C \in \mathcal{F} \),  
   \[
   A_\mathcal{F}(\rho \otimes \tau, \sigma \otimes \tau) \leq A_{\mathcal{F}'}(\rho, \sigma).
   \]  
\end{lemma} 

\begin{proof}
    For the proofs see Appendix A of \cite{ccakan2023computational}.
\end{proof}

\bibliographystyle{plain}
\bibliography{references}

\end{document}